\documentclass[twoside,11pt]{article}

\usepackage{jmlr2e}
\usepackage{booktabs}
\usepackage{makecell}
\usepackage{hyperref}
\usepackage{makeidx}
\usepackage{amsmath}
\usepackage{graphicx}
\usepackage{color}
\usepackage{verbatim}
 \usepackage{algorithm}
\usepackage{accents}
\usepackage{multirow}
\usepackage{bigstrut}
\usepackage{tabularx}
\usepackage{rotating}
\usepackage[flushleft]{threeparttable}
\usepackage[misc]{ifsym}
\usepackage{algorithmic}
\usepackage{ulem}
\usepackage{xcolor}

\definecolor{blue}{rgb}{0,0,0.9}
\definecolor{red}{rgb}{0.9,0,0}
\definecolor{green}{rgb}{0,0.9,0}

\newcommand{\cX}{{\cal X}}
\newcommand{\cO}{{\cal O}}
\newcommand{\cY}{{\cal Y}}
\newcommand{\cT}{{\cal T}}
\newcommand{\cU}{{\cal U}}
\newcommand{\norm}[1]{\left\lVert#1\right\rVert}
\newcommand{\argmin}{{\textup{argmin}}}
\newcommand{\dist}{{\textup{dist}}}

\makeatletter
\def\widebar{\accentset{{\cc@style\underline{\mskip11mu}}}}
\makeatother

\firstpageno{1}

\begin{document}

\title{MARS: A second-order reduction algorithm for \\ high-dimensional sparse precision matrices estimation\thanks{ {Defeng Sun is supported in part by Hong Kong Research Grant Council under grant number 15303720.} }}

\author{\name Qian Li \email qianxa.li@connect.polyu.hk \\
       \addr Department of Applied Mathematics\\
       The Hong Kong Polytechnic University\\
       Hung Hom, Kowloon, Hong Kong
       \AND
       \name Binyan Jiang \email by.jiang@polyu.edu.hk \\
       \addr Department of Applied Mathematics\\
		The Hong Kong Polytechnic University\\
       Hung Hom, Kowloon, Hong Kong
		\AND
		\name Defeng Sun \email defeng.sun@polyu.edu.hk \\
		\addr Department of Applied Mathematics\\
		The Hong Kong Polytechnic University\\
       Hung Hom, Kowloon, Hong Kong
}


\maketitle

\begin{abstract}
	Estimation of the precision matrix (or inverse covariance matrix) is of great importance in statistical data analysis { and machine learning}. However, as the number of parameters scales quadratically with the dimension $p$, computation becomes very challenging when $p$ is large.
	{In this paper, we propose an adaptive sieving reduction algorithm to generate a solution path for the estimation of precision matrices under the $\ell_1$ penalized D-trace loss, with each subproblem being solved by a second-order algorithm.}
	In each iteration of our algorithm, we are able to greatly reduce the number of variables in the {problem} based on the Karush-Kuhn-Tucker (KKT) conditions and the sparse structure of the estimated precision matrix in the previous iteration.
	{As a result, our algorithm is capable of handling datasets with very high dimensions that may go beyond  the capacity of the existing methods.}
	Moreover, for the sub-problem in each iteration, other than solving the primal problem directly, we develop a semismooth Newton augmented Lagrangian algorithm with global linear convergence rate on the dual problem to improve the efficiency.
	Theoretical properties of our proposed algorithm have been established.
	In particular, we show that the convergence rate of our algorithm is asymptotically superlinear.
	The high efficiency and promising performance of our algorithm are illustrated via extensive simulation studies and real data applications, with comparison to several state-of-the-art solvers.
\end{abstract}


\section{Introduction}

The estimation of high dimensional sparse precision matrices has been a central topic in statistical learning {and machine learning \citep{Li2019respre, Du2020},} with a wide range of applications such as genomics \citep{Wille2004, LiH2006}, image analysis \citep{LiSZ2009}, among others. Owing to the fast development of data engineering and technology, modern datasets are oftentimes having much higher dimensions than before, and the estimation of the precision matrices becomes more challenging as the number of variables scales quadratically {with respect to} the dimension $p$.   For example, in the breast cancer data set studied in our numerical experiments, $p$ is equal to $22,283$, and the number of variables is nearly $250$ million.
Many existing algorithms or solvers could easily fail to produce a meaningful estimator in this case.
Highly efficient algorithms with sound theoretical guarantees are thus in great need of meeting the computation requirement of the time.

{Consider the $p$-dimensional Gaussian distributed random variable $X=(X_1,\cdots,X_p) \sim \mathcal{N}(\mu,\Sigma)$, where $\mu \in \mathbb{R}^p$ and $\Sigma \in \mathbb{R}^{p \times p}$ are the mean vector and the covariance matrix, respectively. Assume that the covariance matrix $\Sigma$ is nonsingular.}
The precision matrix $\Theta$ (also known as the concentration matrix) is defined as the inverse of $\Sigma$, i.e., $\Theta = \Sigma ^{-1}$.
It is well known that the conditional independence {in the Gaussian distribution} is directly reflected in the zero components of the precision matrix \citep[Proposition 5.2]{Lauritzen1996}. Specifically, for any $1\leq i\not=j \leq p$, $\Theta_{ij}=0$ if and only if $X_i$ is conditionally independent of $X_j$ given all other random variables $X_k,$ $ k\not=i,j,\, 1\leq k\leq p$.
{In many applications, the conditional independence structure of the variables is usually represented as a graph $G=(V,E)$, where $V$ denotes the $p$ nodes and the  edge set $E\subseteq V\times V$  denotes the set of conditional dependent pairs of the nodes. Thus,  the graph $G$ can be well recovered  if the zeros in the precision matrix can be consistently identified.}

So far, many methods have been proposed to estimate the sparse precision matrix {in the high-dimensional setting ($p \gg n$, where $n$ corresponds to the sample size). }
\cite{Meinshausen2006} estimated {the conditional independence restrictions separately for each node in $G$} through a sequence of lasso penalized least squares regression models.
\cite{Yuan2010} studied the above regression models via a Dantzig selector.
Later, \cite{Cai2011} proposed a constrained $\ell_1$ minimization approach and established the convergence rates {in statistical analysis} under different norms. However, among the methods mentioned above, none of them are truly treating the precision matrix as a matrix form.
There is another well-known estimator called the lasso penalized Gaussian likelihood estimator \citep{Yuan2007, Banerjee2008, Friedman2008}, also known as graphical lasso or glasso. Given the sample covariance matrix $\widehat{\Sigma} \in \mathbb{S}^{p}$ and a regularization parameter $\lambda > 0$, the glasso estimator is obtained by minimizing the $l_1$ penalized log-likelihood function:
\begin{equation}\label{gl0}
	\mathop{\min}\limits_{\Omega\in \mathbb{S}^{p}_+}\left\{{\rm tr}(\Omega \widehat{\Sigma})-{\rm log det} (\Omega )+\lambda \norm{\Omega}_1\right\},
\end{equation}
where $\mathbb{S}^{p}_+$ is the space of $p \times p$ real symmetric positive definite matrices, ${\rm tr}(\cdot)$ and $\norm{\cdot}_{1}$ are the trace and $\ell_1$-norm, respectively. Researchers have designed different optimization algorithms to solve this problem. Some first-order methods have been applied to solve \eqref{gl0}, such as the block coordinate descent method \citep{Banerjee2008, Friedman2008} and the alternating linearization method \citep{Scheinberg2010}. To solve the graphical lasso problem more efficiently, some second-order methods such as the quadratic approximation method \citep{Hsieh2014} and the Newton-like methods \citep{Oztoprak2012} were also developed.
However, these two methods may not be the best choice. For the quadratic approximation method, the computational complexity could be up to $O(p^3)$ {per-iteration}.
As for the Newton-like methods, the algorithms are more or less heuristic and {the} related convergence properties  are yet to be explored.
We note that the graphical lasso {model} brings {challenges} to the calculation because its objective function contains a log determinant term. Although \cite{Witten2011} provides a strategy to further improve the efficiency by identifying the block diagonal structure, which has been implemented in the glasso package,  such a strategy can easily fail in practice, especially when the regularization parameter is small. The total computation time on the breast cancer dataset shown in Table \ref{tab:bc} can further {support} such {an argument}.

Recently, {an} $\ell_1$-penalized D-trace loss estimator was proposed in \citep{Zhang2014, Liu2015}.
This new estimator is obtained by solving a convex composite optimization problem, which involves a quadratic loss function with an $\ell_1$-regularized penalty:
\begin{equation}\label{eq:l2}
	\mathop{\min}\limits_{\Omega \in \mathbb{S}^{p}}\left\{ \frac{1}{2} {\rm tr}(\Omega \widehat{\Sigma} \Omega^T) - {\rm tr}(\Omega) + \lambda\left\|\Omega\right\|_{1,{\rm off}} \right\},
\end{equation}
where  $\mathbb{S}^{p}$ is the space of $p \times p$ real symmetric matrices and
$\norm{\cdot}_{1,{\rm off}}$ is the off-diagonal $\ell_1$-norm, i.e., $\norm{\Omega}_{1,{\rm off}}=\sum_{i\not= j}|\Omega_{i,j}|$.
{{In this paper}, we will {show that} the solutions {obtained by our proposed algorithm to the optimization problem (\ref{eq:l2}) are asymptotically positive definite}. {This key property makes it valid for us to replace the constraint $\Omega \in \mathbb{S}^p_+$ with a much simpler constraint $\Omega \in \mathbb{S}^p$}}.
\cite{Zhang2014} derived the convergence rates in {statistical analysis} of this new estimator and showed that it could be comparable with the graphical lasso model. Clearly, the loss function in problem \eqref{eq:l2} is simpler than the {one in the} graphical lasso {model}, which could bring great {efficiency} to calculation. In particular, when dealing with the big data where the dimension is very large, computational efficiency owing to the simple form of the loss function would be a favorable feature for practical applications. However, {the} existing {popular} methods for solving (\ref{eq:l2}) are first-order methods, such as the coordinate descent method \citep{Liu2015} and the alternating direction method of multipliers (ADMM) \citep{Zhang2014, Wang2020}. {These first-order methods encounter    challenges  in  solving  (\ref{eq:l2}) even up to a moderate accuracy}.

{In order to design an efficient algorithm, in this paper, we first propose an adaptive sieving reduction strategy to solve the model \eqref{eq:l2} by sequentially solving some reduced problems with much smaller dimensions than the original problem. In addition, we point out that this strategy is powerful for generating a solution path, as a better initial point usually leads to a positive impact. The reduced problem is obtained by exploiting the sparsity of the optimal solution. More precisely, based on the KKT conditions of the original problem, we can construct a non-zero index set with any given initial solution, and then make all components outside this index set zero. This leads to the form of the reduced problem as
\begin{equation}\label{eq:lad}
	\mathop{\min}\limits_{\Omega \in \mathbb{S}^{p}}\left\{ \frac{1}{2} {\rm tr}(\Omega \widehat{\Sigma} \Omega^T) - {\rm tr}(\Omega) + \lambda\left\|\Omega\right\|_{1,{\rm off}} - \langle \Delta, \Omega \rangle \mid \Omega_{\bar I} = 0 \right\},
\end{equation}
where $\Delta \in \mathbb{S}^p$ is an error matrix with {$\norm{\Delta} \le  \epsilon$} for some small $\epsilon > 0$, $\norm{\cdot}$ is the Frobenius norm and $\bar I$ is the non-zero components index set. Note that the existence of $\Delta$ means that the minimization problem is solved inexactly, and $\Delta$ does not need to be given in prior but is automatically obtained when the problem is solved inexactly.}
As can be seen from \eqref{eq:lad}, the dimension of the reduced problem is exactly equal to $(|I|+p)/2$. In general, the dimension of the reduced problem 
may not be larger than $p + n(n-1)/2$ to ensure the {statistical} validity of the estimation in some cases {(recall that $n$ is the sample size)}. Therefore, our algorithm is not only significantly efficient, but also can solve the problem of insufficient storage space for massive data to a certain extent. Then, we design an efficient second-order {based} algorithm, or more precisely, a semismooth Newton {based} augmented Lagrangian algorithm, to solve {the dual problem of} (\ref{eq:lad}). {In this paper, we {focus on} estimating the precision matrices in the high dimension setting (i.e., $p \gg n$). Thus, the scale of the dual problem of (\ref{eq:l3}) is much smaller}.
More importantly, due to {the facts that the} piecewise linear-quadratic structure of the primal problem provides an asymptotically superlinear convergence rate of the augmented Lagrange method \citep{Li2018}, and  the dual problem is strongly convex, our algorithm only needs a few iterations to obtain a desirable solution.
{Besides, the per-iteration computational and memory complexities of our algorithm are comparable to or even better than the ones of first-order algorithms, such as ADMM.}
In addition, we provide a technique to determine the maximum $\lambda$. This technique limits the choice of $\lambda$ to avoid unnecessary waste of time in generating a solution path. In subsequent numerical experiments, we shall see that our algorithm significantly outperforms several state-of-the-art solvers and is competent to handle huge-scale problems.

We highlight the main contributions of this paper as follows:
\begin{itemize}
	\item[1.] We develop a dual approach for  {the precision matrix estimation}. By equivalently rewriting the primal problem under the high-dimensional setting where $p$ is much large than $n$, we obtain a dual problem where the dimension of a dual variable is $p \times n$ instead of $p\times p$. Such an approach can fundamentally improve the efficiency when $n$ is much smaller than $p$.
	
	\item[2.] This is the first attempt to implement the adaptive sieving strategy and the semismooth Newton augmented Lagrangian algorithm for variables with matrix forms.
	{The adaptive sieving strategy enables us to solve the original problem by solving some reduced subproblems, which can be remarkably smaller in dimension than the original problem. Therefore, some time-consuming operations in the main loop can be avoided, for example, the multiplication of two $p \times n$ matrices is changed to the multiplication of two vectors of a much lower dimension.}
	More importantly, although we are adopting a second-order method for solving the subproblems, the {per-iteration} computational complexity of our proposed algorithm is comparable to {the} first-order algorithms, {such as ADMM}. The promising numerical performance of our algorithm is also theoretically justified by the global linear convergence rate and asymptotically superlinear convergence rate we established.
	
	\item[3.] We have developed
	{a R package} for applications to estimate the sparse precision matrix effectively.
	Compared {to} other existing solvers/packages, our algorithm is much more efficient and is able to handle datasets with much higher dimensions.
	For instance, on a publicly available breast cancer data set, our algorithm can be up to more than 20 times faster than the popular glasso package \citep{Friedman2008, Witten2011} for estimating a precision matrix with five-fold cross-validation included.
\end{itemize}

The remaining subsequent arrangements are as follows. {For better discussions in later sections, we will first of all  introduce some notation and present some preliminary results of the piecewise linear quadratic function in Section \ref{sec:pre}.}
In Section \ref{sec:as}, we will develop an adaptive sieving reduction strategy for generating solution paths. In Section \ref{sec:ialmssn}, we derive an inexact augmented Lagrangian method (ALM) to solve the dual problem of the inner problem in Section \ref{sec:as}. Then for the subproblem in the inexact ALM, we design a semismooth Newton algorithm to obtain an expected solution. {In Section \ref{sec:dis}, we will summarize the connections between the proposed algorithms in Sections \ref{sec:as} and \ref{sec:ialmssn}, as well as discuss the computational and memory complexities of the proposed algorithms.}
In Section \ref{sec:num}, after introducing some algorithms, by comparing with the introduced algorithms and several state-of-the-art solvers, we will demonstrate the promising performance of our algorithm through some numerical experiments and the analysis of two real datasets. We conclude our paper in Section \ref{sec:col}.

\medskip
{\noindent{\bf Notation}} Throughout this paper, $\cX$ and $\cY$ represent two finite-dimensional real Euclidean spaces.  We use $\langle \cdot,\cdot \rangle$ to denote the inner product and its induced Frobenius norm by  $\norm{\cdot}_F$. Specifically, let $X = (X_{ij})_{1\leq i \leq p,\, 1 \leq j \leq n}$ and $Y= (Y_{ij})_{1\leq i \leq p,\, 1 \leq j \leq n}$ be two real matrices,  $\langle X,Y\rangle = {\rm tr} \left (X^TY \right) = \sum_{i,j} X_{ij} Y_{ij}$ and $\norm{X}_F=\left(\sum_{i,j}X_{ij}^2 \right)^{1/2}$, where $X^T$ denotes the transpose of $X$. The $\ell_1$ norm is denoted by $\norm{\cdot}_1$, i.e., $\norm{X}_1 = \sum_{ij} |X_{ij}|$, and the spectral norm is denoted by $\norm{\cdot}_{(2)}$. The Jacobian of $F:\cX \to \cY $ at $X\in D_F$ is denoted as $F'(X)$, where $D_F=\left\{X \mid F(\cdot) \text{ is differentiable at } X  \right\}$. We also use ``$\circ$'' to denote the Hadamard product, i.e., $(X\circ Y)_{ij}=X_{ij}Y_{ij}$ and $\sup\,\{\cdot\}$ to denote the supremum. {The cardinal number of a real vector or matrix $V$ is denoted by $|V|$, but for a one-component variable $v$, $|v|$ denotes its absolute value.}
The point-to-set distance is defined by ${\rm dist}(Y,C):={\rm inf}_{Y'\in C}\norm{Y-Y'},\, \forall\,Y\in\cY\ {\rm and}\ \forall\, C\subset\cY$, while when $C$ is empty, it is $+\infty$ by convention.

\section{Preliminaries}\label{sec:pre}

{In this section, we will introduce some notation  for later use. Followed by some technical results on the convex piecewise linear quadratic (PLQ) function and its subdifferential, which will be the key to establishing the global linear convergence rate of the proposed inexact augmented Lagrangian method later.}

Let $f:\cX\to(-\infty,+\infty]$ be a closed and proper convex function.  For any given $X \in \cX$, let $\textup{Prox}_f(X)$ be the unique optimal solution of the Moreau-Yosida regularization of $f$ at $X\in \cX$:
\begin{eqnarray}\nonumber
	\mathcal{H}_f(X):=\mathop{\min}_{Y\in\cX}\left\{f(Y)+\frac{1}{2}\norm{Y-X}^2_F\right\}.
\end{eqnarray}
The mapping $\textup{Prox}_f(\cdot)$
is called the proximal point mapping of  $f$.
In addition, $\textup{Prox}_f(\cdot)$ is globally Lipschitz continuous with modulus 1 \citep{Lemarechal1997}.
For later use, we present some useful properties  of the Moreau-Yosida regularization here. {From \citep[Theorem 2.26]{rockafellar2009variational}, we know that} $\mathcal{H}_f(\cdot)$ is continuously differentiable, and furthermore, the gradient of $\mathcal{H}_f(\cdot)$ at $X\in\cX$ is known as  {in the form of} $\nabla\mathcal{H}_f(X)=X-\textup{Prox}_f(X)$. Another important and useful result is the Moreau decomposition {\citep{moreau1962fonctions}}, which is, any $X\in \cX$ has the decomposition $X=\textup{Prox}_f(X)+\textup{Prox}_{f^*}(X)$, where $f^*$ is the Fenchel conjugate of $f$ and is defined by
\[
f^*(Y)=\mathop{\sup}\left\{\langle Y,X\rangle-f(X)\mid X\in\cX \right\}, \ Y\in\cX.
\]
It can be shown that the pointwise supremum function of a collection of convex (closed) functions is convex (closed). Thus, $f^*$ is always convex and closed.

{Next, we begin our discussion with the definition of the PLQ function. A continuous function $f:D\to \mathbb{R}$ defined on a set $D \subseteq \cX$ is PLQ \citep[Section 10.E]{rockafellar2009variational}, if $D$ can be represented as the union of finitely many polyhedral sets $\{D_i\}_{i=1}^{\bar m}$, relative to each of which $f(x) = q_i(x)$, $x \in D_i$ with $q_i$ being a quadratic function. A useful conclusion on the more generally piecewise quadratic function is provided in \citep[Proposition 2.2.4]{Sun1986}. It indicates that a closed proper convex function $f$ is PLQ on $D$, if and only if the graph of $F_f: D \rightrightarrows \cY$ is polyhedral, where $F_f$ denotes the subdifferential of $f$. Moreover, following the work in \citep{Robinson1981} about polyhedral multifunctions, we discuss some continuous properties of $F_f$ in the remaining of this section.
}

{The following definition is given in \citep*{Robinson1981}. This can be viewed as an extension of the Lipschitz condition of the real-valued function to multifunctions.}

\begin{definition}\label{def:errbd}
	Let $F: \cY\rightrightarrows\cX$ be a multifunction. If there exists $\kappa \geq 0$ such that for some neighborhood $N(\bar y)$ of $\bar y$ and for all $y \in N(\bar y)$, \[F(y) \subset F(\bar y) + \kappa \norm{y - \bar y}B_x,\] with $B_x = \{x \mid \norm{x} \leq 1\}$,
	then $F$ is said to be locally upper Lipschitzian at the point $\bar y$ with modulus $\kappa$.
\end{definition}

{Since the graph of $F_f: D \subseteq \cX \to \cY$ is polyhedral, it then follows from Proposition 1 of \citep*{Robinson1981} that $F_f^{-1}$ is locally upper Lipschitzian at each $\bar y \in F_f(D)$ with the same modulus $\kappa \geq 0$. An interesting consequence of this result is given as follows.}

\begin{proposition}\label{lm:eb}
	{Suppose that} $F_f^{-1}(0)$ is nonempty. There is $\kappa \geq  0$ and some neighborhood $N_0$ of the origin such that, for all $x \in D \subseteq \cX$,
	\begin{equation}\label{eq:geb}
		\dist(x, F_f^{-1}(0)) \leq \kappa\, \dist (0, F_f(x)\cap N_0).
	\end{equation}
\end{proposition}

\begin{proof}
	 The nonemptiness of $F_f^{-1}(0)$ implies that, for any $x \in D$, if $F_f(x) \cap N_0 = \emptyset$, the inequality (\ref{eq:geb}) holds automatically. We then consider the case that $F_f(x) \cap N_0 \neq \emptyset$. By the fact that $F_f^{-1}$ is locally upper Lipschitzian at $0$ with modulus $\kappa$, we know that there are some neighborhood $N_0$ of the origin such that, for any $\tilde y \in N_0$,
	 \[F_f^{-1}(\tilde y) \subset F^{-1}_f(0) + \kappa \norm{\tilde y}B_x.\]
	 Since $F_f(x)$ is convex and closed for any $x \in D$ \citep[Theorem 19.1]{Rockafellar1970}, there exists $y \in F_f(x) \cap N_0$ with
	 \[
		\norm{y} = {\rm dist}(0, F_f(x) \cap N_0).
	 \]
	 Note that $x \in F^{-1}_f(y)$ and $y \in N_0$. Therefore,  \[\dist(x, F_f^{-1}(0)) \leq \kappa\, \dist (0, F_f(x)\cap N_0) .\] This completes the proof.
\end{proof}

\begin{remark}
	From  {Lemma} \ref{lm:eb}, we know that there exists $\epsilon > 0$ such that for any $x \in D$ with ${\rm dist}(0, F_f(x)) < \epsilon$ we have $\dist(x, F_f^{-1}(0)) \leq \kappa\, \dist (0, F_f(x))$, which is consistent with  the corollary introduced  by \cite{Robinson1981}.
\end{remark}

{Following Lemma \ref{lm:eb}, the next lemma shows that a more general inequality holds for any point $x$ arbitrarily chosen on the effective domain of $F_f$.}

\begin{lemma}\label{lm:gc}
	{Suppose that} $F_f^{-1}(0)$ is nonempty.
	For any $r > 0$, there exists $\kappa \geq  0$ such that
	\begin{equation}\label{eq:glr}
		\dist(x, F_f^{-1}(0)) \leq \kappa\, \dist (0, F_f(x)), \quad \forall \, x \in D \textup{ satisfying }\dist (x, F_f^{-1}(0)) \leq r .
	\end{equation}
\end{lemma}

\begin{proof}
	From Lemma \ref{lm:eb}, we know that there exists $\kappa_1 \geq 0$ and some neighborhood $N_0$ of the origin such that for all $x \in D$, the inequality (\ref{eq:geb}) holds. Then, for any $x\in D$ satisfying $\dist (x, F^{-1}_f(0)) \leq r$, if $F_f(x)\cap N_0 \neq \emptyset$, we readily have \[
	 \dist(x, F_f^{-1}(0)) \leq \kappa_1\, \dist (0, F_f(x)\cap N_0) = \kappa_1\, \dist (0, F_f(x)),
	\]
	otherwise, there exists $\bar \delta > 0$ satisfying $\dist (0, F_f(x)) \geq \bar \delta$, such that
	\[
	\dist(x, F_f) \leq \kappa_2\, \dist (0, F_f(x)),
	\]
	where $\kappa_2 \geq r / \bar \delta$. Then, let $\kappa = \max \{\kappa_1, \kappa_2\}$. This completes the proof.
\end{proof}

\section{An adaptive sieving reduction strategy}\label{sec:as}

In this section, based on the work of \citep{Lin2020}, we will develop an adaptive sieving reduction strategy to generate solution paths by solving the problem (\ref{eq:l2}). The main idea of this strategy is to solve problem (\ref{eq:l2}) by solving some reduced problems with remarkably smaller dimensions compared to the original problem (\ref{eq:l2}).
As a result, under sparse and high dimensional settings, this strategy can greatly improve the algorithm efficiency, while also saving a lot of storage space.
Since our algorithm is designed for {\underline M}atrix estimation via  {an} {\underline A}daptive sieving {\underline R}eduction strategy and a {\underline S}emismooth Newton augmented Lagrangian algorithm (in Section \ref{sec:ialmssn}), we call our algorithm {\bf MARS}.

As we mentioned in the introduction, we will develop a dual approach to solve problem \eqref{eq:l2}. To facilitate the designing of the dual approach, we write problem \eqref{eq:l2} equivalently as
\begin{equation}\label{eq:l3}
	\mathop{\min}\limits_{\Omega \in \mathbb{S}^{p}}\left\{ \frac{1}{2} \norm{\Omega A}^2_F- \langle \Omega, I_p \rangle + \lambda\left\|\Omega\right\|_{1,{\rm off}} \right\},
\end{equation}
where $A$ is a  real matrix with rank $n$ such that $AA^T = \widehat{\Sigma}$.
{ Note that instead of applying the singular value decomposition (SVD) on $\widehat{\Sigma}$, the matrix $A$ can be efficiently obtained by applying {a} thin SVD on the $p\times n$ dimensional centered data matrix. The thin SVD requires significantly less space and time than the full SVD, especially in the high-dimensional setting.} Without loss of generality, we assume that $A$ is a $p \times n$   matrix with rank $n$. For later use, we denote $\theta(\Omega) := \norm{\Omega}_{1, \textup{off}},  {\forall \, \Omega} \in \mathbb{S}^p$.
Moreover, we further denote the optimal solution set of (\ref{eq:l3}) by $\Theta_\lambda$,
and the associated proximal residual mapping by
\[ R_\lambda (\Omega) := h(\Omega) + \textup{Prox}_{\delta_{B_\lambda}}(\Omega - h (\Omega)), \quad \forall\, \Omega \in \mathbb{S}^p, \]
where $h(\Omega): = \frac{1}{2} (\Omega \widehat{\Sigma} + \widehat{\Sigma} \Omega) -I_p$ with $I_p$ being the $p$ dimensional identity matrix, and $\delta_{B_\lambda}$ is the indicator function with $B_\lambda = \{Z \in \mathbb{S}^p \mid Z_{ii} = 0, \, |Z_{ij}| \leq \lambda, \,  i, j = 1,\cdots,p, \,  i\neq j \}$, i.e., $\delta_{B_\lambda}(Z) = 0$ for any $Z \in B_{\lambda}$ and $\delta_{B_\lambda}(Z) = +\infty$ otherwise.
By the KKT conditions, we know that $\widetilde{\Omega} \in \Theta_\lambda$ if and only if $R_\lambda (\widetilde{\Omega}) = 0$.

Detailed steps of our adaptive sieving reduction strategy are given in Algorithm \ref{alg:as}. For a sequence of positive regularization parameters sorted in descending order, we first solve problem (\ref{eq:l3}) inexactly with $\lambda$ equal to the largest parameter to obtain {an approximate solution with a given tolerable error $\epsilon \geq 0$} and the {corresponding index set for the} non-zero components. Then, for the subsequent smaller $\lambda$, we continuously use the KKT conditions to perform adaptive sieving to obtain a new non-zero components index set, while updating its solution until a desirable solution is obtained. Such a procedure is performed for all the regularization parameters until the algorithm stops  (we will show that the while loop can terminate in a finite number of iterations in the proof of Theorem \ref{thm:asmain}). Note that the existence of $\Delta_0$ and $\{\Delta_i^l\}$ in Steps \ref{algas: initial} and \ref{algas: main} {in Algorithm \ref{alg:as}} means that the minimization problems are solved inexactly. Both of them are not given in prior but are automatically obtained when the original minimization problems are solved inexactly.


\begin{algorithm}[!h]
	\centering
	\caption{An adaptive sieving reduction strategy for generating a solution path.}
	\label{alg:as}
	\begin{algorithmic}[1]
		\REQUIRE ~~\\ 
		A real matrix $A \in \mathbb{R}^{p \times n}$ and a tolerance constant $\epsilon \geq 0$;\\
		A sequence of regularization parameters $\lambda_0 > \lambda_1 > \cdots > \lambda_k > 0$ with $\lambda_{\max} \geq \lambda_0$;
		\ENSURE ~~\\ 
		A solution path: $\Omega^*(\lambda_0), \Omega^*(\lambda_1), \cdots, \Omega^*(\lambda_k)$;
		\STATE {\bf Initialization: } \\
		For $\lambda_0 > 0$, solve \[
		\Omega^*(\lambda_0) \in \argmin_{\Omega \in S^p} \left\{\frac{1}{2} \norm{\Omega A}^2_F - \langle \Omega , I_p \rangle + \lambda_0 \norm{\Omega}_{1,\text{off}} - \left\langle \Delta_0 , \Omega \right\rangle \right\}, \]
		where $\Delta_0 \in S^p$ is an error matrix such that $\norm{\Delta_0}_F \leq \epsilon $. Then let \[I^*(\lambda_0) := \{(i,j) \mid \Omega^*(\lambda_0)_{ij} \neq 0,\, i,j = 1,\cdots, p\};\]
		\label{algas: initial}
		\STATE {\bf Main loop: } \\
		\FOR{$i=1$; $i<k+1$; $i++$ }
		\STATE Set $\Omega^0(\lambda_i) = \Omega^*(\lambda_{i - 1})$ and $I^0(\lambda_i) = I^*(\lambda_{i - 1})$;
		\STATE Calculate $R_{\lambda_i}(\Omega^0(\lambda_i))$ and set $l = 0$;
		\WHILE{$\norm{R_{\lambda_i}(\Omega^l(\lambda_i))}_F > \epsilon$}
		\STATE $l++$;
		\STATE Create $J^{l}(\lambda_i)$ by \footnotesize{ \[ J^{l}(\lambda_i) = \left\{(i,j) \in \bar I^{l - 1}(\lambda_i) \mid - \left(h(\Omega^{l - 1}(\lambda_i))\right)_{ij} \notin \lambda_i \left( \partial \theta (\Omega^{l - 1}(\lambda_i)) + \frac{\epsilon}{\lambda_i \sqrt{2|\bar I^{l - 1}(\lambda_i)|}} \mathbb{B}\infty \right)_{ij} \right\},\]}
		\label{algas: index}
		\normalsize where $ \bar I^{l - 1}(\lambda_i)$ denotes the complement of $I^{l - 1}(\lambda_i)$ and $\mathbb{B}\infty$ is the infinity norm unit ball;
		\STATE Update $I^{l}(\lambda_i) = I^{l - 1}(\lambda_i) \cup J^{l}(\lambda_i)$;
		\STATE Solve \[ \Omega^l(\lambda_i) \in \argmin_{\Omega \in S^p} \left\{\frac{1}{2} \norm{\Omega A}^2_F - \langle \Omega , I_p \rangle + \lambda_i \norm{\Omega}_{1,\text{off}} - \left\langle \Delta_i^l , \Omega \right\rangle  \mid \Omega_{\bar I^l(\lambda_i)} = 0 \right\}, \]
		where $\Delta_i^l \in S^p$ is an error vector such that $\norm{\Delta_i^l}_F \leq \epsilon / \sqrt{2}$ and $(\Delta_i^l)_{\bar I^l(\lambda_i)} = 0$;
		\label{algas: main}
		\STATE Compute $R_{\lambda_i}(\Omega^l(\lambda_i))$;
		\ENDWHILE
		\STATE Set $\Omega^*(\lambda_i) = \Omega^l(\lambda_i)$, $I^*(\lambda_i) = I^l(\lambda_i)$ and $\Delta_i = \Delta_i^l$;
		\label{algs:insert}
		\ENDFOR
		\RETURN $\Omega^*$;
	\end{algorithmic}
\end{algorithm}

{Before establishing the convergence of Algorithm \ref{alg:as}}, we present the following proposition to interpret the connection between the optimal solution in Step \ref{algas: initial} of Algorithm \ref{alg:as} and an approximate solution of (\ref{eq:l3}).

\begin{proposition}\label{prop:exp}
	The optimal solution $\Omega^*(\lambda)$ of
	\begin{equation}\label{pro:p1}
		\min_{\Omega \in S^p} \left\{\frac{1}{2} \norm{\Omega A}^2_F - \langle \Omega , I_p \rangle + \lambda \norm{\Omega}_{1,\textup{off}} - \left\langle \Delta , \Omega \right\rangle \right\}
	\end{equation}
	with {$\norm{\Delta}_F \leq \epsilon  / \sqrt{2}$ } can be equivalently found by
	\begin{equation}\label{pro:sol}
		\Omega^*(\lambda) =\textup{Prox}_{\lambda \theta }(\widehat{\Omega}(\lambda) - h(\widehat{\Omega}(\lambda))),
	\end{equation}
	where $\widehat{\Omega}(\lambda)$ is an approximate solution of
	\begin{equation}\label{pro:p2}
		\min_{\Omega \in S^p} \left\{\frac{1}{2} \norm{\Omega A}^2_F - \langle \Omega , I_p \rangle + \lambda \norm{\Omega}_{1,\textup{off}}  \right\}
	\end{equation}
	such that
	\begin{equation}\label{pro:cost} \norm{R_{\lambda}(\widehat{\Omega}(\lambda))}_F \leq \frac{\epsilon}{\sqrt{2} \Big(1 + \|{\widehat{\Sigma}}\|_F \Big)}. \end{equation}
\end{proposition}

\begin{proof}
	Let $\widetilde{\Omega}(\lambda)$ be an optimal solution of problem (\ref{pro:p2}). {  Note that this solution satisfies $R_\lambda (\widetilde{\Omega}(\lambda)) = 0$. Let $\{\Omega^i\}$ be a sequence that converges to $\widetilde{\Omega}(\lambda)$. We then define
	\begin{align*}
		\Delta^i := &R_\lambda(\Omega^i) + h\left(\textup{Prox}_{\lambda \theta} (\Omega^i - h(\Omega^i))\right) - h(\Omega^i) \\
		= &\Omega^i - \textup{Prox}_{\lambda \theta} (\Omega^i - h(\Omega^i)) + h\left(\textup{Prox}_{\lambda \theta} (\Omega^i - h(\Omega^i))\right) - h(\Omega^i).
	\end{align*}
	Since $h$ is continuously differentiable, from Lemma 4.5 of \citep{Du2015}, we have $\text{lim}_{i \to \infty} \norm{\Delta^i}_F = 0$.}
	{This implies the existence of $\widehat{\Omega}(\lambda)$ satisfying the inequality (\ref{pro:cost}).}
	Beginning with the definition of $R_\lambda$, we have $R_\lambda(\widehat{\Omega}(\lambda)) = \widehat{\Omega}(\lambda) - \textup{Prox}_{\lambda \theta} (\widehat{\Omega}(\lambda) - h(\widehat{\Omega}(\lambda)))$. {By combining} this with equation (\ref{pro:sol}), 	
	we obtain
	\[ R_\lambda (\widehat{\Omega}(\lambda)) - h(\widehat{\Omega}(\lambda)) \in \lambda \partial \theta (\Omega^* (\lambda)). \]
	Now, let us define $\Delta := R_{\lambda}(\widehat{\Omega}(\lambda)) + h(\Omega^*(\lambda)) - h(\widehat{\Omega}(\lambda))$. It can be seen that
	\[ \Delta \in h(\Omega^* (\lambda )) + \lambda \partial \theta (\Omega^* (\lambda)), \]
	which means that $\Omega^*(\lambda)$ is an optimal solution of (\ref{pro:p1}) with the given $\Delta$. Besides, we have
	\begin{align*}
		\norm{\Delta}_F = \norm{R_{\lambda}(\widehat{\Omega}(\lambda)) + h(\Omega^*(\lambda)) - h(\widehat{\Omega}(\lambda))}_F
		\leq \left(1 + \norm{\widehat{\Sigma}}_F\right) \norm{R_{\lambda}(\widehat{\Omega}(\lambda))}_F \leq \epsilon / \sqrt{2} .
	\end{align*}
This completes the proof.
\end{proof}

Proposition \ref{prop:exp} presents the connection between $\Omega^*(\lambda)$ and $\widehat \Omega (\lambda)$.  Also, in its proof, it explains how to obtain the error matrix $\Delta$. As for Step \ref{algas: main} in Algorithm \ref{alg:as},  {a more} detailed interpretation can be found from the proof of the following theorem.

\begin{theorem}\label{thm:asmain}
	The solution path $\{\Omega^*(\lambda_i) \mid i = 0, 1, \cdots, k \}$ generated by Algorithm \ref{alg:as} are approximate optimal solutions of a sequence of problems in the form of \[\min_{\Omega \in S^p} \left\{\frac{1}{2} \norm{\Omega A}^2_F - \langle \Omega , I_p \rangle + \lambda_i \norm{\Omega}_{1, \textup{off}}  \right\}\]  {with $\norm{R_{\lambda_i}(\Omega^*(\lambda_i))}_F \leq \epsilon$, $i =  0, 1, \cdots, k$.}
\end{theorem}

\begin{proof}
	We first show that the index set $J^l(\lambda_i)$ is nonempty whenever $\norm{R_{\lambda_i}(\Omega^l(\lambda_i))}_F > \epsilon$.  {Suppose that $J^l(\lambda_i) = \emptyset$. Then we have} \[ - \left(h(\Omega^{l - 1}(\lambda_i))\right)_{ij} \in \lambda_i \left( \partial \theta (\Omega^{l - 1}(\lambda_i)) + \frac{\epsilon}{\lambda_i \sqrt{2|\bar I^{l - 1}(\lambda_i)|}} \mathbb{B}\infty \right)_{ij}, \quad \forall \, (i,j) \in \bar I^{l - 1} (\lambda_i).\]
	Thus, there is a matrix $\widehat{\Delta}^l_i \in \mathbb{S}^p$ with $(\widehat{\Delta}^l_i)_{I^{l-1}(\lambda_i)} = 0$ and $\norm{\widehat{\Delta}^l_i}_\infty \leq \frac{\epsilon}{ {\sqrt{2|\bar I^{l - 1}(\lambda_i)|}}} $ such that
	\begin{equation}\label{thm:k1}
		- \left(h(\Omega^{l - 1}(\lambda_i)) - \widehat{\Delta}^l_i \right)_{ij} \in \lambda_i \left( \partial \theta (\Omega^{l - 1}(\lambda_i))  \right)_{ij}, \quad \forall \, (i,j) \in \bar I^{l - 1} (\lambda_i).
	\end{equation}
	{Since $\Omega^{l - 1}(\lambda_i)$ is an optimal solution of
		\[ \min_{\Omega \in S^p} \left\{\frac{1}{2} \norm{\Omega A}^2_F - \langle \Omega , I_p \rangle + \lambda_i \norm{\Omega}_{1,\text{off}} - \left\langle \Delta_i^{l - 1} , \Omega \right\rangle  \mid \Omega_{\bar I^{l - 1}(\lambda_i)} = 0 \right\}, \]
		where $\Delta_i^{l - 1}$ is an error matrix with $\norm{\Delta^{l - 1}_i}_F \leq \epsilon / \sqrt{2}$ and $(\Delta_i^{l - 1})_{\bar I^{l - 1}(\lambda_i)} = 0$,   by the KKT conditions, we know that} there exists $\Lambda \in \mathbb{S}^p$ with $\Lambda_{I^{l - 1}(\lambda_i)} = 0$ such that
	\begin{equation}\label{thm:k2}
		\left\{
		\begin{aligned}
			&0 \in h(\Omega^{l - 1}(\lambda_i)) - \Delta^{l - 1}_i + \lambda_i \partial \theta(\Omega^{l - 1}(\lambda_i)) - \Lambda, \\
			&\left(\Omega^{l - 1}(\lambda_i)\right)_{\bar I^{l - 1}(\lambda_i)} = 0.
		\end{aligned}
		\right.
	\end{equation}
	Then, combining (\ref{thm:k1}) and (\ref{thm:k2}), we obtain
	\[ - h(\Omega^{l - 1}(\lambda_i)) + \widetilde{\Delta}^{l - 1}_i \in \lambda_i \partial \theta(\Omega^{l - 1}(\lambda_i)),\]
	where $\widetilde{\Delta}^{l - 1}_i \in \mathbb{S}^p$ with $(\widetilde{\Delta}^{l - 1}_i)_{I^{l - 1}(\lambda_i)} = ({\Delta}^{l - 1}_i)_{I^{l - 1}(\lambda_i)}$ and $(\widetilde{\Delta}^{l - 1}_i)_{\bar I^{l - 1}(\lambda_i)} = (\widehat{\Delta}^{l - 1}_i)_{\bar I^{l - 1}(\lambda_i)}$. This means
	\[ \Omega^{l - 1}(\lambda_i) = \textup{Prox}_{\lambda_i \theta}(\Omega^{l - 1}(\lambda_i) - h(\Omega^{l - 1}(\lambda_i)) + \widetilde{\Delta}^{l - 1}_i).\]
	Therefore, we have
	\begin{align*}
		\norm{R_{\lambda_i}(\Omega^{l - 1}(\lambda_i))}_F = \norm{\Omega^{l - 1}(\lambda_i) - \textup{Prox}_{\lambda_i \theta}(\Omega^{l - 1}(\lambda_i) - h(\Omega^{l - 1}(\lambda_i)) )}_F \leq  \norm{\widetilde{\Delta}^{l - 1}_i}_F \leq \epsilon,
	\end{align*}
	where the first inequality follows from the property that the proximal mapping is globally Lipschitz continuous with modulus 1. Hence, a contradiction is found. Thus, $J^{l}(\lambda_i) \neq \emptyset$ if and only if $\norm{R_{\lambda_i}(\Omega^l(\lambda_i))} > \epsilon$. Since the total number of components of $\Omega$ is finite, the while loop in Algorithm \ref{alg:as} will terminate in a finite number of iterations.  {Additionally,  by the KKT conditions, we have} $\Omega^*(\lambda_0)  = \textup{Prox}_{\lambda_0 \theta}(\Omega^*(\lambda_0)) - h(\Omega^*(\lambda_0))) + \Delta_0)$. Thus
	\[\norm{R_{\lambda_0}(\Omega^*(\lambda_0))}_F \leq \norm{\Delta_0}_F \leq \epsilon, \]
which completes the proof.
\end{proof}

We provide two more remarks for Algorithm \ref{alg:as} as follows.

\begin{remark}\label{rm:maxl}
	\underline{Determination of $\lambda_{\max}$ in the ``Input" of Algorithm \ref{alg:as}.}
	Assume that the solution set to (\ref{eq:l3}) is nonempty. We can set \[\lambda_{\max} = \max_{i < j} \left\{ \frac{1}{2} | \widehat{\Sigma}_{ij} / \widehat{\Sigma}_{ii} + \widehat{\Sigma}_{ij} / \widehat{\Sigma}_{jj} |\right\}.\]
	If $\lambda \geq \lambda_{\max}$, the optimal solution of (\ref{eq:l3}) is a diagonal matrix $\Omega^*$ with $\Omega^*_{ii} = 1 / \widehat{\Sigma}_{ii}, i =  1, \cdots, p$. This can be easily verified by the KKT conditions.
\end{remark}

\begin{remark}\underline{Direct extension to the relaxed lasso.}
	Since we have defined the non-zero components set $\bar I$ in Algorithm \ref{alg:as}, we can easily insert the relaxed lasso \citep{Meinshausen2007} into our algorithm after Step \ref{algs:insert} to obtain a solution with  {a better prediction accuracy.}
\end{remark}

{In the remaining of this section, we will discuss the {statistical consistency} and the asymptotic positive definiteness of solutions generated by Algorithm \ref{alg:as}. The discussion of statistical properties here will focus on two cases. The first one is the case where the samples are independent and identically distributed (i.i.d.)  sub-Gaussian random variables (a random variable $Z \in \mathbb{R}$ is called sub-Gaussian \citep{Ravikumar2011} with parameter $\sigma_s \in (0, +\infty)$, if $\mathbb E[Z] = 0$ and $\mathbb E \left[\exp\{wZ\}\right] \leq \exp\left\{\sigma_s^2w^2/2\right\}, \ \forall \, w \in \mathbb{R}$). The second case is for i.i.d. observations of random variables with bounded moments, in particular, satisfying a polynomial-type tail bound.
	
}


{Before proceeding to the main discussion, we present some notation and describe some assumptions.} Suppose that the true precision matrix $\Theta$ is sparse and its minimum eigenvalue $\gamma_{\min}(\Theta) > r$ with some $r > 0 $. For the associated graph, we denote the maximum node degree and the number of edges by $d$ and $s$, respectively. Then, let $s_\theta = \min \{\sqrt{s + p}, d\}$ to describe the sparse level of $\Theta$. 
Similarly to the assumption in  {Section} 3.3 of \citep{Ravikumar2011}, we assume in addition that the parameters $\kappa_\Gamma =\norm{\Gamma_{\Psi,\Psi}^{-1}}_{1,\infty}$, $\kappa_\Sigma = \norm{\Sigma}_{1,\infty}$, $\kappa^* = \max_i \Sigma_{ii}$, and $\alpha$  are constants (not scaling with $p$ and $d$), where $\norm{X}_{1,\infty} = \max_{i} \sum_j |X_{ij}|$ for a matrix $X$. Assume that the following irrepresentability condition holds:
\begin{equation}\label{eq:irrcond}
	\max_{v \in \widebar\Psi} \norm{\Gamma_{v, \Psi} (\Gamma_{\Psi,\Psi})^{-1}}_1 = 1 - \alpha, \quad 0 < \alpha \leq 1,
\end{equation}
where $\Psi$ is the support set of $\Theta$, $\widebar\Psi$ is its complement, and $\Gamma = \frac{1}{2} \Sigma \oplus \Sigma$ ($\oplus$ denotes the Kronecker matrix sum). {We are now ready to present the following propositions. 
}

{
\begin{proposition}\label{lem:consistency1}
	Consider a zero-mean random variable $X = (X_1,\cdots,X_p)$ with covariance $\Sigma$ such that each $X_i/\Sigma_{ii}^{1/2}$ is sub-Gaussian with parameter $\sigma_s > 0$. Assume that the irrepresentability condition \eqref{eq:irrcond}  holds and the samples are drawn independently. Choose $\widebar \lambda_n = C_1 \sqrt{\eta\log p/n}$ for some $\eta > 2$ and $n > C_2 (s_\theta d / r)^2\eta\log p$ with the scalars $C_1$ and  $C_2$ sufficiently large. Then, with probability greater than $1 - 1/p^{\eta - 2}$, we have
	\[
		\norm{\Omega^*(\widebar \lambda_n)- \Theta}_{(2)} \leq C_s s_{\theta}d\sqrt{\eta\log p / n},
	\]
	where $\Omega^*(\widebar \lambda_n)$ is generated by Algorithm \ref{alg:as} with a small enough tolerance $\epsilon \geq  0$, and
	$C_s > 0$ depends only on $\kappa_\Gamma$, $\kappa_\Sigma$, $\kappa^*$, $\alpha$, and $\sigma_s$.
\end{proposition}
}

{
\begin{proof}
	Let $\Theta_{\widebar \lambda_n}$ be the optimal solution set of \eqref{eq:l3} with regularization parameter $\widebar\lambda_n$, and $\widehat{\Theta}(\widebar \lambda_n)$ be arbitrarily chosen from $\Theta_{\widebar \lambda_n}$. For simplicity, we will use $\widehat{\Theta}$ and $\Omega^*$ to represent $\widehat{\Theta}(\widebar \lambda_n)$ and $\Omega^*(\widebar \lambda_n)$, respectively, in this proof.
	By Theorem \ref{thm:asmain}, we have that $\norm{R_{\widebar\lambda_n}(\Omega^*)}_F \leq \epsilon$. We will then verify that there is $\iota > 0$ such that $\norm{\Omega^* - \widehat{\Theta}}_{(2)} \leq \iota \epsilon$.
	It is reasonable to assume that $\norm{R_{\widebar \lambda_n} (\Omega^*) }_F > 0$, since if it is not, $\Omega^*$ is exactly an optimal solution of \eqref{eq:l3}, and so the inequality holds automatically.
	
	Define $\widehat\Omega := \Omega^* - R_{\widebar \lambda_n} (\Omega^*) = \textup{Prox}_{\lambda\theta} (\Omega^* - h (\Omega^*))$. This implies that $0 \in  \partial \lambda  \theta (\widehat{\Omega}) + h(\Omega^*) - R_{\widebar \lambda_n} (\Omega^*)$ and thus $\textup{dist}(0, \partial \lambda \theta (\widehat{\Omega}) + h(\Omega^*)) \leq \epsilon$. Moreover, since $\lambda \theta$ is piecewise linear quadratic, $\partial \lambda \theta$ is locally upper Lipschitzian at $\Omega^*$. By Definition \ref{def:errbd}, for a small enough $\epsilon$, there exist $\kappa_1 \geq 0$ and some neighborhood $N(\Omega^*)$ of $\Omega^*$ such that $\widehat{\Omega} \in N(\Omega^*)$ and
	$\partial \lambda \theta(\widehat{\Omega}) \subset \partial \lambda \theta({\Omega^*}) + \kappa_1 \epsilon B$. Consequently,
	\[
		\textup{dist}(0, \partial \lambda \theta ({\Omega}^*) + h(\Omega^*) + \kappa_1 \epsilon B)  \leq \textup{dist}(0, \partial \lambda \theta (\widehat{\Omega}) + h(\Omega^*)) \leq \epsilon.
	\]
	Besides, it can be verified that $\textup{dist}(0, \partial \lambda \theta ({\Omega}^*) + h(\Omega^*)) - \kappa_1\epsilon \leq \textup{dist}(0, \partial \lambda \theta ({\Omega}^*) + h(\Omega^*) + \kappa_1 \epsilon B)$. Thus, we have $\textup{dist}(0, \partial \lambda \theta ({\Omega}^*) + h(\Omega^*)) \leq (1 + \kappa_1)\epsilon$.
	Moreover, by Lemma \ref{lm:gc}, we obtain that, for some $\kappa_2 \geq 0$,
	\[
		\dist(\Omega^*,\Theta_{\widebar \lambda_n}) \leq \kappa_2 \textup{dist}(0, \partial \lambda \theta ({\Omega}^*) + h(\Omega^*)) \leq \kappa_2(1 + \kappa_1)\epsilon.
	\]
	The closeness of $\Theta_{\widebar \lambda_n}$ follows from the fact that the graph of $\partial \lambda \theta + h$ is polyhedral and thus closed \citep[Theorem 19.1]{Rockafellar1970}. As a result, we can always choose $\widehat{\Theta} \in \Theta_{\widebar \lambda_n}$ such that $\norm{\Omega^*-\widehat{\Theta}}_F = \dist(\Omega^*,\Theta_{\widebar \lambda_n}) \leq \iota\epsilon$, where $\iota = \kappa_2(1+\kappa_1) \geq 0$.
	Then, we know that
	\[
		\norm{\Omega^* - \widehat{\Theta}}_{(2)} \leq \norm{\Omega^* - \widehat{\Theta}}_{F} \leq \iota \epsilon.
	\]
	Consequently, we have that
	\[
		\norm{ \Omega^* - \Theta}_{(2)} \leq \norm{ \Omega^* - \widehat\Theta}_{(2)} + \norm{\widehat\Theta - \Theta}_{(2)} \leq C_s s_{\theta}d\sqrt{\eta\log p / n},
	\]
	where the first inequality comes from the triangular inequality, and the second inequality follows from \citep[Theorem 2]{Zhang2014} and the fact that $\epsilon$ can be chosen proportional to $s_{\theta}d\sqrt{\eta\log p / n}$. This completes the proof.
\end{proof}
}

{Using similar arguments as in the proof of Proposition \ref{lem:consistency1} and  Theorem 3 in \cite{Zhang2014}, we can obtain the following proposition for the polynomial-type tails case.}

{
	\begin{proposition}\label{lem:consistency2}
		Consider a random variable $X = (X_1,\cdots,X_p)$ with covariance $\Sigma$ such that each $X_i/\Sigma_{ii}^{1/2}$
		has finite $4\tilde q$-th moments, i.e. there exist $\tilde q > 0$ and $K_{\tilde q} \in \mathbb{R}$ such that $\mathbb{E}[X_i/\Sigma_{ii}^{1/2}]^{4\tilde q} \leq K_{\tilde q}$. Assume that the irrepresentability condition \eqref{eq:irrcond}  holds and the samples are drawn independently. Choose $\widebar \lambda_n = C_3 \sqrt{p^{\eta / \tilde q} / n}$ for some $\eta > 2$ and $n > C_4 (s_{\theta}d/r)^2 p ^{\eta / q}$ with the scalars $C_3$ and  $C_4$ sufficiently large. Then, with probability greater than $1 - 1/p^{\eta - 2}$, we have
		\[
		\norm{\Omega^*(\widebar \lambda_n)- \Theta}_{(2)} \leq C_p s_{\theta}d\sqrt{p^{\eta / \tilde q} / n},
		\]
		where $\Omega^*(\widebar \lambda_n)$ is generated by Algorithm \ref{alg:as} with a small enough tolerance $\epsilon \geq  0$, and
		$C_p > 0$ depends only on $\kappa_\Gamma$, $\kappa_\Sigma$, $\kappa^*$, $\alpha$, and $K_{\tilde q}$.
	\end{proposition}
}


\begin{remark}
	{The probabilistic models assumed in Propositions \ref{lem:consistency1} and \ref{lem:consistency2} here are mainly used to establish the {statistical consistency} and the asymptotic positive definiteness of the solution.
	}For example, for the sub-Gaussian case,
	when $s_{\theta}d \sqrt{\eta\log p / n} \to 0$,   the estimated solution $\Omega^*(\widebar \lambda_n)$ would be positive definite with probability tending to $1$. {Thereby making it possible to consider the problem without the positive definite constraint for $\Omega$.} Besides, assuming that $n$ is the same as the statement of Theorem 2 in \citep{Zhang2014},  {then}  the positive definite property of the optimal solution estimated by the D trace estimator can be guaranteed.
	Moreover, if the estimated solution is not positive definite, a common remedy is to add a matrix $\pi I_p$ with a small $\pi > |\gamma_{\min}(\Omega^*(\bar\lambda_n))|$ to obtain a positive definite estimation.
	
\end{remark}

\section{A semismooth Newton augmented Lagrangian method}\label{sec:ialmssn}

In this section, we will develop a semismooth Newton augmented Lagrangian method for solving the minimization problems in Steps \ref{algas: initial} and \ref{algas: main} of Algorithm \ref{alg:as}. In order to implement the adaptive sieving reduction strategy more efficiently, we will define some linear operators which allow us to reformulate the original problem into a neater form by removing the zero components. After that, we shall derive an inexact augmented Lagrangian algorithm (ALM) for solving the dual of the original problem and a semismooth Newton algorithm (SSN) for solving its inner problems. We also analyze the global linear convergence rate and the asymptotically superlinear convergence rate of the proposed algorithm.

After introducing a matrix $W \in \mathbb{R}^{p \times n}$, for any $\lambda \in \{\lambda_i, i = 0, 1, \cdots, k\}$,  we can rewrite the original problems in Steps 1 and 10 as follows,
\begin{equation}\label{pbm:o}
	\min_{\Omega , W} \ \left\{\frac{1}{2} \norm{W}^2_F - \langle \Omega , I_p \rangle + \lambda \norm{\Omega}_{1,\text{off}} \mid W - \Omega A =0, \, \Omega \in S_{\bar I (\lambda)}\right\},
\end{equation}
where $S_{\bar I (\lambda)} := \{\Omega \in S^p \mid \Omega_{ij} = 0,\, (i,j) \in \bar I (\lambda)\}$.  {Note  that}
the number of nonzero components in the upper triangle (including the diagonal) of $\Omega$ is less than or equal to $t := (|I(\lambda)| + p) / 2$. Since we are dealing with a problem that is designed for a sparse estimation, $t$ will not be very large.
In practical applications, in order to ensure the statistical validity of the estimated solution, $t$ is generally no greater than $p + n(n-1)/2$.

{
We will then describe how to solve the problem \eqref{pbm:o} efficiently with using the constraint $\Omega \in S_{\bar I (\lambda)}$. As the zero value will not contribute to the computation, we will construct a linear operator to remove the zero components in $\Omega$ and preserve their index information. Using this operator, we can reformulate the problem \eqref{pbm:o} as a problem with a much smaller dimension. Then, when solving the reformulated problem, the operation widely used in the algorithm described later is the multiplication of two $t$-dimensional vectors, while the multiplications of $p\times n$ matrices are required when directly solving the original problem.
Thus, the complexity of the multiplication reduces from $O(p^2n)$ to $O(tn)$.
Now, let us start with the preparation.}
Define a linear operator {$L_{\bar I(\lambda)}: \mathbb{S}^p \to \mathbb{R}^t$} as follows: for any $\Omega \in \mathbb{S}^p$, let $\omega = L_{\bar I(\lambda)}(\Omega)$ be the vector of the remaining components of $\text{svec}(\Omega)$ with those components $\Omega_{ij},\,(i,j)\in {\bar I(\lambda)}$ being removed, where $\text{svec}(\Omega)$ is the vectorized components of the upper triangular (including the diagonal) of $\Omega$. {For ease of notation, we will use $L$ to represent $L_{\bar I(\lambda)}$ throughout this paper.}
Correspondingly, we define the generalized inverse $L^\dagger: \mathbb{R}^t \to \mathbb{S}^p$ of $L$ as follows: for any vector $\omega \in \mathbb{R}^t$, let $\Omega = L^\dagger (\omega)$ be a $p \times p$ symmetric matrix, where all the components $\Omega_{ij},\,(i,j) \in \bar I(\lambda)$ are equal to $0$ and the remaining vectorized upper triangular (including the diagonal) components are exactly $\omega$.
For later use, we further denote
\[ e_1 := L(I_p);\quad e_2 := 2L(E - I_p); \quad e_3 := e_1 + e_2 /4; \quad e_4 := e_1 + e_2,\]
where $E$ is the p-dimensional all-one matrix.
Let $L^*$ and $(L^\dagger)^*$ denote the adjoints of $L$ and $L^\dagger$, respectively.
For any vector $v \in \mathbb{R}^t$, by the definition of the adjoint operator, we have $\langle L(\Omega), v \rangle = \langle \Omega , L^*(v) \rangle$.   {Then we} immediately have
\[ L^*(v) = L^\dagger (v \circ e_3).
\]
Similarly, for any matrix $V \in \mathbb{S}^p$, we know
\[ (L^\dagger)^* (V) = L(V) \circ e_4.\]
We also define another linear operator ${S}: \mathbb{R}^{p \times n} \to \mathbb{R}^t$ by $S(Y): = \frac{1}{2} L(YA^T + AY^T), \ \forall \, Y \in \mathbb{R}^{p \times n}$, whose adjoint $S^* : \mathbb{R}^t \to \mathbb{R}^{p \times n}$ is {given by}
\[S^*(v) = L^*(v)A, \quad \forall \, v \in \mathbb{R}^t.\]

Then, we put a negative sign in front of the objective function of the problem (\ref{pbm:o}). {Following that, we rewrite this problem with the operators defined above and introducing another variable $x \in \mathbb{R}^t$ with $x = \omega \circ e_4$, $\omega = L(\Omega)$.} Consequently, we have the following equivalent problem:
\[ {\bf (P)} \quad  \max_{x \in \mathbb{R}^t}  \left\{ - \left( \Gamma(x) := \frac{1}{2}\norm{S^*(x)}^2_F - \langle x, e_1 \rangle + \lambda / 2\norm{x \circ e_2}_1 \right) \right\}, \]
whose dual is
\begin{equation}\nonumber
	{\bf (D)} \quad \min_{Y \in \mathbb{R}^{p \times n}, \, z \in \mathbb{R}^t  }   \left\{ \frac{1}{2} \norm{Y}^2_F + \delta_{b_\lambda} (z) \mid  S(Y) + z  = e_1 \right\},
\end{equation}
where $\delta_{b_\lambda}$ is an indicator function with $b_\lambda = \{z \in \mathbb{R}^t \mid e_1 \circ z = 0, \, |z_i| \leq \lambda, \, i = 1,\cdots, t \}$.

The KKT system {corresponding to $(\boldsymbol{D})$ is}
\begin{equation}\label{eq:KKT}
	\begin{cases}
		Y-S^*(x)=0, \\
		0\in\partial(\delta_{b_\lambda}(z))- x, \\
		S(Y) + z - e_1=0,
	\end{cases}
	(Y, z, x) \in\mathbb{R}^{p\times n}\times\mathbb{R}^t\times\mathbb{R}^t.
\end{equation}
As  mentioned earlier, we solve $(\boldsymbol{P})$ by solving its dual, provided that the KKT system is nonempty \citep*[Corollary 28.3.1]{Rockafellar1970}. 	
For any $(Y, z, x) \in\mathbb{R}^{p\times n}\times\mathbb{R}^t\times\mathbb{R}^t$,  the Lagrangian function for $(\boldsymbol{D})$ is $	\mathcal{L}(Y,z,x)= \frac{1}{2} \norm{Y}^2_F + \delta_{b_\lambda} (z) - \left \langle S(Y) + z - e_1, x \right \rangle.$
{For any given constant} $\sigma>0$, the augmented Lagrangian function associated with $(\boldsymbol{D})$ is given by, $\forall \, Y\times z\times x \in\mathbb{R}^{p\times n}\times\mathbb{R}^t\times\mathbb{R}^t$,
\begin{equation*}
	\mathcal{L}_{\sigma}(Y,z;x) = \frac{1}{2} \norm{Y}^2_F + \delta_{b_\lambda} (z) - \left \langle S(Y) + z - e_1, x \right \rangle + \frac{\sigma}{2} \left\| {S(Y) + z  - e_1}\right\| ^2_F.
\end{equation*}

Next, we introduce the semismooth Newton augmented Lagrangian method in the following two subsections.

\subsection{An inexact augmented Lagrangian algorithm}\label{subsec:alm}

In this subsection, we will develop an inexact ALM for solving $(\boldsymbol{D})$, and establish the global linear convergence rate and asymptotically superlinear convergence rate of the proposed algorithm. We remark that some standard stopping criteria are used for analyzing the convergence rate of our algorithm here {since the inner problem is solved inexactly. A semismooth Newton algorithm to solve the inner problems of the inexact ALM together with the implementable stopping criteria will be introduced in the next subsection.}

\begin{algorithm}[htb]
	\centering
	\caption{An inexact augmented Lagrangian method for solving ({\bf D}).}
	\label{alg:main}
	\begin{algorithmic}[1]
		\REQUIRE ~~\\ 
		A given parameter $\sigma_0 > 0$;\\
		An initial point $(Y^0,z^0,x^0)\in \mathbb{R}^{p\times n} \times \mathbb{R}^{t} \times \mathbb{R}^{t}$;
		An integer $k = 0$;
		\ENSURE ~~\\ 
		Approximate optimal solution $(\widehat{Y}, \hat z, \hat x)$;
		\WHILE{Stopping  criteria are not satisfied}
		\STATE Compute
		\begin{equation}\label{eq:subprob}
			(Y^{k+1},z^{k+1}) \approx \arg\min \{\Psi_k (Y,z):=\mathcal{L}_{\sigma_k}(Y,z; x^k)\};
		\end{equation}
		\STATE Compute
		$x^{k+1} = x^k - \sigma_k  \left( S(Y^{k+1})+ z^{k+1} - e_1 \right)  $
		and update $\sigma_{k+1} \uparrow \sigma_\infty\leq \infty$;
		\STATE Update $\widehat{Y} = Y^{k+1}, \,\hat z = z^{k + 1}, \, \hat x = x^{k+1}$;
		\STATE $k++$;
		\ENDWHILE
	\end{algorithmic}
\end{algorithm}

Details of the inexact ALM are provided in Algorithm \ref{alg:main}. For later use, we define two maximal monotone operators $\cT_\Gamma$ and $\cT_l$ as follows
\begin{equation}\nonumber
	\cT_\Gamma(x):=\partial \Gamma(x), \quad \cT_\mathcal{L}(Y,z,x):=\{(Y',z',x')\mid(Y',z',-x')\in\partial \mathcal{L}(Y,z,x)\}.
\end{equation}
{To establish the global linear convergence rate of Algorithm \ref{alg:main}, we shall analyze that $\cT_\Gamma$ and $\cT_l$ globally satisfy the error bound condition given in \citep*{Li2018}. Since the objective function $\Gamma$ in $(\boldsymbol{P})$ is PLQ, we know that $\cT_\Gamma$ and $\cT_\Gamma^{-1}$ are both polyhedral by \citep[Proposition 2.2.4]{Sun1986}. It then follows from Lemma \ref{lm:gc} that $\cT_\Gamma$ satisfies the condition \eqref{eq:glr} for $\cT_\Gamma$ with some modulus $\kappa_\gamma \geq 0$ when the optimal solution set $\cO := \cT_\Gamma^{-1}(0)$ of $(\boldsymbol{P})$ is nonempty. In addition, following a similar argument, one can easily obtain that the operator $\cT_l$ satisfies the condition \eqref{eq:glr} with some modulus $\kappa_l \geq 0$ when the KKT system associated with $(\boldsymbol{P})$ and $(\boldsymbol{D})$ is nonempty.}

Now, we are ready to proceed with the analysis of the convergence properties of Algorithm \ref{alg:main}.
Since we solve the inner problem (\ref{eq:subprob}) inexactly, we use the following standard stopping criterion introduced in \citep{Rockafellar1976} to obtain an approximated solution:
\begin{equation}\label{criteria1}
	\Psi_k (Y^{k+1},z^{k+1})-{\rm inf}\Psi_k\leq\epsilon_k^2/2\sigma_k,\ \sum_{k=0}^\infty\epsilon_k \leq \alpha_\epsilon <\infty.
\end{equation}
Besides, for analyzing the convergence rate, we need to introduce  the following two stopping criteria \citep{Rockafellar1976}:
\begin{align*}
	&(S1) \quad \Psi_k (Y^{k+1},z^{k+1})-{\rm inf}\Psi_k\leq (\theta^2_k/2\sigma_k)\norm{x^{k+1}-x^k}^2,\ \sum_{k=0}^\infty\theta_k<+\infty, \\
	&(S2) \quad {\rm dist}(0,\partial\Psi_k(Y^{k+1},z^{k+1}))\leq (\theta'_k/\sigma_k)\norm{x^{k+1}-x^k},\ 0\leq\theta'_k\to 0.
\end{align*}
{ Further discussions on how to implement these criteria into our algorithm in solving the subproblem will be provided at the end of the next subsection.}
Based on \citep{Rockafellar1976, Li2018, zhang2020efficient} and Lemma \ref{lm:gc}, the following theorem establishes  convergence results for the primal sequence $\{x^k\}$ and the dual sequence $\{(y^k,\, z^k)\}$ generated by the inexact ALM.

\begin{theorem}\label{thm:MC}
	Suppose that the solution set to $(\boldsymbol{P})$ is nonempty and the initial point $x^0 \in \mathbb{R}^t$ satisfies $\dist (x^0, \cO) \leq r - \alpha_\epsilon$, where $\alpha_\epsilon$ is given in (\ref{criteria1}). Let $\{(Y^k,z^k,x^k)\}$ be any infinite sequence generated by Algorithm \ref{alg:main} satisfying stopping criteria {\rm (\ref{criteria1})} and $(S1)$. Then, the sequence $\{x^k\}$ converges to an optimal solution $x^* \in \cO$, and for all $k \geq 0$,
	\begin{equation}\label{eq:MC1}
		{\rm dist}(x^{k+1},\cO)\leq\zeta_k{\rm dist}(x^k,\cO),
	\end{equation}
	where $\zeta_k=(\kappa_\gamma(\theta_k+1)(\kappa_\gamma^2+\sigma^2_k)^{-1/2}+\theta_k)(1-\theta_k)^{-1}$ and $\zeta_k\to\zeta_{\infty}=\kappa_\gamma(\kappa_\gamma^2+\sigma^2_{\infty})^{-1/2}<1$ when $k\to\infty$. In addition, the sequence $\{Y^k,z^k\}$ converges to the unique optimal solution $(Y^*,z^*)$ to $(\boldsymbol{D})$. Furthermore, if the stopping criterion $(S2)$ is satisfied, then for all $k \geq 0$,
	\begin{equation}\label{eq:condthm2}
		\norm{(Y^{k+1},z^{k+1})-(Y^*,z^*)} \leq \zeta'_k\norm{x^{k+1}-x^k},
	\end{equation}\nonumber
	where $\zeta'_k=\kappa_l(1+\theta'_k)/\sigma_k$ and $\zeta'_k\to\zeta'_{\infty}=\kappa_l/\sigma_{\infty}$ as $k\to\infty$.
\end{theorem}
\begin{proof}
	Since the solution set to $(\boldsymbol{P})$ is nonempty, the optimal value of $(\boldsymbol{P})$ is finite. Besides, the effective domain of the quadratic function in the objective function of $(\boldsymbol{P})$ is $\mathbb{R}^t$ and the objective function in $(\boldsymbol{D})$ is strongly convex. Then, according to Fenchel's duality theorem \citep*[Corollary 31.2.1]{Rockafellar1970}, the solution set to $(\boldsymbol{D})$ is nonempty and the optimal values of $(\boldsymbol{P})$ and $(\boldsymbol{D})$ are equal to each other and also finite. This implies that the KKT system associated with $(\boldsymbol{P})$ and $(\boldsymbol{D})$ is nonempty. The uniqueness of the optimal solution of $(\boldsymbol{D})$ is obtained directly by the strong convexity of $(\boldsymbol{D})$. Then, from \citep*[Theorem 4]{Rockafellar1976}, we have that the sequence $\{(Y^k,z^k)\}$ is bounded.
	The results under stopping criteria {\rm (\ref{criteria1})} and $(S1)$ can be obtained directly from Lemma \ref{lm:gc}, Lemma 4.1 in \citep*{zhang2020efficient}, and Theorem 5 in \citep{Rockafellar1976}. The remaining result follows from Theorem 3.3 in \citep*{Li2018}. This completes the proof.
\end{proof}

\begin{remark}
	{Suppose that $\{\theta_k\}$ in $(S1)$ and $\{\theta'_k\}$ in $(S2)$ are both nonincreasing for all $k \geq 0$. Since $\{\sigma_k\}$ is nondecreasing, we know that $\{\zeta_k\}$ and $\{\zeta'_k\}$ are nonincreasing. Thus, if we choose $\sigma_0$ large enough such that $\zeta_0, \zeta'_0 < 1$, we have $\zeta_k , \zeta'_k <1$, $\forall \, k \geq 0$.} Then, from Theorem \ref{thm:MC}, we know that Algorithm \ref{alg:main} enjoys a global linear convergence rate.
	If $\sigma_{\infty} = +\infty$, from (\ref{eq:MC1}), the sequence $\{x^k\}$ generated by Algorithm \ref{alg:main} will converge Q-superlinearly. Combing this with (\ref{eq:condthm2}), we know that the sequence $\{(y^k,z^k)\}$ converges R-superlinearly. Thus, according to Theorem \ref{thm:MC}, we can say that our algorithm converges asymptotically superlinearly.
\end{remark}

\subsection{A semismooth Newton algorithm for solving the subproblem in Algorithm \ref{alg:main}}\label{subs:ssn}
In this subsection, we will develop a semismooth Newton (SSN) algorithm for solving (\ref{eq:subprob}), and  introduce the implementations of the stopping criteria used in the previous subsection. Given $\sigma> 0 $ and $x \in\mathbb{R}^t$, the problem is to find an optimal solution for ${\rm min}_{Y,z} \Psi(Y,z),$ $\forall \, (Y,z)\in\mathbb{R}^{p\times n}\times \mathbb{R}^t$. Since $\Psi(\cdot)$ is strongly convex, there is a unique optimal solution $(\widebar Y,\bar z)\in\mathbb{R}^{p\times n}\times \mathbb{R}^t$ and it can be obtained by solving $\mathop{\min}\limits_Y\{\mathop{\inf}\limits_{z}\Psi(Y, z)\}$. For any $Y\in\mathbb{R}^{p\times n}$, we first denote $\psi(Y) := \inf_{z} \Psi (Y,z)$. That is
\begin{align*}
	\psi(Y) =\frac{1}{2} \norm{Y}^2_F - \frac{1}{2\sigma} \left\| x \right\|^2_F + \sigma \inf_{z} \left\{ \sigma^{-1} \delta_{b_\lambda} (z) + \frac{1}{2} \left\|z - (x /\sigma - S(Y) + e_1 )\right\|^2_F\right\}.
\end{align*}
Thus, we can obtain $(\overline Y,\bar z)$ simultaneously by
\begin{equation}\label{eq:subm}
	\overline Y= \arg\min \psi(Y), \qquad
	\bar z = \textup{Prox}_{\delta_{b_\lambda}}(x /\sigma - S(\overline Y) + e_1).
\end{equation}

Define $f(Y) := x /\sigma - S(Y) + e_1$, $\forall \, Y \in \mathbb{R}^{p \times n}$ and $\mathcal{G}(v) := \inf_{z} \sigma^{-1} \delta_{B_\lambda} (z) + \frac{1}{2} \norm{z - v}^2_F$, $\forall \, v \in \mathbb{R}^t$.
Notice that  $\nabla\mathcal{G}(\cdot)$ is continuously differentiable \citep[Theorem 2.26]{rockafellar2009variational}. We then have
\begin{equation}
	\nabla \psi (Y) = Y - \sigma S^*\left(\nabla \mathcal{G}(f(Y))\right),
\end{equation}
where $\nabla \mathcal{G}(f(Y)) = \textup{Prox}_\phi(f(Y))$
with $\phi(v) =\lambda / 2 \norm{v \circ e_4}_1$, $\forall \, v \in \mathbb{R}^t$.
Therefore, $\overline Y$ can be obtained by finding a root of
\begin{equation}\label{Dfy}
	\nabla \psi (Y) = 0.
\end{equation}
For later use, we define $\hat{\partial}^2 \psi(Y)$ as follows:
\begin{equation}\nonumber
	V\in\hat{\partial}^2 \psi(Y)\ \Leftrightarrow \ \exists\, u \in \cU \ {\rm such \ that} \ V(D)\in \partial ^2 \psi(Y)(D), \quad \forall D\in\mathbb{R}^{p\times n},
\end{equation}
where ${\partial}^2 \psi(Y)$ is the generalized Hessian of $\psi$ at $Y$, and $\cU := \partial \textup{Prox}_\phi (f(Y))$ is the Clarke subdifferential of $\textup{Prox}_\phi (f(\cdot))$ at $Y$ \citep{Clarke1990}.
Besides,
\begin{equation}\label{eq:partalhessian}
	\hat\partial ^2 \psi(Y)(D) = \{D + \sigma L^\dagger (S(D) \circ u \circ e_3)A \mid u \in \cU\},\quad \forall D\in\mathbb{R}^{p\times n}.
\end{equation}
It then follows from \citep*[Proposition 2.3.3 and Theorem 2.6.6]{Clarke1990} that $\partial ^2 \psi(Y)(D) = \hat\partial ^2 \psi(Y)(D)$, $\forall D\in\mathbb{R}^{p\times n}$.
Now, we can introduce our SSN algorithm for solving (\ref{Dfy}) in Algorithm \ref{alg:sub}.

\begin{algorithm}[htb]
	\centering
	\caption{A semismooth Newton algorithm for solving (\ref{Dfy}).}
	\label{alg:sub}
	\begin{algorithmic}[1]
		\REQUIRE ~~\\ 
		Given parameters $\mu \in (0, 1/2)$, $\bar{\eta} \in (0, 1)$, $\tau \in (0,1]$, and $\delta \in (0, 1)$;\\
		An initial point $ Y^0\in \mathbb{R}^{p\times n}$ and a given $x \in \mathbb{R}^t$; \\
		An integer $j = 0$;
		\ENSURE ~~\\ 
		Approximate optimal solution $\widehat{Y}$;
		\WHILE{Stopping criteria are not satisfied}
		\STATE  Choose $u_j \in \partial \textup{Prox}_{\phi}(x/\sigma - S(Y^j) + e_1) $. For $D \in \mathbb{R}^{p \times n}$, let $ V_jD :=D+\sigma L^{\dagger} (S(D)\circ u_j \circ e_3)A$. Solve the equation
		\begin{equation}\label{eq:CGm}
			V_jD+\nabla \psi (Y^j)=0
		\end{equation}
		by the conjugate gradient algorithm to find $D^j$ such that
		\begin{equation}\label{eq:tau}
			\norm{V_jD^j+\nabla \psi(Y^j)} \leq \min ( \bar{\eta},\norm{\nabla \psi(Y^j)}^{1+\tau} );
		\end{equation}
		\STATE (Line search) Set $\alpha_j=\delta^{m_j}$, where $m_j$ is the first nonnegative integer m such that
		\begin{equation}
			\psi(Y^j+\delta^{m_j}D^j) \leq \psi(Y^j)+\mu\delta^{m_j} \left \langle \nabla \psi(Y^j),D^j \right \rangle;
		\end{equation}
		\STATE Set $Y^{j+1}=Y^j+\alpha_jD^j$ and update $\hat Y = Y^{k+1}$;
		\STATE $j++$;
		\ENDWHILE
	\end{algorithmic}
\end{algorithm}

{It is well known that continuous piecewise affine functions and twice continuously differentiable functions are all strongly semismooth everywhere. Besides, the composition preserves the (strongly) semismooth \citep{fischer1997solution}. Since $\textup{Prox}_{\lambda\norm{\cdot}_{1}}(\cdot)$ is Lipschitz continuous piecewise affine and $S^*$ is differentiable, we know that $\nabla \psi(\cdot)$ is strongly semismooth. Now, we are ready to state the convergence result of SSN in the following theorem.
}

\begin{theorem}\label{convergence-SSN}
	 Let the sequence $\{Y^k\}$ generated by Algorithm \ref{alg:sub}. Then $\{Y^k\}$ converges to the unique optimal solution $\widebar Y\in\mathbb{R}^{p\times n}$ of the problem in {\rm (\ref{eq:subm})} and the convergence is of order $1+\tau$, that is
	\begin{align*}
		\norm{Y^{j+1}-\widebar Y}=O \left(\norm{Y^j-\widebar Y}^{1+\tau}\right),
	\end{align*}
	{where $\tau$ is defined in \eqref{eq:tau}.}
\end{theorem}
\begin{proof}
	From \eqref{eq:partalhessian}, we know that $V_j\in\hat{\partial}^2 \psi(Y^j)$, $\forall \, j \geq 0$. Besides, all $V_j, \, j \geq 0$ are self-adjoint positive definite. Because $\nabla \psi(\cdot)$ is strongly semismooth, the stated conclusion can be derived by following the proofs of Theorem 3.5 in \citep{Zhao2010}.
\end{proof}

Theorem \ref{convergence-SSN} shows that the convergence rate of SSN is of order $1 + \tau$. This implies that SSN can converge quadratically if $\tau = 1$. However, in practice, we often set $\tau$ to be smaller, such  as  $0.1$ or $0.2$, for computational considerations.

We then discuss how to implement the stopping criteria (\ref{criteria1}), (S1), and (S2) into Algorithm \ref{alg:sub} to guarantee the convergence results as discussed in Section \ref{subsec:alm}. Since $\psi$ is strongly convex with a parameter $\tau_c > 0$, we can obtain
\begin{align*}
	\Psi_k(Y^{k+1},z^{k+1})-{\rm inf}\Psi_k=\psi_k(Y^{k+1})-{\rm inf}\psi_k\leq 1/(2\tau_c)\norm{\nabla\psi_k(Y^{k+1})}^2
\end{align*}
and $(\nabla\psi_k(Y^{k+1}),0)\in\partial\Psi_k(Y^{k+1},Z^{k+1})$. As a result, in practical  implementation, we can replace the stopping criteria (\ref{criteria1}), (S1) and (S2) to the following implementable criteria
\begin{equation}\nonumber
	\left\{
	\begin{aligned}
		&     \norm{\nabla\psi_k(Y^{k+1})}\leq\sqrt{\tau_c/\sigma_k}\epsilon_k,\ \sum_{k=0}^\infty\epsilon_k <\infty,      \\
		&      \norm{\nabla\psi_k(Y^{k+1})}\leq \sqrt{\tau_c\sigma_k}\theta_k\norm{\frac{1}{2}Y^{k+1}A^T+\frac{1}{2}A(Y^{k+1})^T+z^{k+1}-c},\quad \sum_{k=0}^\infty\theta_k<+\infty,      \\
		&      \norm{\nabla\psi_k(Y^{k+1})}\leq \theta'_k\norm{\frac{1}{2}Y^{k+1}A^T+\frac{1}{2}A(Y^{k+1})^T+z^{k+1}-c},\quad 0\leq\theta'_k\to 0.
	\end{aligned}
	\right.
\end{equation}
In other words, the stopping criteria (\ref{criteria1}), (S1), and (S2) will be satisfied when $\norm{\nabla\psi_k(Y^{k+1})}$ is small enough.

\section{Some further discussions}\label{sec:dis}

{

Our main purpose in this section is to discuss the computational and memory complexities of MARS compared to some state-of-the-art algorithms to demonstrate its performance. Before that, to facilitate the discussion, we will describe the overall structure of MARS by connecting the three algorithms it contains.

As we mentioned in the introduction, the three algorithms presented in the two sections above comprise MARS. For clarity, an illustration is provided in Figure \ref{fig:mars} to connect the three algorithms. This figure demonstrates how MARS generates a precision matrix solution path. Specifically, to generate a precision matrix solution path, we need to solve the $\ell_1$ penalized D-trace estimator for a collection of given regularization parameters $\{\lambda_i\},\, i = 0,1,\cdots,k$. For each $\lambda_i$,  Algorithm \ref{alg:as} enables the generation of a precision matrix by solving some reduced subproblems with a considerably smaller dimension than the original problem. Following that, Algorithm \ref{alg:main} solves each reduced subproblem, whereas the challenging problems (\ref{eq:subprob}) are solved by Algorithm \ref{alg:sub}. Note that, while MARS is designed to generate a solution path, it is also able to efficiently solve problems with a single regularization parameter (extensive numerical experiments can be found in Section \ref{sec:numexpcom}).
}

{

\begin{figure}[htbp]
	\centering
	\includegraphics[width=1\textwidth]{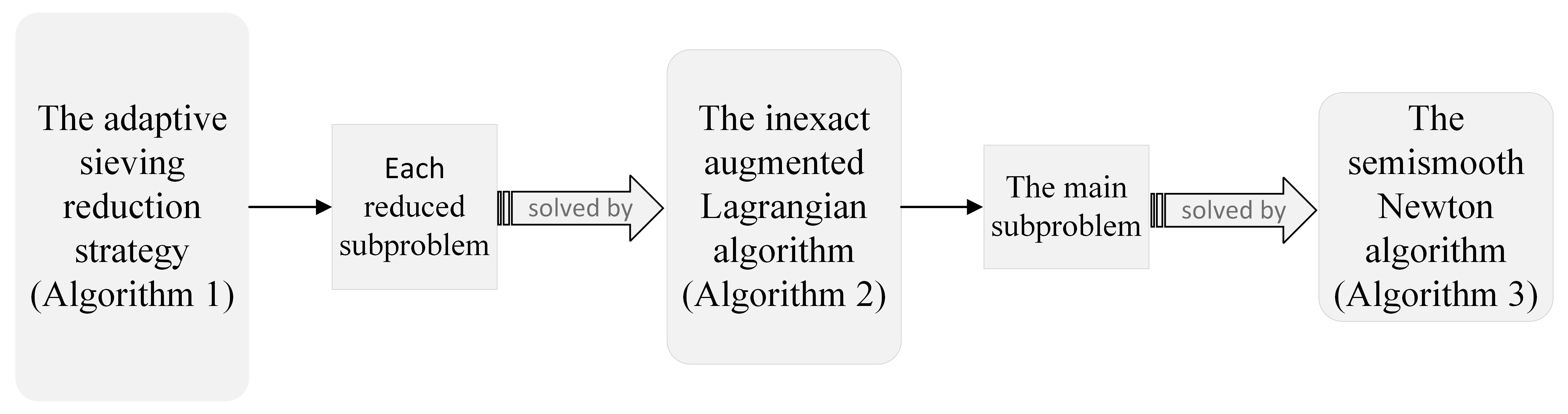}
	\caption{The overall structure of MARS}
	\label{fig:mars}
\end{figure}
}

{
Next, we discuss the convergence properties of the proposed three algorithms. The conclusions reached in Sections \ref{sec:as} and \ref{sec:ialmssn} concerning this are summarized below. In the proof of Theorem \ref{thm:asmain}, we know that Algorithm \ref{alg:as} can terminate in a finite number of iterations; Theorem \ref{thm:MC} shows that Algorithm \ref{alg:main} is asymptotically superlinearly convergent; According to Theorem \ref{convergence-SSN}, the local convergence rate of Algorithm \ref{alg:sub} is at least superlinear (capable of reaching a quadratic convergence rate). For Algorithm \ref{alg:as}, one may be concerned that it needs a large number of iterations in real applications. In fact, from the extensive numerical experiments we have conducted, we found that for each regularization parameter, it usually obtained a solution satisfying the stopping criterion within $3$ steps. For some large regularization parameters, the while loop in Algorithm $1$ stops even after one iteration.
The {sensible} worst-case iteration complexity of the inexact ALM (Algorithm \ref{alg:main}) and the semismooth Newton method (Algorithm \ref{alg:sub}) are still not very clear, although they have been widely used to solve different problems. We leave these two topics for further research in the future.
We point out that, due to the close connection between the proximal point method and ALM \citep{Rockafellar1976}, for the exact ALM, the result in \citep{guler1991convergence} indicates that the convergence rate in terms of the objective function value of the primal problem $\bf (P)$ is
\[
O\left(\frac{1}{\sum_{j=0}^{k-1}\sigma_j}\right),
\]
where $\{\sigma_j\}$ is a given nondecreasing sequence. This suggests that its convergence rate is at least $O(1/k)$ and it can be arbitrarily fast. For the inexact ALM, there is no such result yet to the best of our Knowledge. However, the fast convergence rate has been proved in Theorem \ref{thm:MC}. The numerical experiments presented in the next section also demonstrate the superior performance of the inexact ALM, where Algorithm \ref{alg:main} could reach a satisfied solution within a few iterations (typically no more than $7$).
Besides, the empirical performance of Algorithm \ref{alg:sub} is promising due to its super-fast local convergence rate (up to quadratic). As a piece of evidence, in most of the simulation studies we conducted, the required Newton steps are not larger than $5$. We want to emphasize that the semismooth Newton method is the key to the success of the designed algorithm. On the one hand, the performance of the inexact ALM highly depends on the accuracy of the obtained solution to (\ref{eq:subm}). On the other hand, as we will show in the following paragraph, the per-iteration computational and memory complexities are comparable to or even better than the ones of some first-order methods, such as ADMM.
}

\begin{table}[]
\centering
\caption{{The per-iteration complexities of the most internal algorithms of MARS, SSNAL, iADMM, and eADMM.}}
	\resizebox{\textwidth}{!}{
	\begin{threeparttable}
\begin{tabular}{lllll}
\hline
\multicolumn{1}{|c|}{}           & \multicolumn{1}{c|}{MARS} & \multicolumn{1}{c|}{SSNAL} & \multicolumn{1}{c|}{iADMM} & \multicolumn{1}{c|}{eADMM} \\ \hline
\multicolumn{1}{|c|}{Complexity} & \multicolumn{1}{c|}{$O(4w_mn\sqrt{\kappa_m})$}     & \multicolumn{1}{c|}{$O((3p^2+w_s)n\sqrt{\kappa_s})$}      & \multicolumn{1}{c|}{$O(4p^2n\sqrt{\kappa_i})$}      & \multicolumn{1}{c|}{$O(6p^2n+2n^2p+C_{s})$}
\\ \hline
\end{tabular}
		\begin{tablenotes}
		    {
		    \item[1.] The parameter $w_m = 2t-p$ denotes the cardinality of the nonzero index set $I$ as given in \eqref{pbm:o}, $w_s$ is the number of nonzero elements in $U_j$ (defined in Step \ref{step:ssnalmain} of Algorithm \ref{alg:ssnal}), and $C_s$ is the average per-iteration computational complexity of the SVD decomposition and the construction of two matrices $\Lambda_1$ and $\Lambda_2$ in Algorithm \ref{alg:eADMM}.
			\item[2.] The parameters $\kappa_m$, $\kappa_s$, and $\kappa_i$ are the condition numbers of the matrices in the linear systems solved by the CG algorithm in MARS, SSNAL, and iADMM respectively.
			\item[3.] Under the sparse and high dimensional settings, due to $w_m$ and $w_s$ representing the number of nonzero elements in the generated precision matrix, they both are far smaller than $p^2$.
			}
		\end{tablenotes}
	\end{threeparttable}
    	}
\label{tab:itercom}
\end{table}

{
From the detailed steps of our algorithms, we can see that the vast majority of the computations in Algorithm \ref{alg:as} are contained in solving the reduced subproblems, and almost all of the computations in Algorithm \ref{alg:main} are contained in solving its subproblems (\ref{eq:subprob}). Therefore, based on the connections among the three algorithms shown in Figure \ref{fig:mars}, we may conclude that Algorithm \ref{alg:sub} is responsible for a significant portion of the computations in MARS. Then, for simplicity, we only discuss the per-iteration computational complexity of Algorithm \ref{alg:sub} here. We know that in each iteration of Algorithm \ref{alg:sub}, a Newton equation is solved using the conjugate gradient (CG) algorithm, and a line search step is applied to determine the step size. {Then, according to the convergence result of the CG algorithm in \citep[Chapter 10]{Shewchuk1994introduction}, we know that the computational complexity of CG in each iteration of Algorithm \ref{alg:sub} is $O(4w_m\sqrt{\kappa_m})$, where $w_m = 2t-p$ is the cardinality of the nonzero index set $I$ as given in \eqref{pbm:o}, and $\kappa_m$ is the condition number of the corresponding matrix in the Newton system and it contributes to measuring the maximum iteration number of CG. According to the numerical experiments we performed in Section \ref{sec:num}, the iteration number of CG is usually small (about $5$), and the semismooth Newton algorithm takes a unit step in most of the cases (i.e., the iteration number of the line search seldom exceeds $2$). Thus, as shown by the numerical results, we may conclude that the per-iteration complexity of Algorithm \ref{alg:sub} is $O(4w_mn\sqrt{\kappa_m})$. }
For comparison, the per-iteration complexities of the most internal algorithms for SSNAL, iADMM, and eADMM (see Section \ref{soa} for details) are also listed in Table \ref{tab:itercom}. Recall that, $t$ is the dimension of problem ${\bf (P)}$ and it is far smaller than $p(p+1)/2$ under the sparse and high-dimensional settings. Therefore, MARS can be much more efficient than other algorithms. As can be seen from Algorithms \ref{alg:main} and \ref{alg:sub}, the memory complexity of our algorithm is $O(pn + t)$.  Under the sparse and high-dimensional settings, $pn$ is generally greater than $t$, and so the memory complexity is $O(pn)$ instead of the usual $O(p^2)$ (such as in eADMM).
}

{
From the above discussion, we know that the computational and memory complexities of our algorithm are satisfactory compared to other algorithms, which also explains the promising performance of our algorithm given in the next section.
}

\section{Numerical experiments}\label{sec:num}

In this section, we will conduct several tests to illustrate the performance of our MARS. For comparison, we consider several popular solvers including scio \citep{Liu2015}, EQUAL \citep{Wang2020}, glasso \citep{Friedman2008}, and QUIC \citep{Hsieh2014}. Since the existing popular methods are mainly first-order methods and the stopping criteria of those algorithms are different from each other and also ours, for better comparison, we will also introduce some other algorithms for solving (\ref{eq:l3}) in Section \ref{soa}. Specifically, we will introduce a second-order algorithm, namely a semismooth Newton augmented Lagrangian method (SSNAL), and two kinds of alternating direction methods of multipliers (ADMM), where one is derived by solving the sub-problem inexactly (iADMM) and the other derived by solving it exactly (eADMM). The numerical experiments here are divided into two parts by the source of the data. The first part is conducted for some random data generated by five given models, and the second part uses data derived from real-world applications.

Before proceeding to the experiments, we provide some explanations about our MARS. For any vector $\nu \in \mathbb{R}^t$,  we can choose the $i$-th component of $u \in \partial \textup{Prox}_{\phi}(\mathcal{\nu})$ as
\begin{equation}\nonumber
	u_{i}=
	\begin{cases}
		0,& \ \text{if } d_i \neq 0\,\&\, |\nu_i|\leq\lambda, \\
		1,& \ \text{otherwise},
	\end{cases}
	\quad i = 1,2,\ldots,t.
\end{equation}
This is because the components of $\nu^p:=\textup{Prox}_{\phi}(\nu)$ can be found by
\begin{equation}\nonumber
	\nu^p_i=
	\begin{cases}
		\textup{sign}(\nu_{i})\cdot \max\{|\nu_i|-\lambda\, ,0\},& \ \text{if } d_i \neq 0,\\
		\nu_i,& \ \text{if } d_i = 0,
	\end{cases}
	\quad i = 1,2,\ldots,t.
\end{equation}
In our MARS, we use the relative KKT residual
\begin{equation}\nonumber
	\eta=\frac{\norm{R(\Omega)}_F}{1+\norm{h\left(\Omega\right)}_F+\norm{\Omega}_F}
\end{equation}
to measure the accuracy of the generated solution $\Omega$. That is, we use $\eta$ to decide whether our MARS should be stopped.
Unless otherwise specified, we set the stopping tolerance to $10^{-4}$ for all the solvers/algorithms except EQUAL in the following experiments. Based on several tests, the stopping tolerance of EQUAL is set to $10^{-6}$. The reason for such an adjustment is that their stopping criterion is determined by the distance between two solutions in two consecutive iterations, and a slightly larger stopping tolerance may cause the generated solution to be too far from the optimal solution set, in terms of the relative KKT residual. Moreover, from the test results in Section \ref{NE:real}, we found that even if the stopping tolerance is set to $10^{-6}$, the relative KKT residuals of some solutions obtained by EQUAL are still not less than $10^{-3}$ (we also try to set the stopping tolerance to $10^{-5}$, but none of the associated relative KKT residuals is less than $5 \times 10^{-2}$. More details can be found from Appendix \ref{app:stopEQUAL}).

All the numerical results are obtained by running Microsoft R Open 4.0.2 on a Windows workstation (Intel(R) Core(TM) i7-10700 CPU @2.90GHz 2.00GHz RAM 32GB). For simplicity, we will use R to represent Microsoft R Open 4.0.2.

\subsection{Some other algorithms}\label{soa}

In this subsection, we will introduce some other algorithms to compare the performance with our MARS. The first algorithm is the SSNAL, which is similar to our algorithm, but does not use the adaptive sieving reduction strategy. For later use, we define a linear operator $\mathcal{S}:\mathbb{R}^{p\times n}\to \mathbb{S}^p$ by $\mathcal{S}(Y)=\frac{1}{2}(YA^T+AY^T),\,\forall \, Y\in\mathbb{R}^{p\times n}$, whose conjugate $\mathcal{S}^*$ is in the form of $\mathcal{S}^*(\Omega)=\Omega A,\,\forall \, \Omega\in\mathbb{S}^p$.
By putting a negative sign in front of the objective function of the original problem (\ref{eq:l3}), we obtain
\begin{equation}\label{eq:Pm}
	\mathop{\max}\limits_{\Omega \in \mathbb{S}^{p}}\left\{ - \left( \frac{1}{2} \norm{\Omega A}^2_F- \langle \Omega, I_p \rangle + \lambda\left\|\Omega\right\|_{1,{\rm off}} \right) \right\},
\end{equation}
whose dual is
\begin{equation}\label{eq:Dm}
	\mathop{\min} \limits_{Y\in\mathbb{R}^{p\times n},Z\in\mathbb{S}^{p}}\left\{ \frac{1}{2}\norm{Y}^2_F +\delta_{B_\lambda}(Z)\mid \frac{1}{2}(YA^T+AY^T)+Z=I_p \right\}.
\end{equation}
Given $\sigma > 0$, the augmented Lagrangian function associated with (\ref{eq:Dm}) is given by
\begin{align*}\nonumber
	\mathcal{L}^m_{\sigma}(Y,Z;\Omega) =\frac{1}{2} \norm{Y}^2_F+\delta_{B_\lambda}(Z)- \left \langle \mathcal{S}(Y)+Z-I_p,\Omega \right \rangle+\frac{\sigma}{2}\left\|\mathcal{S}(Y)+Z-I_p\right\|^2_F.
\end{align*}

\begin{algorithm}[htb]
	\centering
	\caption{An SSNAL for solving (\ref{eq:Dm}).}
	\label{alg:ssnal}
	\begin{algorithmic}[1]
		\REQUIRE ~~\\ 
		Given parameters $\sigma_0 > 0$, $\mu \in (0, 1/2)$, $\bar{\eta} \in (0, 1)$, $\tau \in (0,1]$, and $\delta \in (0, 1)$;\\
		An initial point $(Y^0,Z^0,\Omega^0)\in \mathbb{R}^{p\times n} \times \mathbb{S}^{p} \times \mathbb{S}^{p}$;
		An integer $k = 0$;
		\ENSURE ~~\\ 
		Approximate optimal solution $(\widehat{Y}, \widehat{Z}, \widehat{\Omega})$;
		\WHILE{Stopping criteria are not satisfied}
		\STATE An integer j = 0; Set an initial value $Y_0 = Y^k$ for the inner loop;
		\WHILE{Stopping criteria of the inner problem are not satisfied}
		\STATE\label{step:ssnalmain}  Choose $U_j \in \partial \textup{Prox}_{\theta}(\Omega^k/\sigma - \mathcal{S}(Y_j) + I_p) $. For $D \in \mathbb{R}^{p \times n}$, let $ V_jD :=D+\sigma  (\mathcal{S}(D)\circ U_j )A$. Solve the equation
		\begin{equation}\label{eq:CG}
			V_jD+\nabla \psi^m (Y_j)=0
		\end{equation}
		by the conjugate gradient algorithm to find $D_j$ such that
		\begin{equation}
			\norm{V_jD_j+\nabla \psi^m(Y_j)} \leq \min ( \bar{\eta},\norm{\nabla \psi^m(Y_j)}^{1+\tau} );
		\end{equation}
		\STATE (Line search) Set $\alpha_j=\delta^{m_j}$, where $m_j$ is the first nonnegative integer m such that
		\begin{equation}
			\psi^m(Y_j+\delta^{m_j}D_j) \leq \psi^m(Y_j)+\mu\delta^{m_j} \left \langle \nabla \psi^m(Y_j),D_j \right \rangle;
		\end{equation}
		\STATE Set $Y_{j+1}=Y_j+\alpha_jD_j$ and update $ Y^{k + 1} = Y_{j+1}$;
		\STATE $j++$;
		\ENDWHILE
		\STATE Compute
		$ Z^{k + 1} = \textup{Prox}_{\delta_{B_\lambda}}( \Omega^k/\sigma_k - \mathcal{S}(Y^{k + 1}) + I_p);$
		\STATE Compute
		$\Omega^{k+1} = \Omega^k - \sigma_k  \left( \mathcal{S}(Y^{k+1})+ Z^{k+1} - I_p \right)  $
		and update $\sigma_{k+1} \uparrow \sigma_\infty\leq \infty$;
		\STATE Update $\widehat{Y} = Y^{k+1}, \,\widehat{Z} = Z^{k+1},\, \widehat{\Omega} = \Omega^{k+1}$;
		\STATE $k++$;
		\ENDWHILE
	\end{algorithmic}
\end{algorithm}

For any $Y\in\mathbb{R}^{p\times n}$, we define $\psi^m(Y) := \inf_{Z} \mathcal{L}^m_{\sigma}(Y,Z;\Omega)$.
Then, we can obtain a unique optimal solution $(\widebar Y,\widebar Z)$ of $\min_{Y,Z} \mathcal{L}^m_\sigma(Y,Z;\Omega)$ by
\begin{equation}\label{eq:sub1}
	\widebar Y= \arg\min \psi^m(Y), \quad
	\widebar Z = \textup{Prox}_{\delta_{B_\lambda}}( \Omega/\sigma -\mathcal{S}(\widebar Y) + I_p).
\end{equation}
Similar to the arguments in Section \ref{subs:ssn}, we have
\begin{equation}
	\nabla \psi^m(Y)=Y-\sigma \textup{Prox}_{\theta}( \Omega/\sigma -\mathcal{S}( Y) + I_p)   A, \quad \forall \, Y \in \mathbb{R}^{p \times n}
\end{equation}
and
\begin{equation}\nonumber
	\hat{\partial}^2 \psi^m(Y)(D)= \{ D+\sigma \mathcal{S}(D)\circ U)A\mid U\in\cU\},\quad \forall D\in\mathbb{R}^{p\times n},
\end{equation}
where $\cU = \partial \textup{Prox}_{\theta}( \Omega/\sigma -\mathcal{S}( Y) + I_p)$.
Now, we are ready to introduce the detailed steps of SSNAL in Algorithm \ref{alg:ssnal}.

Next, we will introduce an iADMM for solving (\ref{eq:Dm}). Given $Z,\, \Omega \in \mathbb{S}^p$, we define
\begin{align*}
	&\psi^a(Y) := \mathcal{L}^m_{\sigma}(Y,Z;\Omega) = \frac{1}{2}\norm{Y}^2_F +\frac{\sigma}{2}\left\|\mathcal{S}(Y) - C\right\|^2_F +\delta_{B_\lambda}(Z) -\frac{1}{2\sigma}\norm{\Omega}^2_F\,,
\end{align*}
where $C=\Omega/\sigma+I_p-Z$. The gradient of $\psi^a (\cdot)$ at $Y\in\mathbb{R}^{p\times n}$ is
\begin{equation}
	\nabla\psi^a(Y)=Y+\mathcal{S}(Y)A-\sigma CA.
\end{equation}
Then the optimal solution $\widebar Y_a={\rm arg min}_{Y}\psi^a(Y)$ can be found by finding a root of $\nabla\psi^a(Y)=0$, which can be rewritten as
\begin{equation}\label{eq:ADMMmain}
	(I_p+\sigma\mathcal{S}^* \mathcal{S})(Y)=\sigma CA, \quad\forall\, Y\in\mathbb{R}^{p\times n}.
\end{equation}
We then use the conjugate gradient method to find a solution of (\ref{eq:ADMMmain}). Detailed steps of iADMM are provided in Algorithm \ref{alg:iADMM}.

\begin{algorithm}[htb]
	\centering
	\caption{An inexact ADMM for solving (\ref{eq:Dm}).}
	\label{alg:iADMM}
	\begin{algorithmic}[1]
		\REQUIRE ~~\\ 
		Given parameters $\sigma > 0$ and $\pi \in (0, (1+\sqrt{5})/2)$;
		An initial point $(Y^0,Z^0,\Omega^0)\in \mathbb{R}^{p\times n} \times \mathbb{S}^{p} \times \mathbb{S}^{p}$;
		An integer $k = 0$;
		\ENSURE ~~\\ 
		Approximate optimal solution $(\widehat{Y}, \widehat{Z}, \widehat{\Omega})$;
		\WHILE{Stopping criteria are not satisfied}
		\STATE Use conjugate gradient method to find an optimal solution $Y^{k+1}$ such that
		\[Y^{k + 1} \approx \argmin_{Y \in \mathbb{R}^{p \times n}} \psi^a(Y).\]
		\STATE Compute
		$ Z^{k + 1} = \textup{Prox}_{\delta_{B_\lambda}}( \Omega^k/\sigma - \mathcal{S}(Y^{k + 1}) + I_p);$
		\STATE Compute
		$\Omega^{k+1} = \Omega^k - \pi \sigma  \left( \mathcal{S}(Y^{k+1})+ Z^{k+1} - I_p \right)  $;
		\STATE Update $\widehat{Y} = Y^{k+1}, \,\widehat{Z} = Z^{k+1},\, \widehat{\Omega} = \Omega^{k+1}$;
		\STATE $k++$;
		\ENDWHILE
	\end{algorithmic}
\end{algorithm}

For introducing the eADMM, we need equivalently rewrite (\ref{eq:Pm}) to be
\begin{equation}
	\mathop{\max} \limits_{\Omega, \, M \in \mathbb{S}^p}\ \left\{ - \left(\frac{1}{2} \norm{MA}^2_F -\left \langle \Omega,I_p \right \rangle + \lambda\left\|\Omega\right\|_{1,\textup{off}} \right) \mid M -\Omega = 0 \right\},
\end{equation}
whose dual is
\begin{equation}\label{eq:De}
	\mathop{\min} \limits_{V, \, Z \in \mathbb{S}^p}\ \left\{ \frac{1}{2} \norm{VA}^2_F + \delta_{B_\lambda}(Z)  \mid \frac{1}{2}(V\widehat{\Sigma}+\widehat{\Sigma}V)+Z=I_p \right\}.
\end{equation}
This reformulation is designed for implementing the conclusion in \citep{Wang2020}.
We point out that, the main difference between $(\ref{eq:Dm})$ and $(\ref{eq:De})$ is that the dimensions of variables $Y$ and $V$ are $p \times n$ and $p \times p$ respectively.   Under the high-dimensional setting, it can be easily seen that solving $(\ref{eq:Dm})$ can be  much more efficient than solving $(\ref{eq:De})$.  But, when $p$ is relatively small, eADMM could be more efficient than other algorithms.
Then, we shall start introducing eADMM. Given $\sigma > 0$, for any $V, Z \in \mathbb{S}^p$, the augmented Lagrangian function associated with $(\ref{eq:De})$ is given by
\begin{equation}\nonumber
	\mathcal{L}^e_{\sigma}(V,Z;\Omega)=\frac{1}{2} \norm{VA}^2_F+\delta_{B_\lambda}(Z)- \left \langle T(V)+Z-I_p,\Omega \right \rangle+\frac{\sigma}{2}\left\|T(V)+Z-I_p\right\|^2_F,
\end{equation}
where $T(V) := \frac{1}{2}(VS+SV)$. Likewise, given $Z, \Omega \in \mathbb{S}^p$, we define
\begin{align*}
	\psi^e(V) := \mathcal{L}^e_{\sigma}(V,Z;\Omega) = \frac{1}{2}\norm{VA}^2_F +\frac{\sigma}{2}\left\|T(V) - C\right\|^2_F +\delta_{B_\lambda}(Z) -\frac{1}{2\sigma}\norm{\Omega}^2_F,
\end{align*}
where $C = \Omega/\sigma+I_p-Z$. Then the optimal solution $\widebar V_e$ of $ \arg\min\psi^e(V)$ can be obtained by solving
\begin{equation}
	V/\sigma+ T(V)-C=0.
\end{equation}

\begin{algorithm}[htb]
	\centering
	\caption{An exact ADMM for solving (\ref{eq:De}).}
	\label{alg:eADMM}
	\begin{algorithmic}[1]
		\REQUIRE ~~\\ 
		Given parameters $\sigma > 0$ and $\pi \in (0, (1+\sqrt{5})/2)$;
		An initial point $(V^0,Z^0,\Omega^0)\in \mathbb{S}^{p}  \times \mathbb{S}^{p} \times \mathbb{S}^{p}$;
		An integer $k = 0$;
		\ENSURE ~~\\ 
		Approximate optimal solution $(\widehat{V}, \widehat{Z}, \widehat{\Omega})$;
		\STATE Calculate $\Lambda_1$ and $\Lambda_2$;
		\WHILE{Stopping criteria are not satisfied}
		\STATE Update $C^k = \Omega^k/\sigma+I_p-Z^k$;
		\STATE Compute $V^{k + 1}$ using formula (\ref{eq:eadmmmain});
		\STATE Compute
		$ Z^{k + 1} = \textup{Prox}_{\delta_{B_\lambda}}( \Omega^k/\sigma - T(V^{k + 1}) + I_p);$
		\STATE Compute
		$\Omega^{k+1} = \Omega^k - \pi \sigma \left( T(V^{k+1})+ Z^{k+1} - I_p \right)  $;
		\STATE Update $\widehat{Y} = Y^{k+1}, \,\widehat{Z} = Z^{k+1},\, \widehat{\Omega} = \Omega^{k+1}$;
		\STATE $k++$;
		\ENDWHILE
	\end{algorithmic}
\end{algorithm}

By applying the thin SVD to the sample covariance matrix, we can obtain $\mathcal{V}\in\mathbb{R}^{p\times n}$ and $\Lambda=\textup{diag}(\tau_1,\cdots,\tau_m)$ with $\tau_1,\cdots,\tau_m\geq 0$ such that $\widehat{\Sigma}=\mathcal{V}\Lambda \mathcal{V}^T$ and $\mathcal{V}^T\mathcal{V}=I_n$.
After calculating $\Lambda_1$ and $\Lambda_2$ by
\begin{align*}
	\Lambda_1&=\textup{diag}\left\{\frac{\tau_1}{\tau_1+2/\sigma},\cdots,\frac{\tau_m}{\tau_m+2/\sigma}\right\}, \\
	\Lambda_2&=\left\{\frac{\tau_i\tau_j(\tau_i+\tau_j+4/\sigma)}{(\tau_i+2/\sigma)(\tau_j+2/\sigma)(\tau_i+\tau_j+2/\sigma)}\right\}_{m\times m},
\end{align*}
we then have
\begin{equation}\label{eq:eadmmmain}
	\widebar V_e=\sigma
	\left( C- C\mathcal{V}\Lambda_1\mathcal{V}^T- \mathcal{V}\Lambda_1\mathcal{V}^TC+ \mathcal{V}(\Lambda_2 \circ (\mathcal{V}^TC\mathcal{V}))\mathcal{V}^T\right).
\end{equation}
Then, we give the detailed steps of eADMM in Algorithm \ref{alg:eADMM}.

\begin{table}[htbp]
	\centering
	\caption{Average performance among different methods for precision matrix estimation with 100 replications, p = 1000 in Models 1 to 4, p = 1024 in Model 5 and n = 400 for all the Models.}
	\scriptsize
	\resizebox{\textwidth}{!}{
		\begin{tabular}{lccccccc}
			\hline
			& Frobenius & Spectral & Infinity & TP    & TN    & $s_{\textup{off}}$ & $\bar s_{\textup{off}}$ \bigstrut\\
			\hline
			& mean $|$ sd & mean $|$ sd & mean $|$ sd & mean $|$ sd & mean $|$ sd & mean  & mean \bigstrut\\
			\hline
			\multicolumn{8}{l}{Model 1} \bigstrut\\
			\hline
			MARS & 10.9437 $|$ 0.0433 & 0.8193 $|$ 0.0119 & 1.0949 $|$ 0.0280 & 0.8134 $|$ 0.0069 & 0.9947 $|$ 0.0001 & 8353.36 & 8351.84 \bigstrut[t]\\
			SSNAL & 10.9469 $|$ 0.0438 & 0.8195 $|$ 0.0119 & 1.0947 $|$ 0.0280 & 0.8134 $|$ 0.0069 & 0.9947 $|$ 0.0001 & 8375.55 & 8350.06 \\
			iADMM & 10.9460 $|$ 0.0433 & 0.8195 $|$ 0.0119 & 1.0949 $|$ 0.0280 & 0.8148 $|$ 0.0070 & 0.9945 $|$ 0.0001 & 8528.64 & 8348.64 \\
			eADMM & 10.9460 $|$ 0.0433 & 0.8195 $|$ 0.0119 & 1.0949 $|$ 0.0280 & 0.8148 $|$ 0.0070 & 0.9945 $|$ 0.0001 & 8529.00 & 8349.04 \\
			scio  & 11.2131 $|$ 0.0424 & 0.8315 $|$ 0.0117 & 1.0782 $|$ 0.0237 & 0.7697 $|$ 0.0065 & 0.9964 $|$ 0.0001 & 6399.24 & 6396.24 \\
			glasso & 11.2632 $|$ 0.0346 & 0.8287 $|$ 0.0081 & 1.2692 $|$ 0.0374 & 0.8111 $|$ 0.0065 & 0.9903 $|$ 0.0002 & 12724.18 & 12717.21 \bigstrut[b]\\
			\hline
			\multicolumn{8}{l}{Model 2} \bigstrut\\
			\hline
			MARS & 17.2284 $|$ 0.0372 & 1.6560 $|$ 0.0094 & 1.9535 $|$ 0.0282 & 0.5399 $|$ 0.0052 & 0.9944 $|$ 0.0001 & 9390.04 & 9388.16 \bigstrut[t]\\
			SSNAL & 17.2320 $|$ 0.0372 & 1.6564 $|$ 0.0094 & 1.9534 $|$ 0.0282 & 0.5399 $|$ 0.0051 & 0.9944 $|$ 0.0001 & 9407.36 & 9382.32 \\
			iADMM & 17.2295 $|$ 0.0372 & 1.6561 $|$ 0.0094 & 1.9536 $|$ 0.0282 & 0.5427 $|$ 0.0050 & 0.9942 $|$ 0.0001 & 9624.22 & 9377.70 \\
			eADMM & 17.2295 $|$ 0.0372 & 1.6561 $|$ 0.0094 & 1.9536 $|$ 0.0282 & 0.5427 $|$ 0.0050 & 0.9942 $|$ 0.0001 & 9625.02 & 9377.74 \\
			scio  & 17.4139 $|$ 0.0363 & 1.6704 $|$ 0.0092 & 1.9402 $|$ 0.0236 & 0.4965 $|$ 0.0052 & 0.9962 $|$ 0.0001 & 7208.28 & 7204.52 \\
			glasso & 17.3500 $|$ 0.0300 & 1.6623 $|$ 0.0064 & 2.1461 $|$ 0.0480 & 0.5721 $|$ 0.0052 & 0.9883 $|$ 0.0003 & 15757.24 & 15748.44 \bigstrut[b]\\
			\hline
			\multicolumn{8}{l}{Model 3} \bigstrut\\
			\hline
			MARS & 12.6053 $|$ 0.0964 & 0.9441 $|$ 0.0117 & 1.0930 $|$ 0.0247 & 0.4480 $|$ 0.0351 & 0.9978 $|$ 0.0007 & 3461.28 & 3459.78 \bigstrut[t]\\
			SSNAL & 12.6059 $|$ 0.0961 & 0.9441 $|$ 0.0117 & 1.0928 $|$ 0.0247 & 0.4498 $|$ 0.0329 & 0.9978 $|$ 0.0007 & 3486.69 & 3473.66 \\
			iADMM & 12.6056 $|$ 0.0957 & 0.9442 $|$ 0.0117 & 1.0931 $|$ 0.0248 & 0.4508 $|$ 0.0332 & 0.9977 $|$ 0.0007 & 3528.22 & 3473.72 \\
			eADMM & 12.6057 $|$ 0.0959 & 0.9442 $|$ 0.0117 & 1.0931 $|$ 0.0248 & 0.4508 $|$ 0.0332 & 0.9977 $|$ 0.0007 & 3528.30 & 3473.78 \\
			scio  & 12.8151 $|$ 0.0246 & 0.9555 $|$ 0.0103 & 1.0715 $|$ 0.0142 & 0.3708 $|$ 0.0089 & 0.9991 $|$ 0.0001 & 1710.32 & 1709.18 \\
			glasso & 12.7205 $|$ 0.0376 & 0.9425 $|$ 0.0099 & 1.1469 $|$ 0.0319 & 0.4338 $|$ 0.0151 & 0.9976 $|$ 0.0005 & 3583.68 & 3581.24 \bigstrut[b]\\
			\hline
			\multicolumn{8}{l}{Model 4} \bigstrut\\
			\hline
			MARS & 7.7640 $|$ 0.0493 & 0.5496 $|$ 0.0115 & 0.7630 $|$ 0.0332 & 0.0064 $|$ 0.0001 & 0.9973 $|$ 0.0001 & 4289.52 & 4288.96 \bigstrut[t]\\
			SSNAL & 7.7651 $|$ 0.0493 & 0.5496 $|$ 0.0115 & 0.7629 $|$ 0.0331 & 0.0064 $|$ 0.0001 & 0.9973 $|$ 0.0001 & 4367.77 & 4347.22 \\
			iADMM & 7.7642 $|$ 0.0493 & 0.5497 $|$ 0.0115 & 0.7631 $|$ 0.0331 & 0.0065 $|$ 0.0001 & 0.9972 $|$ 0.0001 & 4414.12 & 4347.84 \\
			eADMM & 7.7642 $|$ 0.0493 & 0.5497 $|$ 0.0115 & 0.7631 $|$ 0.0331 & 0.0065 $|$ 0.0001 & 0.9972 $|$ 0.0001 & 4414.08 & 4347.86 \\
			scio  & 7.9572 $|$ 0.0492 & 0.5570 $|$ 0.0120 & 0.7479 $|$ 0.0267 & 0.0054 $|$ 0.0001 & 0.9982 $|$ 0.0001 & 3341.48 & 3339.84 \\
			glasso & 7.7360 $|$ 0.0516 & 0.5459 $|$ 0.0094 & 0.8877 $|$ 0.0399 & 0.0096 $|$ 0.0003 & 0.9943 $|$ 0.0003 & 7458.78 & 7454.63 \bigstrut[b]\\
			\hline
			\multicolumn{8}{l}{Model 5} \bigstrut\\
			\hline
			MARS & 8.0015 $|$ 0.0603 & 0.6061 $|$ 0.0145 & 0.9493 $|$ 0.0358 & 0.9893 $|$ 0.0021 & 0.9938 $|$ 0.0001 & 10389.18 & 10386.06 \bigstrut[t]\\
			SSNAL & 8.0058 $|$ 0.0603 & 0.6064 $|$ 0.0145 & 0.9493 $|$ 0.0358 & 0.9893 $|$ 0.0021 & 0.9932 $|$ 0.0001 & 11059.55 & 10378.32 \\
			iADMM & 8.0000 $|$ 0.0604 & 0.6060 $|$ 0.0145 & 0.9494 $|$ 0.0358 & 0.9904 $|$ 0.0019 & 0.9935 $|$ 0.0001 & 10662.36 & 10385.10 \\
			eADMM & 8.0002 $|$ 0.0604 & 0.6060 $|$ 0.0145 & 0.9494 $|$ 0.0358 & 0.9904 $|$ 0.0019 & 0.9935 $|$ 0.0001 & 10663.20 & 10384.46 \\
			scio  & 8.3156 $|$ 0.0613 & 0.6213 $|$ 0.0147 & 0.9294 $|$ 0.0325 & 0.9825 $|$ 0.0026 & 0.9963 $|$ 0.0001 & 7747.98 & 7745.16 \\
			glasso & 8.2708 $|$ 0.0559 & 0.6212 $|$ 0.0111 & 1.1149 $|$ 0.0560 & 0.9951 $|$ 0.0015 & 0.9881 $|$ 0.0002 & 16319.38 & 16310.56 \bigstrut[b]\\
			\hline
		\end{tabular}%
	}
	\label{tab:pfbig}%
\end{table}%

\subsection{Simulation studies}\label{sec:results}

{
In this subsection, we will present two group tests to illustrate the performance of MARS on some synthetic data generated by five different models. At first, we will demonstrate that the performance of the D-trace loss estimator is comparable to the negative log-likelihood-based estimator. Then, extensive numerical experiments were conducted to show the high efficiency and promising performance of MARS compared to some state-of-art solvers. Additionally, we will test MARS on some higher dimension datasets where existing solvers may fail to generate solutions or be significantly time-consuming.

In all the simulation studies, the data were generated from the following five different models:}
\begin{enumerate}
	\item $\Theta_{ij}=0.2, \,  \text{if $i\neq j$ and $1 \leq |i - j| \leq 2$};\, \Theta_{ii}=1$;\, $\Theta_{ij} = 0$ otherwise.
	\item $\Theta_{ij}=0.2, \,  \text{if $i\neq j$ and $1 \leq |i - j| \leq 4$};\, \Theta_{ii}=1$;\, $\Theta_{ij} = 0$ otherwise.
	\item $\Theta = \textup{diag}\{\Theta_0, \cdots, \Theta_0\}$ with $\Theta_0 \in \mathbb{S}^5$, the off-diagonal components are equal to $0.2$ and the diagonal is all $1$;  $\Theta_{ij} = 0$ otherwise.
	\item $\Theta_{ij} = 0.2 ^{|i - j|}$.
	\item $\Theta_{ij} = 0.2, \, \text{if the remainder after division of $i$ by $\sqrt{p}$ is not equal to 0 and $j = i + 1$}$; $\Theta_{ij} = 0.2, \, \text{if $j = i + \sqrt{p}$}$; $\Theta_{ii}=1$;\, $\Theta_{ij} = 0$ otherwise.
\end{enumerate}
We point out that Models 1, 2, and 5 are derived from \citep{Zhang2014}.

\begin{figure}[htbp]
	\centering
	\includegraphics[width=1\textwidth]{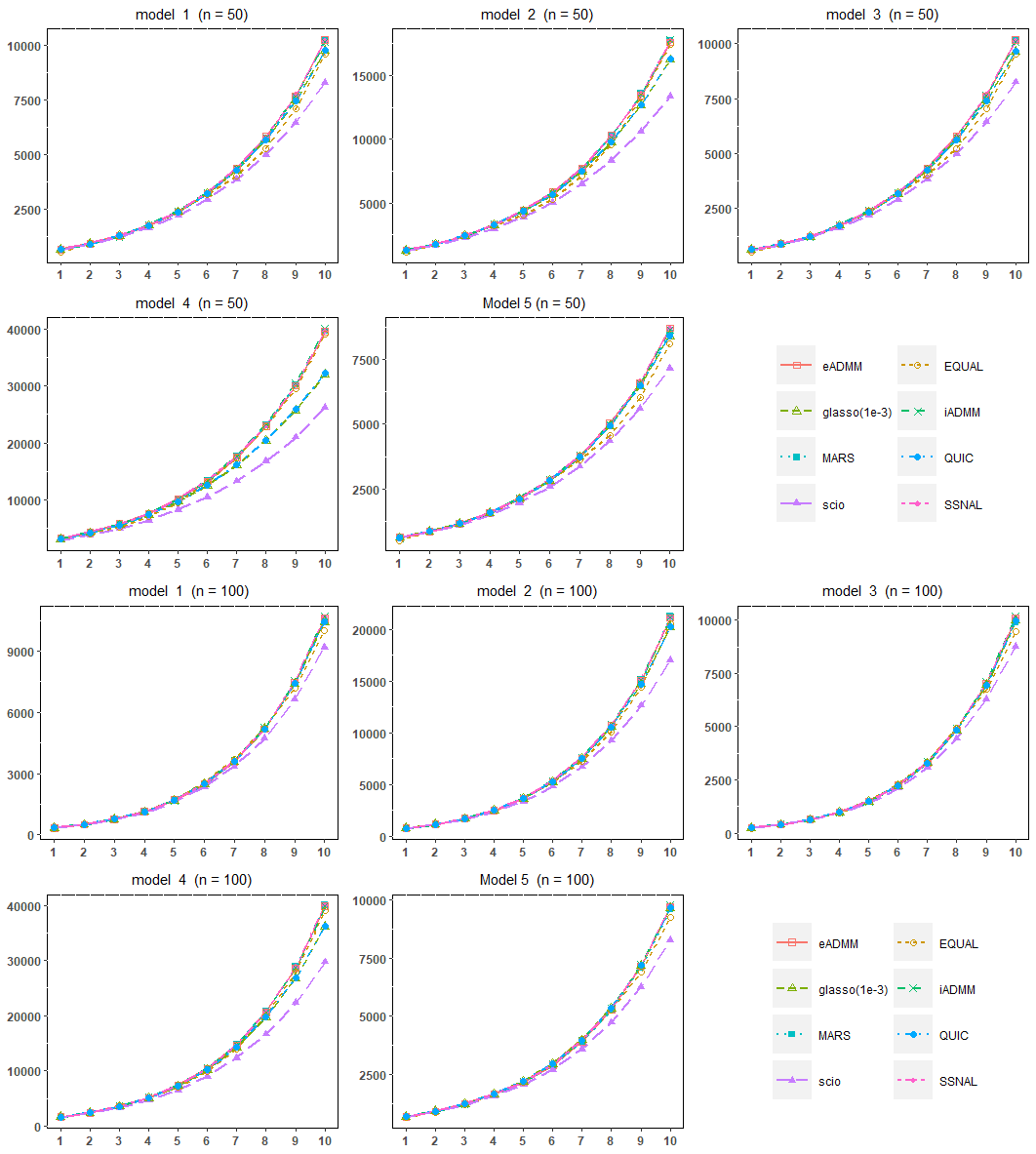}
	\caption{Average number of $\bar{s}_{\textup{off}}$ for $10$ different regularization parameters  with $10$ replications.  $p$ is set to $2000$ in Models 1-4, and $2025$ in Model 5. The horizontal axis is the $\lambda$ index, and the vertical axis is the number of off-diagonal elements whose absolute value is greater than $10^{-5}$.}
	\label{fig:perfbig}

\end{figure}

\begin{table}[htbp]
	\centering
	\caption{Average computation time (seconds) of different algorithms and 10 regularization parameters with 10 replications under Models 1-5 ($p=2000$ for Models 1-4 and $p=2025$ for Model 5). The times listed in this table are the total time to obtain estimated precision matrices for the 10 regularization parameters.}
	\small
	\resizebox{\textwidth}{!}{
		\begin{threeparttable}
			\centering
			\begin{tabular}{clccccc}
				\hline
				&       & \multicolumn{1}{c}{Model 1} & \multicolumn{1}{c}{Model 2} & \multicolumn{1}{c}{Model 3} & \multicolumn{1}{c}{Model 4} & \multicolumn{1}{c}{Model 5} \bigstrut[t]\\
				&       & \multicolumn{1}{c}{mean $|$ sd} & \multicolumn{1}{c}{mean $|$ sd} & \multicolumn{1}{c}{mean $|$ sd} & \multicolumn{1}{c}{mean $|$ sd} & \multicolumn{1}{c}{mean $|$ sd} \\
				\hline
				\multicolumn{7}{l}{problems solved sequentially (warm-started for each subsequent regularization parameter)} \bigstrut\\
				\midrule
				\multirow{8}[2]{*}{\begin{sideways}n $=$ 50\end{sideways}} & MARS  & 2.84 $ | $ 0.12 & 3.59 $ | $ 0.14 & 2.68 $ | $ 0.13 & 6.77 $ | $ 0.25 & 2.96 $ | $ 0.14 \\
				& SSNAL & 24.23 $ | $ 1.11 & 33.81 $ | $ 0.83 & 24.8 $ | $ 0.86 & 46.19 $ | $ 1.71 & 31.15 $ | $ 1.83 \\
				& iADMM & 153.66 $ | $ 1.44 & 157.07 $ | $ 1.37 & 155.58 $ | $ 1.26 & 166.24 $ | $ 1.11 & 166.51 $ | $ 3.98 \\
				& eADMM & 276.14 $ | $ 6.08 & 318.4 $ | $ 5.74 & 275.12 $ | $ 10.01 & 353.27 $ | $ 7.4 & 324.26 $ | $ 13.63 \\
				& scio$^{\#}$  & 45.88 $ | $ 0.62 & 46.54 $ | $ 0.31 & 45.86 $ | $ 0.7 & 228.93 $ | $ 12.07 & 48.03 $ | $ 0.7 \\
				& EQUAL$^{\#}$ & 84.13 $ | $ 1.26 & 102.79 $ | $ 2.03 & 83.38 $ | $ 0.98 & 133.43 $ | $ 3.23 & 86.36 $ | $ 2.05 \\
				& glasso(1e-3) & 51.13 $ | $ 0.97 & 80.53 $ | $ 5.12 & 50.87 $ | $ 0.85 & 176.14 $ | $ 5.11 & 45.53 $ | $ 1.04 \\
				& QUIC$^{\#}$  & 7.65 $ | $ 0.36 & 9.36 $ | $ 0.16 & 7.43 $ | $ 0.07 & 14.71 $ | $ 0.76 & 7.69 $ | $ 0.11 \\
				\midrule
				\multirow{8}[2]{*}{\begin{sideways}n $=$ 100\end{sideways}} & MARS  & 3.35 $ | $ 0.38 & 4.55 $ | $ 0.05 & 3.16 $ | $ 0.02 & 7.28 $ | $ 0.18 & 3.54 $ | $ 0.15 \\
				& SSNAL & 24.56 $ | $ 2.93 & 26.24 $ | $ 0.47 & 21.73 $ | $ 0.49 & 34.83 $ | $ 0.44 & 26.39 $ | $ 1.25 \\
				& iADMM & 116.87 $ | $ 9.35 & 121.41 $ | $ 1.26 & 110.4 $ | $ 1.13 & 124.92 $ | $ 0.99 & 126.65 $ | $ 0.94 \\
				& eADMM & 167.19 $ | $ 4.98 & 175.22 $ | $ 1.78 & 160.72 $ | $ 2.05 & 203.17 $ | $ 1.77 & 177.44 $ | $ 1.74 \\
				& scio$^{\#}$  & 47.39 $ | $ 5.21 & 46.37 $ | $ 0.74 & 45.76 $ | $ 0.79 & 48.01 $ | $ 0.46 & 47.72 $ | $ 0.63 \\
				& EQUAL$^{\#}$ & 42.69 $ | $ 0.98 & 50.81 $ | $ 0.49 & 40.07 $ | $ 0.21 & 61.68 $ | $ 0.35 & 48.2 $ | $ 0.57 \\
				& glasso(1e-3) & 46.17 $ | $ 3.46 & 95.2 $ | $ 7.35 & 41.54 $ | $ 0.89 & 150.71 $ | $ 7.67 & 50.27 $ | $ 4.3 \\
				& QUIC$^{\#}$  & 7.3 $ | $ 1.02 & 9.06 $ | $ 0.13 & 6.69 $ | $ 0.11 & 13.15 $ | $ 0.25 & 7.75 $ | $ 0.15 \\
				\midrule
				\multicolumn{7}{l}{problems solved independently (cold-started for each subsequent regularization parameter)} \bigstrut\\
				\midrule
				\multirow{8}[2]{*}{\begin{sideways}n $=$ 50\end{sideways}} & MARS  & 4.63 $|$ 0.26 & 6.37 $|$ 0.27 & 4.82 $|$ 0.19 & 11.76 $|$ 0.52 & 5.1 $|$ 0.21 \\
				& SSNAL & 38.6 $|$ 0.66 & 53.51 $|$ 0.62 & 44.08 $|$ 1.20 & 68.54 $|$ 0.53 & 43.69 $|$ 1.98 \\
				& iADMM & 298.54 $|$ 1.98 & 305.43 $|$ 3.28 & 298.8 $|$ 2.61 & 336.61 $|$ 3 & 319.75 $|$ 3.26 \\
				& eADMM & 528.30 $|$ 4.36 & 591.72 $|$ 8.14 & 616.26 $|$ 3.29 & 528.12 $|$ 6.2 & 592.57 $|$ 6.22 \\
				& scio$^{\#}$  & 46.01 $|$ 0.68 & 47.66 $|$ 0.94 & 46.12 $|$ 0.83 & 216.71 $|$ 9.61 & 47.88 $|$ 0.85 \\
				& EQUAL$^{\#}$ & 270.91 $|$ 1.61 & 268.12 $|$ 1.22 & 272.18 $|$ 1.46 & 287.32 $|$ 0.91 & 288.26 $|$ 1.66 \\
				& glasso(1e-3) & 148.47 $|$ 4.31 & 232.31 $|$ 2.98 & 145.68 $|$ 4.87 & 361.98 $|$ 7.52 & 131.26 $|$ 3.82 \\
				& QUIC$^{\#}$  & 11.04 $|$ 0.35 & 13.27 $|$ 0.43 & 10.88 $|$ 0.36 & 22.01 $|$ 1.85 & 11.79 $|$ 0.66 \\
				\midrule
				\multirow{8}[2]{*}{\begin{sideways}n $=$ 100\end{sideways}} & MARS  & 5.5 $|$ 0.11 & 7.74 $|$ 0.22 & 5.34 $|$ 0.15 & 12.42 $|$ 0.15 & 6.52 $|$ 0.17 \\
				& SSNAL & 32.43 $|$ 0.36 & 37.37 $|$ 0.57 & 30.31 $|$ 0.36 & 48.52 $|$ 0.52 & 40.85 $|$ 0.86 \\
				& iADMM & 240.56 $|$ 1.99 & 225.71 $|$ 2.31 & 240.81 $|$ 1.71 & 218.42 $|$ 1.03 & 249.7 $|$ 3.43 \\
				& eADMM & 381.09 $|$ 1.36 & 332.7 $|$ 3.37 & 364.37 $|$ 3.04 & 342.43 $|$ 2.79 & 380.1 $|$ 2.09 \\
				& scio$^{\#}$  & 46.24 $|$ 0.37 & 46.88 $|$ 0.79 & 46.27 $|$ 0.69 & 48.64 $|$ 0.38 & 49.11 $|$ 0.45 \\
				& EQUAL$^{\#}$ & 163.35 $|$ 1.69 & 164.52 $|$ 1.31 & 160.54 $|$ 1.16 & 166.48 $|$ 1.75 & 178.81 $|$ 0.7 \\
				& glasso(1e-3) & 128.6 $|$ 3.06 & 216.83 $|$ 2.76 & 123.16 $|$ 2.66 & 298.55 $|$ 3.15 & 143.83 $|$ 4.4 \\
				& QUIC$^{\#}$  & 8.55 $|$ 0.09 & 10.62 $|$ 0.15 & 8.42 $|$ 0.14 & 15.31 $|$ 0.14 & 9.59 $|$ 0.36 \\
				\bottomrule
			\end{tabular}%
			\begin{tablenotes}
				\item[1] The symbol ``$\#$" indicates that the average relative KKT residuals do not reach the stopping tolerance of $10^{-4}$. \item[2] The notation ``glasso(1e-3)" means an inputted stopping tolerance of $10^{-3}$, the associated solution is already reached the stopping tolerance of $10^{-4}$ for this test, see Table \ref{tab:kktsmall} in Appendix \ref{app:comperf} for details.
			\end{tablenotes}
		\end{threeparttable}
	}
	\label{tab:times}%
\end{table}%

\subsubsection{Statistical performance}
{ In this part, we will show the performance of the $\ell_1$ penalized D-trace estimator solved by MARS and others, as well as the performance of the graphical lasso estimator, to illustrate the comparable performance among them.}

For the last model, the sample dimension $p$ must satisfy that $\sqrt{p}$ is an integer. Thus, we set $p = 1000$ in Models $1$ to $4$, $p = 1024$ in Model $5$, and $n = 400$ for all the models.  In the test, the estimated precision matrix for each random sample is chosen by five-fold cross-validation. The test results conducted by 100 replications are shown in Table \ref{tab:pfbig}. The performance among different algorithms is compared in terms of seven quantities: the Frobenius norm (Frobenius), the spectral norm (Spectral) and the infinity norm (Infinity) between the estimated precision matrix and the true precision matrix, the ratio (TP) of correctly estimated non-zero components, the ratio (TN) of correctly estimated zero components, the number ($s_{\textup{off}}$) of the off-diagonal non-zero components and the number ($\bar s_{\textup{off}}$) of the off-diagonal components whose absolute values are greater than $10^{-5}$ in the estimated solution.

By comparing the first five quantities of the results in Table \ref{tab:pfbig}, we can see that the performance among the first four algorithms is similar to each other. The reason is that they are using the same estimator and stopping criteria. The performance of these four algorithms is slightly better than that of both scio and glasso in most cases. We point out that the main difference between our MARS and scio is that scio estimates the precision matrix column by column.
As for the sparsity of the estimated solutions, the results are quite different. Specifically, scio always generates more sparse solutions compared with others and the solutions obtained by glasso are less sparse than others for all the cases. MARS can generate more sparse solutions when compared with SSNAL and two kinds of ADMM,  and more importantly, it can generate solutions with fewer small-value components.
In other words, we do not need to artificially remove some components with insignificant values. In addition, all the estimated precision matrices in this test are nonsingular.

\begin{table}[htbp]
	\centering
	\caption{Average computation time (seconds) of different algorithms
		with 50 regularization parameters and 10 replications for generating a solution path.}
	\tiny
	\resizebox{\textwidth}{!}{
		\begin{threeparttable}
			\centering
			\begin{tabular}{clccccc}
				\hline
				&       & Model 1 & Model 2 & Model 3 & Model 4 & Model 5 \bigstrut[t]\\
				&       & mean $|$ sd & mean $|$ sd & mean $|$ sd & mean $|$ sd & mean $|$ sd \bigstrut[b]\\
				\hline
				\multicolumn{7}{l}{Models 1 to 4: p $=$ 2000;   Model 5: p $=$ 2025} \bigstrut\\
				\hline
				\multirow{10}[2]{*}{\begin{sideways}n$=$50\end{sideways}} & MARS(1e-4) & 6.78 $|$ 0.22 & 7.32 $|$ 0.18 & 6.39 $|$ 0.13 & 10.64 $|$ 0.37 & 7.61 $|$ 0.13 \bigstrut[t]\\
				& MARS(1e-8) & 13.59 $|$ 1.67 & 15.33 $|$ 0.67 & 11.93 $|$ 0.76 & 27.79 $|$ 1.57 & 15.95 $|$ 0.98 \\
				& SSNAL & 41.57 $|$ 1.97 & 45.24 $|$ 1.94 & 39.96 $|$ 1.41 & 66.01 $|$ 2.12 & 51.43 $|$ 1 \\
				& iADMM & 242.65 $|$ 6.83 & 260.9 $|$ 7.55 & 237.81 $|$ 6.2 & 297.62 $|$ 7.93 & 286.6 $|$ 6 \\
				& eADMM & 440.09 $|$ 6.28 & 439.84 $|$ 12.72 & 417.39 $|$ 12.85 & 581.59 $|$ 16.75 & 470.29 $|$ 14.46 \\
				& scio$^{\#}$  & 225.61 $|$ 1.88 & 226.54 $|$ 1.76 & 226.13 $|$ 1.75 & 390.91 $|$ 13.12 & 236.39 $|$ 2.17 \\
				& EQUAL$^{\#}$ & 103.76 $|$ 4.15 & 127.98 $|$ 0.75 & 102.78 $|$ 2.24 & 167.17 $|$ 1.95 & 117.45 $|$ 1.57 \\
				& glasso(1e-3) & 50.44 $|$ 0.39 & 78.67 $|$ 7.08 & 49.82 $|$ 0.73 & 145.97 $|$ 5.02 & 57.94 $|$ 1.1 \\
				& glasso(1e-4) & 92.04 $|$ 0.99 & 136.63 $|$ 1.68 & 90.7 $|$ 1.62 & 224.05 $|$ 1.71 & 107.26 $|$ 1.96 \\
				& QUIC$^{\#}$  & 24.95 $|$ 0.17 & 26.58 $|$ 0.14 & 24.87 $|$ 0.37 & 30.71 $|$ 0.89 & 27.56 $|$ 0.4 \bigstrut[b]\\
				\hline
				\multirow{10}[2]{*}{\begin{sideways}n$=$100\end{sideways}} & MARS(1e-4) & 7.17 $|$ 0.24 & 8.4 $|$ 0.36 & 6.98 $|$ 0.18 & 11.2 $|$ 0.32 & 8.11 $|$ 0.34 \bigstrut[t]\\
				& MARS(1e-8) & 14.19 $|$ 1.23 & 20.2 $|$ 1.8 & 13.92 $|$ 0.99 & 29.02 $|$ 1.21 & 17.39 $|$ 1.1 \\
				& SSNAL & 32.67 $|$ 1.83 & 38.01 $|$ 1.89 & 30.89 $|$ 1.46 & 45.21 $|$ 1.55 & 42.63 $|$ 1.67 \\
				& iADMM & 139.63 $|$ 4.03 & 152.38 $|$ 3.87 & 133.35 $|$ 2.85 & 165.33 $|$ 5.86 & 176.43 $|$ 3.7 \\
				& eADMM & 179.79 $|$ 7.63 & 208.62 $|$ 11.41 & 178.69 $|$ 7.98 & 228.48 $|$ 8.45 & 224.87 $|$ 4.31 \\
				& scio$^{\#}$  & 229.94 $|$ 2.49 & 228.34 $|$ 0.98 & 228.96 $|$ 1.39 & 228.04 $|$ 0.97 & 236.28 $|$ 1.18 \\
				& EQUAL$^{\#}$ & 48.54 $|$ 1.26 & 59.96 $|$ 0.91 & 48.03 $|$ 0.59 & 71.12 $|$ 1 & 55.8 $|$ 1.6 \\
				& glasso(1e-3) & 47.84 $|$ 7.57 & 89.34 $|$ 4.11 & 48.43 $|$ 8.67 & 122.91 $|$ 7.43 & 52.34 $|$ 7.29 \\
				& glasso(1e-4) & 72.72 $|$ 1.04 & 110.16 $|$ 1.62 & 69.49 $|$ 0.93 & 165.18 $|$ 7.7 & 77.69 $|$ 1.47 \\
				& QUIC$^{\#}$  & 24.39 $|$ 0.32 & 26 $|$ 0.12 & 24.21 $|$ 0.24 & 28.73 $|$ 0.26 & 26.39 $|$ 0.18 \bigstrut[b]\\
				\hline
				\multicolumn{7}{l}{Models 1 to 4: p $=$ 3000;   Model 5: p $=$ 3024} \bigstrut\\
				\hline
				\multirow{6}[2]{*}{\begin{sideways}n$=$50\end{sideways}} & MARS(1e-4) & 13.99 $|$ 0.77 & 16.4 $|$ 0.69 & 14.48 $|$ 0.65 & 20.47 $|$ 0.74 & 15.47 $|$ 0.63 \bigstrut[t]\\
				& MARS(1e-8) & 26.58 $|$ 2.74 & 36.72 $|$ 2.12 & 26.99 $|$ 2.37 & 50.5 $|$ 3.07 & 29.64 $|$ 2.15 \\
				& SSNAL & 91.14 $|$ 3.68 & 117.92 $|$ 4.09 & 93.07 $|$ 2.54 & 144.3 $|$ 4.33 & 104.6 $|$ 4.63 \\
				& EQUAL$^{\#}$ & 258.29 $|$ 2.85 & 344.63 $|$ 5.24 & 261.65 $|$ 13.18 & 432.73 $|$ 10.04 & 272.97 $|$ 3.48 \\
				& glasso(1e-3) & 146.51 $|$ 2.3 & 260.34 $|$ 2.5 & 145.85 $|$ 2.74 & 419.1 $|$ 22.35 & 164.27 $|$ 2.69 \\
				& glasso(1e-4) & 271.12 $|$ 6.34 & 497.62 $|$ 5.07 & 269.1 $|$ 4.9 & 691.6 $|$ 9.53 & 306.98 $|$ 5.72 \bigstrut[b]\\
				\hline
				\multirow{6}[2]{*}{\begin{sideways}n$=$100\end{sideways}} & MARS(1e-4) & 15.24 $|$ 0.74 & 17.42 $|$ 0.87 & 15.36 $|$ 0.72 & 19.41 $|$ 0.82 & 17.13 $|$ 1 \bigstrut[t]\\
				& MARS(1e-8) & 30.68 $|$ 2.69 & 39.1 $|$ 3.69 & 29.4 $|$ 2.2 & 45.08 $|$ 2.7 & 33 $|$ 2.98 \\
				& SSNAL & 71.42 $|$ 3.65 & 87.63 $|$ 2.6 & 69.48 $|$ 1.64 & 96.58 $|$ 4.8 & 96.36 $|$ 4.02 \\
				& EQUAL$^{\#}$ & 133.33 $|$ 0.7 & 159.65 $|$ 4.39 & 132.13 $|$ 1.04 & 183.01 $|$ 0.96 & 146.28 $|$ 0.94 \\
				& glasso(1e-3) & 133.15 $|$ 1.41 & 247.69 $|$ 1.4 & 127.63 $|$ 1.04 & 325.16 $|$ 1.75 & 140.07 $|$ 4.02 \\
				& glasso(1e-4) & 244.54 $|$ 2.88 & 374.9 $|$ 4.36 & 233.46 $|$ 1.68 & 481.69 $|$ 3.39 & 255.37 $|$ 7.79 \bigstrut[b]\\
				\hline
				\multicolumn{7}{l}{Models 1 to 4: p $=$ 5000;   Model 5: p $=$ 5041} \bigstrut\\
				\hline
				\multirow{3}[2]{*}{\begin{sideways}n$=$50\end{sideways}} & MARS(1e-4) & 22.14 $|$ 0.68 & 28.17 $|$ 1.15 & 22.53 $|$ 0.85 & 35.7 $|$ 1.87 & 24.51 $|$ 1.79 \bigstrut[t]\\
				& MARS(1e-8) & 55.4 $|$ 7.57 & 75.15 $|$ 8.43 & 59.67 $|$ 5.96 & 107.86 $|$ 9.6 & 63.66 $|$ 6.9 \\
				& SSNAL & 218.99 $|$ 5.15 & 299.27 $|$ 7.52 & 221.99 $|$ 10.28 & 366.03 $|$ 5.89 & 233.71 $|$ 6.59 \bigstrut[b]\\
				\hline
				\multirow{3}[2]{*}{\begin{sideways}n$=$100\end{sideways}} & MARS(1e-4) & 18.66 $|$ 1.17 & 23.71 $|$ 1.76 & 20.02 $|$ 1.17 & 33.08 $|$ 0.95 & 25.15 $|$ 1.22 \bigstrut[t]\\
				& MARS(1e-8) & 50.02 $|$ 5.83 & 72.22 $|$ 5.86 & 59.36 $|$ 7.34 & 102.32 $|$ 6.29 & 66.58 $|$ 5.79 \\
				& SSNAL & 147.6 $|$ 5.19 & 179.6 $|$ 9.08 & 156.94 $|$ 6.73 & 231.67 $|$ 7.94 & 207.09 $|$ 5.12 \bigstrut[t]\\
				\hline
			\end{tabular}%
			\begin{tablenotes}
				\item[1] The symbol ``$\#$" indicates that the average relative KKT residuals do not reach the stopping tolerance of $10^{-4}$. \item[2]  The notation ``glasso(1e-3)" means an inputted stopping tolerance of $10^{-3}$, the associated solution is already reached the stopping tolerance of $10^{-4}$ for this test, see Table \ref{tab:pathkkt50} in Appendix \ref{app:comperf} for details. We also show results for MARS and glasso with inputted stopping tolerances of $10^{-8}$ and $10^{-4}$ respectively, for obtaining solutions with relative KKT residuals smaller than $10^{-8}$.
				\item[3] In the last test, due to out of memory of glasso and less efficiency of EQUAL, these two solvers were not tested.
			\end{tablenotes}
		\end{threeparttable}
	}
	\label{tab:pathtimes}%
\end{table}%

\subsubsection{Computational performance}\label{sec:numexpcom}

{ In this part,  we will demonstrate the computational performance of MARS compared with some state-of-art solvers for both a solution path  and also some single regularization problems. Furthermore, we will test our MARS on some higher dimension datasets to characterize the ability of MARS to handle higher size data.}
	
{For simplicity, here we still use the same five models to generate random data. We set $p = 2000$ in Models  1-4, $p = 2025$ in Model $5$, and set $n$ to $50$ and $100$.
Then, we use $10$ equally decreasing regularization parameters to demonstrate the performance of different algorithms. The regularization parameters are chosen by a pretest such that the oracle sparsity of the precision matrix lies in the sparsity of those ten generated solutions, and the sparsity levels of the solutions from different algorithms are close. Specifically, we start the test with  $\lambda_{\max}$ and decrease it by $0.01$ until the sparsity of the associated solution is less than the oracle sparsity, and then use the 10 smallest regularization parameters for the following tests. It can be seen from Figure \ref{fig:perfbig} that under the same regularization parameter range, the estimated solutions obtained by different algorithms are similar in sparsity, except for scio, which tends to get more sparse solutions.}
	
{For the test on Model 4 with $n=50$ and $p=2000$, the sparsities of the solutions for some smaller regularization parameters obtained by glasso and QUIC are slightly different from that of MARS, SNNAL, and the two ADMMs. However, the difference in quantity is not large, see Figure \ref{fig:perfbig} for details. The corresponding mean and standard deviation (sd) of the computation time for solving a solution path with the $10$ parameters are shown in Table \ref{tab:times} (the relative KKT residuals $(\eta)$ are listed in Table \ref{tab:kktsmall} in Appendix \ref{app:comperf}). We can see that MARS can be up to $25$ times faster than glasso and $10$ times faster than SSNAL. Besides, for each $\lambda$, we list the statistical performance of MARS and glasso in Table \ref{tab:errors} in Appendix \ref{app:comperf}. To compare the computation efficiency of different methods with a single regularization parameter, we also test the problems under the 10 regularization parameters without using any prior information (cold-started). The average computation times are also reported in Table \ref{tab:times}. It can be seen that MARS can be up to $32$ times faster than glasso and $9$ times faster than SSNAL. From the results above, we can conclude that MARS is significantly more efficient than others both in  generating a solution path and in solving the problem with a single regularization parameter.

Next, we conduct tests for generating a solution path with $50$ regularization parameters. The regularization parameter interval is 50 equally decreasing values from $1$ to $\lambda_{ \textup{min}}$, which is the same as in the previous test. After observing the test results in Table \ref{tab:pathtimes}, we found that MARS, SSNAL, glasso, and EQUAL are much more efficient than other algorithms,  so when $p$ is set to $3000$ in Models $1$ to $4$, and $3025$ in Model $5$, we will only focus on the comparison among these four algorithms. Besides, due to the out-of-memory error of glasso and less efficiency of EQUAL, we will only test MARS and SSNAL in subsequent higher dimension tests. For fairness of comparison, in all tests where $p$ is less than $5000$, we did not use Remark \ref{rm:maxl} to narrow the range of $\lambda$ path for MARS, SSNAL, iADMM, and eADMM and the generated data are all standardized in the first beginning. It can be seen from Table \ref{tab:pathtimes} that MARS is significantly faster than other algorithms. Especially, when $p$ is larger, the computation time of MARS could be roughly $1/10$ of that of glasso and SSNAL.

}

\begin{sidewaystable}[htbp]
	\centering
	\caption{The objective values and the relative KKT residuals ($\eta$) of paths of estimated precision matrices generated by different algorithms for the 		prostate data set.}
	\tiny
	\begin{threeparttable}
		\begin{tabular}{c|c|c|c|c}
			\hline
			\multirow{3}[6]{*}{$\lambda$} & \multicolumn{2}{c|}{control group} & \multicolumn{2}{c}{cancer group} \bigstrut\\
			\cline{2-5}          & Objective value & $\eta$ & Objective value & $\eta$ \bigstrut\\
			\cline{2-5}          & MARS $|$ SSNAL $|$ EQUAL $|$ scio & MARS $|$ SSNAL $|$ EQUAL $|$ scio & MARS $|$ SSNAL $|$ EQUAL $|$ scio & MARS $|$ SSNAL $|$ EQUAL $|$ scio \bigstrut\\
			\hline
			0.99  & -3016.50  $|$ -3016.50  $|$ 924.68  $|$ -3016.50  & 6.47e-06 $|$ 4.80e-06 $|$ 5.05e-02 $|$ 2.05e-07 & -3016.50 $|$ -3016.50 $|$ -1964.80 $|$ -3016.50 & 9.36e-06 $|$ 6.95e-06 $|$ 3.40e-02 $|$ 2.33e-07 \bigstrut[t]\\
			0.98  & -3016.50  $|$ -3016.50  $|$ -861.70  $|$ -3016.53  & 3.31e-05 $|$ 3.15e-05 $|$ 3.89e-01 $|$ 3.44e-07 & -3016.50 $|$ -3016.50 $|$ -1966.85 $|$ -3016.56 & 4.33e-05 $|$ 4.11e-05 $|$ 3.41e-02 $|$ 4.69e-07 \\
			0.97  & -3016.50  $|$ -3016.50  $|$ -1798.68  $|$ -3016.63  & 7.56e-05 $|$ 7.43e-05 $|$ 3.54e-02 $|$ 5.44e-07 & -3016.50 $|$ -3016.50 $|$ -2448.14 $|$ -3016.72 & 9.84e-05 $|$ 9.67e-05 $|$ 6.39e-02 $|$ 6.86e-07 \\
			0.96  & -3016.83  $|$ -3016.52  $|$ -1802.12  $|$ -3016.83  & 4.86e-06 $|$ 9.40e-05 $|$ 3.56e-02 $|$ 7.04e-07 & -3017.04 $|$ -3016.68 $|$ -2447.90 $|$ -3017.05 & 7.30e-06 $|$ 9.43e-05 $|$ 6.40e-02 $|$ 7.99e-07 \\
			0.95  & -3017.10  $|$ -3016.94  $|$ -2351.89  $|$ -3017.19  & 9.54e-05 $|$ 7.80e-05 $|$ 7.22e-02 $|$ 8.51e-07 & -3017.58 $|$ -3017.24 $|$ -2456.80 $|$ -3017.59 & 1.04e-05 $|$ 8.73e-05 $|$ 6.31e-02 $|$ 1.01e-06 \\
			0.94  & -3017.75  $|$ -3017.53  $|$ -2356.68  $|$ -3017.75  & 3.39e-06 $|$ 7.18e-05 $|$ 7.18e-02 $|$ 1.03e-06 & -3018.40 $|$ -3018.06 $|$ -2700.02 $|$ -3018.41 & 1.15e-05 $|$ 8.54e-05 $|$ 2.20e-02 $|$ 1.21e-06 \\
			0.93  & -3018.55  $|$ -3018.22  $|$ -2371.45  $|$ -3018.55  & 7.20e-06 $|$ 9.92e-05 $|$ 7.04e-02 $|$ 5.63e-05 & -3019.57 $|$ -3019.23 $|$ -2698.98 $|$ -3019.58 & 1.16e-05 $|$ 8.50e-05 $|$ 2.22e-02 $|$ 1.35e-06 \\
			0.92  & -3019.66  $|$ -3019.39  $|$ -2645.23  $|$ -3019.66  & 8.62e-06 $|$ 8.12e-05 $|$ 2.33e-02 $|$ 1.21e-04 & -3021.14 $|$ -3020.79 $|$ -2704.31 $|$ -3021.15 & 1.17e-05 $|$ 8.53e-05 $|$ 2.22e-02 $|$ 1.49e-06 \\
			0.91  & -3021.13  $|$ -3020.90  $|$ -2643.28  $|$ -3021.13  & 7.98e-06 $|$ 7.48e-05 $|$ 2.35e-02 $|$ 1.82e-04 & -3023.16 $|$ -3022.80 $|$ -2858.82 $|$ -3023.17 & 1.24e-05 $|$ 8.55e-05 $|$ 2.50e-02 $|$ 1.58e-06 \\
			0.9   & -3023.02  $|$ -3022.80  $|$ -2648.98  $|$ -3023.01  & 8.74e-06 $|$ 7.22e-05 $|$ 2.35e-02 $|$ 2.47e-04 & -3025.69 $|$ -3025.32 $|$ -2858.34 $|$ -3025.70 & 1.17e-05 $|$ 8.59e-05 $|$ 2.54e-02 $|$ 1.70e-06 \\
			0.89  & -3025.39  $|$ -3025.18  $|$ -2832.42  $|$ -3025.37  & 8.11e-06 $|$ 7.04e-05 $|$ 2.70e-02 $|$ 3.28e-04 & -3028.77 $|$ -3028.40 $|$ -2860.17 $|$ -3028.78 & 1.16e-05 $|$ 8.60e-05 $|$ 2.55e-02 $|$ 1.04e-05 \\
			0.88  & -3028.30  $|$ -3028.09  $|$ -2833.94  $|$ -3028.27  & 8.69e-06 $|$ 6.95e-05 $|$ 2.73e-02 $|$ 4.00e-04 & -3032.46 $|$ -3032.08 $|$ -2863.98 $|$ -3032.47 & 1.22e-05 $|$ 8.63e-05 $|$ 2.56e-02 $|$ 1.13e-04 \\
			0.87  & -3031.79  $|$ -3031.60  $|$ -2837.75  $|$ -3031.75  & 1.51e-05 $|$ 6.88e-05 $|$ 2.74e-02 $|$ 4.81e-04 & -3036.81 $|$ -3036.42 $|$ -2931.86 $|$ -3036.77 & 1.17e-05 $|$ 8.69e-05 $|$ 1.40e-02 $|$ 4.20e-04 \\
			0.86  & -3035.94  $|$ -3035.74  $|$ -2912.43  $|$ -3035.86  & 1.12e-05 $|$ 6.84e-05 $|$ 1.49e-02 $|$ 6.08e-04 & -3041.88 $|$ -3041.49 $|$ -2939.44 $|$ -3041.71 & 1.17e-05 $|$ 8.67e-05 $|$ 1.39e-02 $|$ 7.96e-04 \\
			0.85  & -3040.77  $|$ -3040.58  $|$ -2918.38  $|$ -3040.61  & 1.57e-05 $|$ 6.79e-05 $|$ 1.49e-02 $|$ 8.34e-04 & -3047.73 $|$ -3047.34 $|$ -2995.49 $|$ -3046.91 & 1.30e-05 $|$ 8.65e-05 $|$ 1.05e-02 $|$ 1.68e-03 \\
			0.84  & -3046.36  $|$ -3046.16  $|$ -2989.64  $|$ -3046.13  & 9.47e-06 $|$ 6.76e-05 $|$ 1.12e-02 $|$ 9.71e-04 & -3054.41 $|$ -3054.01 $|$ -3002.93 $|$ -3050.94 & 1.19e-05 $|$ 8.64e-05 $|$ 1.08e-02 $|$ 3.73e-03 \\
			0.83  & -3052.73  $|$ -3052.54  $|$ -2993.46  $|$ -3051.80  & 1.60e-05 $|$ 6.73e-05 $|$ 1.20e-02 $|$ 1.81e-03 & -3061.97 $|$ -3061.57 $|$ -3008.80 $|$ -3057.54 & 1.15e-05 $|$ 8.62e-05 $|$ 1.14e-02 $|$ 3.70e-03 \\
			0.82  & -3059.95  $|$ -3059.77  $|$ -2997.82  $|$ -2992.17  & 1.74e-05 $|$ 6.70e-05 $|$ 1.30e-02 $|$ 5.80e-02 & -3070.43 $|$ -3070.04 $|$ -3020.96 $|$ -3015.15 & 1.97e-05 $|$ 8.61e-05 $|$ 9.59e-03 $|$ 3.61e-02 \\
			0.81  & -3068.06  $|$ -3067.88  $|$ -3014.42  $|$ -806.76  & 1.82e-05 $|$ 6.67e-05 $|$ 1.00e-02 $|$ 4.25e-01 & -3079.89 $|$ -3079.48 $|$ -3040.97 $|$ -2534.57 & 1.35e-05 $|$ 8.59e-05 $|$ 7.75e-03 $|$ 1.47e-01 \\
			0.8   & -3077.11  $|$ -3076.93  $|$ -3057.88  $|$ 1170.63  & 1.84e-05 $|$ 6.66e-05 $|$ 5.86e-03 $|$ 4.60e-01 & -3090.34 $|$ -3089.94 $|$ -3067.80 $|$ -464.84 & 2.07e-05 $|$ 8.56e-05 $|$ 5.64e-03 $|$ 2.49e-01 \\
			0.79  & -3087.13  $|$ -3086.95  $|$ -3069.02  $|$ 4.49e+03 & 1.90e-05 $|$ 6.62e-05 $|$ 6.06e-03 $|$ 4.36e-01 & -3101.87 $|$ -3101.46 $|$ -3081.18 $|$ 9861.79 & 1.38e-05 $|$ 8.55e-05 $|$ 5.79e-03 $|$ 3.60e-01 \\
			0.78  & -3098.15  $|$ -3097.98  $|$ -3066.00  $|$ 1.24e+04 & 1.86e-05 $|$ 6.59e-05 $|$ 1.06e-02 $|$ 4.22e-01 & -3114.48 $|$ -3114.08 $|$ -3084.28 $|$ 6.00e+04 & 2.51e-05 $|$ 8.52e-05 $|$ 8.97e-03 $|$ 4.75e-01 \\
			0.77  & -3110.23  $|$ -3110.05  $|$ -3085.58  $|$ 3.13e+04 & 1.85e-05 $|$ 6.57e-05 $|$ 7.64e-03 $|$ 4.21e-01 & -3128.26 $|$ -3127.84 $|$ -3102.08 $|$ 1.73e+05 & 3.06e-05 $|$ 8.50e-05 $|$ 7.10e-03 $|$ 5.18e-01 \\
			0.76  & -3123.39  $|$ -3123.21  $|$ -3106.15  $|$ 7.23e+04 & 1.82e-05 $|$ 6.54e-05 $|$ 5.13e-03 $|$ 4.29e-01 & -3143.20 $|$ -3142.77 $|$ -3124.53 $|$ 3.72e+05 & 2.91e-05 $|$ 8.46e-05 $|$ 4.94e-03 $|$ 5.32e-01 \\
			0.75  & -3137.67  $|$ -3137.49  $|$ -3121.64  $|$ 1.54e+05 & 3.02e-05 $|$ 6.55e-05 $|$ 4.35e-03 $|$ 4.38e-01 & -3159.40 $|$ -3158.91 $|$ -3144.88 $|$ 6.84e+05 & 2.53e-05 $|$ 9.04e-05 $|$ 3.75e-03 $|$ 5.34e-01 \\
			0.74  & -3153.12  $|$ -3152.93  $|$ -3139.48  $|$ 3.12e+05 & 1.95e-05 $|$ 6.48e-05 $|$ 4.26e-03 $|$ 4.52e-01 & -3176.87 $|$ -3176.32 $|$ -3164.75 $|$ 1.15e+06 & 2.60e-05 $|$ 8.64e-05 $|$ 3.75e-03 $|$ 5.32e-01 \\
			0.73  & -3169.78  $|$ -3169.57  $|$ -3144.13  $|$ 5.54e+05 & 3.63e-05 $|$ 6.50e-05 $|$ 8.05e-03 $|$ 4.58e-01 & -3195.68 $|$ -3195.04 $|$ -3172.27 $|$ 1.82e+06 & 4.81e-05 $|$ 8.83e-05 $|$ 7.52e-03 $|$ 5.29e-01 \\
			0.72  & -3187.73  $|$ -3187.46  $|$ -3162.36  $|$ 9.22e+05 & 4.16e-05 $|$ 6.75e-05 $|$ 6.12e-03 $|$ 4.62e-01 & -3215.87 $|$ -3215.12 $|$ -3190.12 $|$ 2.77e+06 & 5.54e-05 $|$ 9.05e-05 $|$ 6.33e-03 $|$ 5.25e-01 \\
			0.71  & -3206.96  $|$ -3206.65  $|$ -3188.84  $|$ 1.55e+06 & 2.69e-05 $|$ 6.93e-05 $|$ 4.08e-03 $|$ 4.72e-01 & -3237.51 $|$ -3236.59 $|$ -3220.13 $|$ 4.06e+06 & 4.12e-05 $|$ 9.38e-05 $|$ 3.79e-03 $|$ 5.21e-01 \\
			0.7   & -3227.61  $|$ -3227.17  $|$ -3204.06  $|$ 2.58e+06 & 6.86e-05 $|$ 7.28e-05 $|$ 6.15e-03 $|$ 4.85e-01 & -3260.82 $|$ -3259.52 $|$ -3244.18 $|$ 5.76e+06 & 2.41e-05 $|$ 9.77e-05 $|$ 3.82e-03 $|$ 5.15e-01 \\
			0.69  & -3249.70  $|$ -3249.08  $|$ -3229.79  $|$ 4.01e+06 & 6.15e-05 $|$ 7.72e-05 $|$ 3.71e-03 $|$ 4.93e-01 & -3285.61 $|$ -3284.16 $|$ -3271.68 $|$ 8.07e+06 & 3.79e-05 $|$ 8.62e-05 $|$ 3.70e-03 $|$ 5.11e-01 \\
			0.68  & -3273.48  $|$ -3272.43  $|$ -3257.15  $|$ 6.07e+06 & 2.37e-05 $|$ 8.30e-05 $|$ 3.59e-03 $|$ 4.99e-01 & -3312.60 $|$ -3310.34 $|$ -3285.98 $|$ 1.13e+07 & 3.12e-05 $|$ 8.70e-05 $|$ 7.49e-03 $|$ 5.08e-01 \\
			0.67  & -3298.90  $|$ -3297.28  $|$ -3289.17  $|$ 9.13e+06 & 2.45e-05 $|$ 8.99e-05 $|$ 3.31e-03 $|$ 5.07e-01 & -3342.17 $|$ -3338.18 $|$ -3310.62 $|$ 1.56e+07 & 3.54e-05 $|$ 9.14e-05 $|$ 6.34e-03 $|$ 5.07e-01 \\
			0.66  & -3326.19  $|$ -3323.70  $|$ -3310.93  $|$ 1.38e+07 & 2.91e-05 $|$ 9.78e-05 $|$ 5.40e-03 $|$ 5.18e-01 & -3374.62 $|$ -3367.78 $|$ -3349.71 $|$ 2.16e+07 & 3.83e-05 $|$ 9.92e-05 $|$ 3.58e-03 $|$ 5.06e-01 \\
			0.65  & -3355.73  $|$ -3352.01  $|$ -3330.89  $|$ 2.05e+07 & 3.15e-05 $|$ 9.37e-05 $|$ 6.44e-03 $|$ 5.28e-01 & -3410.80 $|$ -3400.77 $|$ -3370.41 $|$ 2.97e+07 & 4.81e-05 $|$ 9.16e-05 $|$ 7.03e-03 $|$ 5.07e-01 \\
			0.64  & -3388.70  $|$ -3382.07  $|$ -3370.90  $|$ 3.08e+07 & 3.97e-05 $|$ 9.88e-05 $|$ 3.53e-03 $|$ 5.40e-01 & -3453.01 $|$ -3436.03 $|$ -3424.12 $|$ 4.05e+07 & 5.28e-05 $|$ 9.07e-05 $|$ 2.13e-03 $|$ 5.08e-01 \\
			0.63  & -3426.33 $|$ -3415.49 $|$ -3399.63 $|$ $-$ & 4.39e-05 $|$ 8.55e-05 $|$ 5.77e-03 $|$ $-$ & -3501.98 $|$ -3474.07 $|$ -3445.53 $|$ $-$ & 6.00e-05 $|$ 9.92e-05 $|$ 7.45e-03 $|$ $-$ \\
			0.62  & -3469.25 $|$ -3449.88 $|$ -3429.19 $|$ $-$ & 5.38e-05 $|$ 9.86e-05 $|$ 5.79e-03 $|$ $-$ & -3556.03 $|$ -3519.00 $|$ -3475.93 $|$ $-$ & 6.72e-05 $|$ 9.98e-05 $|$ 6.86e-03 $|$ $-$ \bigstrut[b]\\
			\hline
		\end{tabular}%
		\begin{tablenotes}
			\item The symbol ``$-$" indicates out of memory.
		\end{tablenotes}
	\end{threeparttable}
	\label{tab:realsol}%
\end{sidewaystable}%

\subsection{Real data analysis}\label{NE:real}

In this subsection, we will use some real data sets to demonstrate the promising performance of our MARS for generating a solution path. The publicly available data sets we are going to use include a prostate data set (\href{https://web.stanford.edu/~hastie/CASI_files/DATA/prostate.html}{\url{https://web.stanford.edu/~hastie/CASI_files/DATA/prostate.html}}) and a breast cancer data set \citep{Hess2006}, which can be found on (\href{https://bioinformatics.mdanderson.org/public-datasets/}{\url{https://bioinformatics.mdanderson.org/public-datasets/}}). The prostate data set contains two groups, the first one is $6033$ genetic activity measurements of $50$ control subjects and the other is that of $52$ prostate cancer subjects. Thus, the number of variables contained in the precision matrix that needs to be estimated is more than 18 million. As for the breast cancer data set, it contains the measurements of $22283$ genes with $133$ subjects, where $99$ of them are labeled as residual disease (RD) and the remaining $34$ subjects are labeled as pathological complete response (pCR). For this data set, the estimated precision matrix contains about 250 million parameters.

After standardizing the two groups of the prostate data set, we use MARS, SSNAL, EQUAL, and scio to generate solution paths for the two groups separately. We should note that, when $\lambda$ is too small, there may not exist optimal solutions for the precision matrix estimator. Therefore, before going further to the main comparison tests, we should conduct some pretests to find a suitable smallest $\lambda$. The performance of the estimations generated by different algorithms is concluded in Table \ref{tab:realsol}. Since $\eta$ is to measure the accuracy of the generated solution, we notice that when $\lambda$ is larger, the estimated solutions generated by scio perform very well, but when $\lambda$ gradually becomes smaller, its estimated solutions are not desired, which can also be observed from the associated objective value. The performance of EQUAL is the opposite, that is, it performs better when $\lambda$ is smaller. We should point out that even if the stopping tolerance of EQUAL has been set to be $10^{-6}$, none of the generated solutions makes $\eta$ less than $10^{-3}$.
Thus, by comparing the objective value and $\eta$, we conclude that MARS and SSNAL can outperform both EQUAL and scio since all the $\eta$ are smaller than the set tolerance $10^{-4}$. Although both MARS and SSNAL can generate satisfactory solutions, from Table \ref{tab:tmprost}, we find that MARS is much more efficient. In particular,  the computation time of SSNAL to generate the solution path is more than 14 times that of MARS in the Control group and more than 18 times that of MARS in the Cancer group. This can also be seen in Figure \ref{fig:tmpath}, which illustrates that MARS has high efficiency in generating solutions for each $\lambda$. Besides, we obtain the final precision matrix estimations of the two different groups through 5-fold cross-validation, and the corresponding graphs are shown in Figure \ref{fig:pcsiginv}. From this figure, we can clearly see that the genes of the control group and the cancer group have different connections.

\begin{table}[htbp]
	\centering
	\caption{The computation time (seconds) of different algorithms for generating a solution path with the prostate data set.}
	\small
	\begin{threeparttable}
		\begin{tabular}{ccccc}
			\hline
			& MARS  & SSNAL & EQUAL & scio \bigstrut\\
			\hline
			control group & 113.6  & 1629.88  & 1325.67*  & 5690.73+ \bigstrut[t]\\
			cancer group & 112.43  & 2108.06  & 1273.22*  & 5949.71+ \bigstrut[b]\\
			\hline
		\end{tabular}%
		\begin{tablenotes}
			\item[1.] The symbol ``$*$" indicates that none of the relative KKT residuals of EQUAL is less than $10^{-3}$.
			\item[2.] The symbol ``$+$" indicates that, due to out of memory, the time here does not include the time for generating estimations by scio with the two smallest $\lambda$.
		\end{tablenotes}
	\end{threeparttable}
	\label{tab:tmprost}%
\end{table}%

\begin{figure}[htbp]
	\centering
	\includegraphics[width=0.7\textwidth]{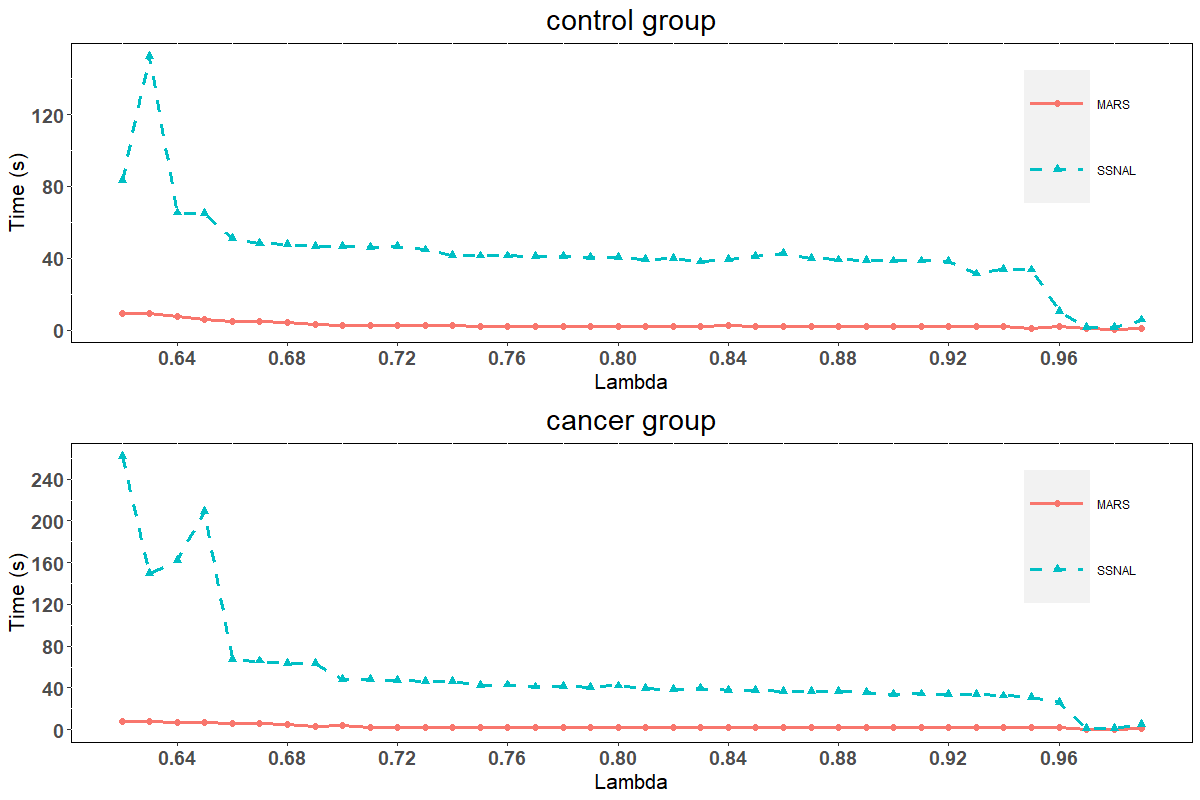}
	\caption{The computation time of MARS and SSNAL for each $\lambda$ with the prostate data set.}
	\label{fig:tmpath}
\end{figure}

\begin{figure}[htbp]
	\centering
	\includegraphics[width=0.7\textwidth]{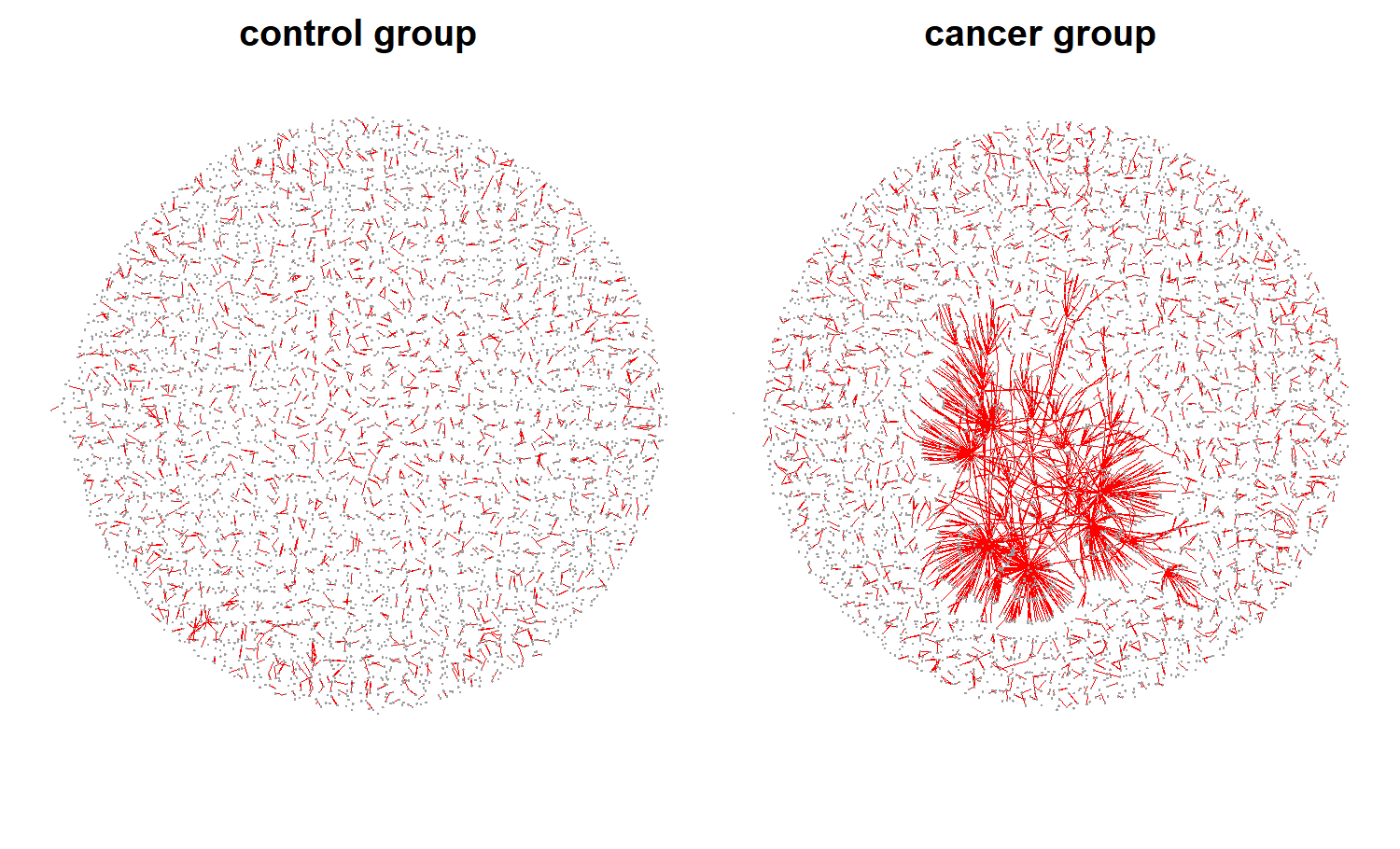}
	\caption{The estimated graphs chosen by five-fold cross-validation generated by MARS with the prostrate data set.}
	\label{fig:pcsiginv}
\end{figure}

Next, we will test the performance of MARS and glasso on the breast cancer data set. We follow the same assumption stated in \citep{Cai2011} that this gene measurements data are normally distributed with $N(\mu_k, \Sigma), \, k = 1,\,2$, where $\Sigma$ is the same for RD group and pCR group, but the means are different. Some two-sample t-tests are performed with given p-value tolerances, which are set to $0.005$, $0.01$, $0.05$, $0.1$, and $1$, to obtain the most significant genes (with a smaller p-value). Under those set p-values, the numbers of chosen genes are $1228$, $1646$, $3640$, $5418$, and $22283$ respectively. Note that, the last one contains all the genes with nearly $250$ million parameters. We point out that, the $\lambda$ paths for all the tests, except the test with p-value tolerance $0.05$, are set from $\lambda_{\textup{min}}$ to $1$ by $0.01$, where $\lambda_{\textup{min}}$ is decided by some pre-tests with the D-trace estimator. When the p-value tolerance is set to $0.05$, if the gap between two subsequent regularization parameters in the path is $0.01$, glasso will fail due to insufficient memory, so we set the $\lambda$ gap for this test to $0.02$. The regularization parameters for each test are chosen by five-fold cross-validation, and the total computation times are concluded in Table \ref{tab:bc}. { The stopping tolerance for MARS and glasso is set to $10^{-4}$ to ensure that all the relative KKT residuals are less than $10^{-4}$.} The estimated graphs obtained by MARS and glasso with p-value tolerance $0.005$, $0.01$, and $0.01$ can be found in Figure \ref{fig:bcsum}. From this figure, we notice that the graphs obtained by MARS and glasso are similar to each other, but the times taken by MARS are obviously less than those taken by glasso. Especially when the p-value tolerance is $0.05$, the total computation time of glasso is more than 20 times that of MARS. Besides, Figure \ref{fig:bcsum} also shows the estimated graphs obtained by MARS when the p-value tolerances are 0.1 and 1, but the figure for the latter one only plots the connections among the most significant $5418$ genes.

\begin{table}[htbp]
	\centering
	\caption{Test results of MARS and glasso on the breast cancer data sets with different p-value tolerances.}
	\scriptsize
	\begin{threeparttable}
		\begin{tabular}{c|c|c|c|c}
			\hline
			\multirow{2}[4]{*}{p-value tolerance} & \multirow{2}[4]{*}{No. of genes} & \multicolumn{2}{c|}{time (mins) including cross-validation} & \multirow{2}[4]{*}{No. of $\lambda$} \bigstrut\\
			\cline{3-4}          &       & MARS  & glasso &  \bigstrut\\
			\hline
			0.005 & 1228  & 20.26  & 71.51  & 63 \bigstrut[t]\\
			0.01  & 1646  & 23.06  & 159.79  & 60 \\
			0.05  & 3640  & 58.32  & 1257.81  & 28 \\
			0.1   & 5418  & 150.54  & $-$   & 54 \\
			1     & 22283 & 553.35  & $-$   & 29 \bigstrut[b]\\
			\hline
		\end{tabular}%
		\begin{tablenotes}
			\item The symbol ``$-$" indicates out of memory.
		\end{tablenotes}
	\end{threeparttable}
	\label{tab:bc}%
\end{table}%

\begin{figure}[htbp]
	\centering
	\includegraphics[width=0.95\textwidth]{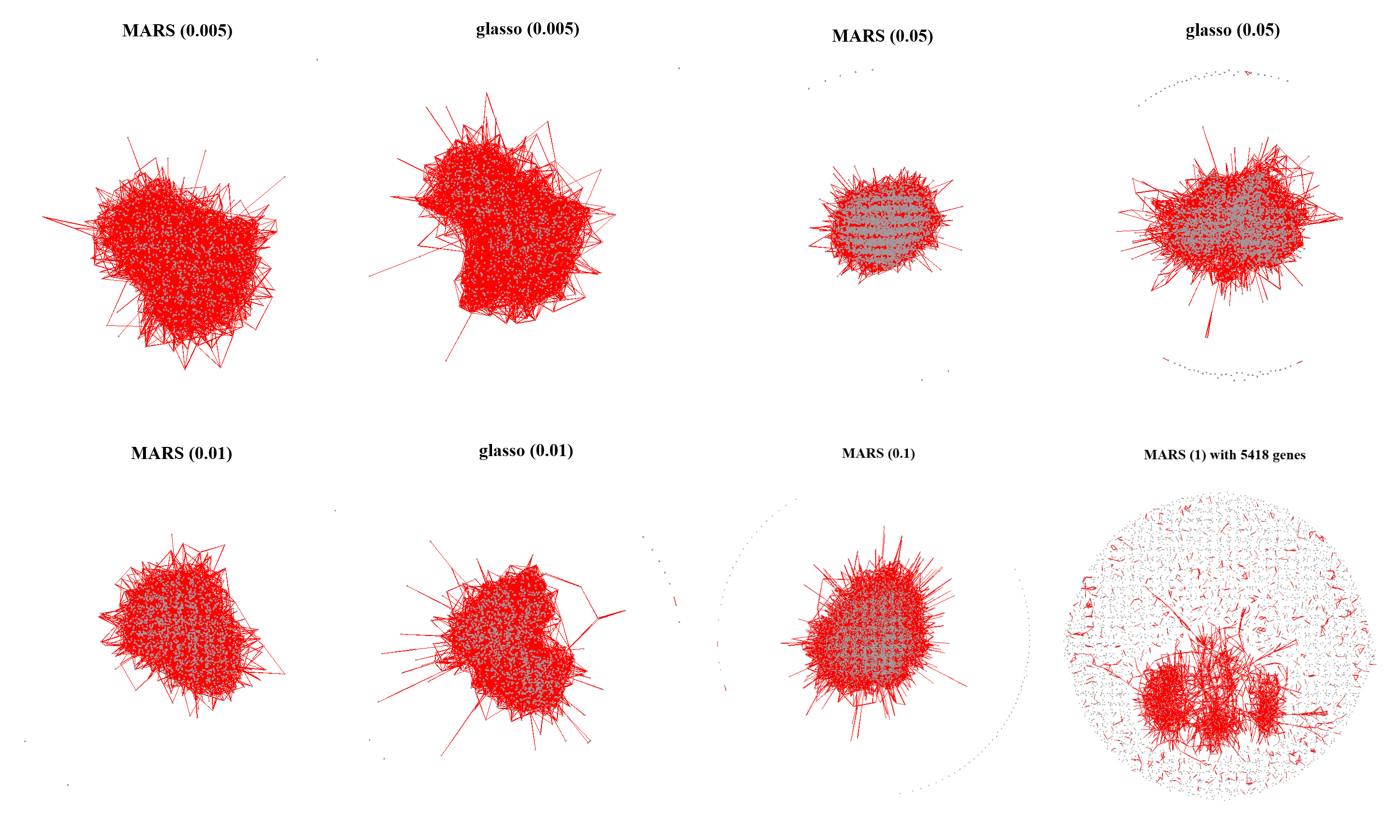}
	\caption{The estimated graphs for the breast cancer data set chosen by five-fold cross-validation with using MARS and glasso under different p-value tolerances.}
	\label{fig:bcsum}
\end{figure}

\section{Conclusions}\label{sec:col}
In this paper, we have derived an efficient second-order algorithm for high-dimensional sparse precision matrices estimation under the $\ell_1$-penalized D-trace loss. By using a dual approach and adopting an adaptive sieving reduction strategy, our algorithm is capable of handling large-scale datasets. The theoretical properties of our algorithm have been well established. In particular, we have shown that our algorithm enjoys a global linear convergence rate and converges asymptotically superlinearly. Numerical results have further convincingly demonstrated the promising performance and high efficiency of our algorithm when compared with other state-of-the-art solvers. For instance, our algorithm can be up to $20$ times faster than glasso for some subsets of a breast cancer dataset with much less storage requirement.

We conclude by pointing out that our algorithm is not only designed for sparse precision matrix estimation but also can be extended to solve other matrix-form problems under a  penalized quadratic loss function. More specifically, our algorithm can be extended to solve problems of the following form:
\begin{equation}\label{ext}
	\mathop{\min}\limits_{\Omega \in \mathbb{S}^{p}}\left\{ \frac{1}{2} {\rm tr}(\Omega \widehat{\Sigma} \Omega^T) - {\rm tr}(\widehat{Q}\Omega) + \textup{pen}_\lambda(\Omega) \right\},
\end{equation}
where $\textup{pen}_\lambda(\Omega)$ is a penalty term to encourage different structures.
The algorithm we derived in this paper can be viewed as a special case with $\widehat{Q}=I_p$, and $\textup{pen}_\lambda(\Omega)=\lambda\left\|\Omega\right\|_{1,{\rm off}}$.
With different choices of $\widehat\Sigma$ and $\widehat{Q}$, the quadratic loss \eqref{ext} outputs sparse solutions for different statistical analysis such as canonical vectors for Fisher's LDA \citep{GAYNANOVA2016}, sparse canonical correlation analysis, and sparse sliced inverse regression \citep{Tan2018}.
This is left for future work.

\section*{Acknowledgments}

The authors would like to thank the action editor and the three anonymous referees for their constructive comments and suggestions to improve the quality of this paper.  Thanks also go to Dr. Yancheng Yuan at  The Hong Kong Polytechnic University  for his many helpful suggestions on the revision of this paper.


\appendix
\section{Performance of EQUAL}
\label{app:stopEQUAL}

This is the test performance of EQUAL, in terms of the objective value and relative KKT residual ($\eta$), for generating solution paths on the prostrate datasets under different stopping tolerances, which are set to be $10^{-5}$ and $10^{-6}$ respectively.
\begin{table}[htbp]
	\centering
	\caption{Test performance of the estimated solution paths with different stopping tolerances by using EQUAL on the prostate data set.}
	\tiny
		\resizebox{\textwidth}{!}{
	\begin{tabular}{c|cc|cc|cc|cc}
		\hline
		\multirow{3}[5]{*}{$\lambda$} & \multicolumn{4}{c|}{control group} & \multicolumn{4}{c}{cancer group} \bigstrut\\
		\cline{2-9}          & \multicolumn{2}{c|}{EQUAL (1e-6)} & \multicolumn{2}{c|}{EQAUL (1e-5)} & \multicolumn{2}{c|}{EQUAL (1e-6)} & \multicolumn{2}{c}{EQAUL (1e-5)} \bigstrut\\
		\cline{2-9}          & objective value & $\eta$  & objective value & $\eta$  & objective value & $\eta$  & objective value & $\eta$ \bigstrut\\
		\hline
		0.99  & 924.68 & 5.05e-02 & 276912.70 & 5.16e-01 & -1964.80 & 3.40e-02 & 264907.50 & 5.03e-01 \\
		0.98  & -861.70 & 3.89e-01 & 21190.85 & 1.30e-01 & -1966.85 & 3.41e-02 & 263387.60 & 5.06e-01 \\
		0.97  & -1798.68 & 3.54e-02 & 20554.88 & 1.31e-01 & -2448.14 & 6.39e-02 & 18529.18 & 1.36e-01 \\
		0.96  & -1802.12 & 3.56e-02 & 9137.18 & 9.49e-02 & -2447.90 & 6.40e-02 & 7894.90 & 8.70e-02 \\
		0.95  & -2351.89 & 7.22e-02 & 9518.00 & 9.85e-02 & -2456.80 & 6.31e-02 & 8194.18 & 8.98e-02 \\
		0.94  & -2356.68 & 7.18e-02 & 9696.31 & 1.01e-01 & -2700.02 & 2.20e-02 & 8304.20 & 9.18e-02 \\
		0.93  & -2371.45 & 7.04e-02 & 9670.20 & 1.03e-01 & -2698.98 & 2.22e-02 & 8225.00 & 9.29e-02 \\
		0.92  & -2645.23 & 2.33e-02 & 9444.32 & 1.04e-01 & -2704.31 & 2.22e-02 & 7962.94 & 9.30e-02 \\
		0.91  & -2643.28 & 2.35e-02 & 9029.55 & 1.03e-01 & -2858.82 & 2.50e-02 & 7530.24 & 9.22e-02 \\
		0.90  & -2648.98 & 2.35e-02 & 8442.32 & 1.01e-01 & -2858.34 & 2.54e-02 & 6943.96 & 9.04e-02 \\
		0.89  & -2832.42 & 2.70e-02 & 2736.37 & 2.65e-01 & -2860.17 & 2.55e-02 & 2018.53 & 3.05e-01 \\
		0.88  & -2833.94 & 2.73e-02 & 2927.30 & 2.55e-01 & -2863.98 & 2.56e-02 & 2186.41 & 2.92e-01 \\
		0.87  & -2837.75 & 2.74e-02 & 3034.11 & 2.49e-01 & -2931.86 & 1.40e-02 & 2280.82 & 2.84e-01 \\
		0.86  & -2912.43 & 1.49e-02 & 3051.66 & 2.47e-01 & -2939.44 & 1.39e-02 & 2295.27 & 2.81e-01 \\
		0.85  & -2918.38 & 1.49e-02 & 2977.98 & 2.49e-01 & -2995.49 & 1.05e-02 & 2226.93 & 2.82e-01 \\
		0.84  & -2989.64 & 1.12e-02 & 2814.58 & 2.53e-01 & -3002.93 & 1.08e-02 & 2076.93 & 2.88e-01 \\
		0.83  & -2993.46 & 1.20e-02 & 2566.14 & 2.62e-01 & -3008.80 & 1.14e-02 & 1850.30 & 2.98e-01 \\
		0.82  & -2997.82 & 1.30e-02 & 2240.63 & 2.76e-01 & -3020.96 & 9.59e-03 & 1554.29 & 3.14e-01 \\
		0.81  & -3014.42 & 1.00e-02 & 1848.97 & 2.96e-01 & -3040.97 & 7.75e-03 & 1198.97 & 3.37e-01 \\
		0.80  & -3057.88 & 5.86e-03 & 1404.33 & 3.24e-01 & -3067.80 & 5.64e-03 & 796.57 & 3.67e-01 \\
		0.79  & -3069.02 & 6.06e-03 & -480.09 & 5.09e-02 & -3081.18 & 5.79e-03 & -942.61 & 4.70e-02 \\
		0.78  & -3066.00 & 1.06e-02 & -236.76 & 5.41e-02 & -3084.28 & 8.97e-03 & -714.84 & 4.97e-02 \\
		0.77  & -3085.58 & 7.64e-03 & -41.10 & 5.71e-02 & -3102.08 & 7.10e-03 & -528.46 & 5.22e-02 \\
		0.76  & -3106.15 & 5.13e-03 & 89.13 & 5.97e-02 & -3124.53 & 4.94e-03 & -401.82 & 5.43e-02 \\
		0.75  & -3121.64 & 4.35e-03 & 141.05 & 6.16e-02 & -3144.88 & 3.75e-03 & -345.45 & 5.59e-02 \\
		0.74  & -3139.48 & 4.26e-03 & 110.25 & 6.29e-02 & -3164.75 & 3.75e-03 & -363.44 & 5.68e-02 \\
		0.73  & -3144.13 & 8.05e-03 & -2.31 & 6.32e-02 & -3172.27 & 7.52e-03 & -455.55 & 5.70e-02 \\
		0.72  & -3162.36 & 6.12e-03 & -191.14 & 6.27e-02 & -3190.12 & 6.33e-03 & -616.71 & 5.65e-02 \\
		0.71  & -3188.84 & 4.08e-03 & -2316.43 & 5.47e-02 & -3220.13 & 3.79e-03 & -837.45 & 5.51e-02 \\
		0.70  & -3204.06 & 6.15e-03 & -2231.19 & 6.00e-02 & -3244.18 & 3.82e-03 & -2394.92 & 5.10e-02 \\
		0.69  & -3229.79 & 3.71e-03 & -2140.55 & 6.55e-02 & -3271.68 & 3.70e-03 & -2338.51 & 5.46e-02 \\
		0.68  & -3257.15 & 3.59e-03 & -2048.69 & 7.11e-02 & -3285.98 & 7.49e-03 & -2279.14 & 5.85e-02 \\
		0.67  & -3289.17 & 3.31e-03 & -1958.72 & 7.68e-02 & -3310.62 & 6.34e-03 & -2220.84 & 6.23e-02 \\
		0.66  & -3310.93 & 5.40e-03 & -1873.05 & 8.23e-02 & -3349.71 & 3.58e-03 & -2166.39 & 6.62e-02 \\
		0.65  & -3330.89 & 6.44e-03 & -1793.58 & 8.77e-02 & -3370.41 & 7.03e-03 & -2117.65 & 7.00e-02 \\
		0.64  & -3370.90 & 3.53e-03 & -1721.91 & 9.27e-02 & -3424.12 & 2.13e-03 & -2075.66 & 7.36e-02 \\
		0.63  & -3399.63 & 5.77e-03 & -1659.54 & 9.72e-02 & -3445.53 & 7.45e-03 & -2040.94 & 7.71e-02 \\
		0.62  & -3429.19 & 5.79e-03 & -1607.92 & 1.01e-01 & -3475.93 & 6.86e-03 & -2013.80 & 8.04e-02 \bigstrut[b]\\
		\hline
	\end{tabular}%
}
	\label{tab:comstop}%
\end{table}%

\section{Supplementary tables}\label{app:comperf}

{ This section is a supplement to the numerical experiments in Section \ref{sec:results}. 
Specifically, two tables list the KKT residuals and one table provides the statistical performance of the generated solutions by MARS and glasso corresponding to the tests performed in Section \ref{sec:numexpcom}.
}


\begin{table}[htbp]
	\centering
	\caption{Average relative KKT residuals ($\eta$) of different algorithms with 10 regularization parameters and 10 replications.}
	\tiny
	\resizebox{\textwidth}{!}{
		\begin{tabular}{clccccc}
			\toprule
			&       & Model 1 & Model 2 & Model 3 & Model 4 & Model 5 \\
			&       & mean $|$ sd & mean $|$ sd & mean $|$ sd & mean $|$ sd & mean $|$ sd \\
			\midrule
			\multicolumn{7}{l}{Models 1 to 4: p $=$ 2000;   Model 5: p $=$ 2025} \\
			\midrule
			\multirow{8}[2]{*}{\begin{sideways}n $=$ 50\end{sideways}} & MARS  & 3.56E-05 $|$ 8.84E-06 & 3.85E-05 $|$ 6.73E-06 & 3.31E-05 $|$ 6.88E-06 & 3.73E-05 $|$ 6.66E-06 & 3.61E-05 $|$ 4.87E-06 \\
			& SSNAL & 6.87E-05 $|$ 7.16E-06 & 5.74E-05 $|$ 7.61E-06 & 6.79E-05 $|$ 6.13E-06 & 6.12E-05 $|$ 8.03E-06 & 6.38E-05 $|$ 7.44E-06 \\
			& iADMM & 9.12E-05 $|$ 3.29E-06 & 9.31E-05 $|$ 2.44E-06 & 9.05E-05 $|$ 3.67E-06 & 9.48E-05 $|$ 3.11E-06 & 9.08E-05 $|$ 2.33E-06 \\
			& eADMM & 9.48E-05 $|$ 3.42E-06 & 9.37E-05 $|$ 3.52E-06 & 9.54E-05 $|$ 2.58E-06 & 9.51E-05 $|$ 2.88E-06 & 9.37E-05 $|$ 3.80E-06 \\
			& scio  & 3.94E-03 $|$ 3.66E-04 & 3.50E-02 $|$ 3.20E-02 & 3.96E-03 $|$ 5.93E-04 & 2.79E-01 $|$ 6.61E-02 & 3.50E-03 $|$ 3.27E-04 \\
			& EQUAL & 1.58E-03 $|$ 4.06E-05 & 1.07E-03 $|$ 9.52E-05 & 1.58E-03 $|$ 3.92E-05 & 7.16E-04 $|$ 3.18E-05 & 1.60E-03 $|$ 2.55E-04 \\
			& glasso(1e-3) & 1.71E-05 $|$ 1.28E-06 & 3.70E-05 $|$ 5.87E-06 & 1.58E-05 $|$ 9.21E-07 & 3.59E-05 $|$ 5.39E-06 & 1.54E-05 $|$ 1.44E-06 \\
			& QUIC  & 6.53E-02 $|$ 2.56E-05 & 6.25E-02 $|$ 2.98E-05 & 6.54E-02 $|$ 2.48E-05 & 5.82E-02 $|$ 2.45E-05 & 6.62E-02 $|$ 2.90E-05 \\
			\midrule
			\multirow{8}[2]{*}{\begin{sideways}n $=$ 100\end{sideways}} & MARS  & 3.23E-05 $|$ 5.84E-06 & 3.64E-05 $|$ 4.13E-06 & 2.95E-05 $|$ 3.91E-06 & 4.36E-05 $|$ 7.38E-06 & 3.07E-05 $|$ 5.96E-06 \\
			& SSNAL & 5.36E-05 $|$ 7.29E-06 & 5.99E-05 $|$ 5.69E-06 & 6.04E-05 $|$ 7.03E-06 & 5.84E-05 $|$ 6.39E-06 & 6.42E-05 $|$ 9.51E-06 \\
			& iADMM & 8.51E-05 $|$ 3.81E-06 & 8.99E-05 $|$ 2.30E-06 & 9.00E-05 $|$ 3.38E-06 & 9.15E-05 $|$ 3.59E-06 & 8.78E-05 $|$ 4.70E-06 \\
			& eADMM & 8.02E-05 $|$ 6.43E-06 & 8.42E-05 $|$ 6.99E-06 & 8.14E-05 $|$ 4.46E-06 & 8.84E-05 $|$ 5.95E-06 & 8.37E-05 $|$ 6.41E-06 \\
			& scio  & 1.42E-03 $|$ 9.31E-05 & 3.56E-03 $|$ 1.87E-04 & 1.22E-03 $|$ 8.86E-05 & 9.92E-03 $|$ 3.79E-04 & 1.94E-03 $|$ 1.12E-04 \\
			& EQUAL & 2.11E-03 $|$ 5.80E-04 & 1.49E-03 $|$ 1.72E-04 & 2.61E-03 $|$ 4.12E-05 & 9.88E-04 $|$ 4.77E-05 & 1.44E-03 $|$ 3.83E-04 \\
			& glasso(1e-3) & 1.20E-05 $|$ 1.38E-06 & 1.47E-05 $|$ 5.38E-06 & 9.37E-06 $|$ 8.78E-07 & 1.50E-05 $|$ 5.23E-06 & 2.21E-05 $|$ 3.94E-06 \\
			& QUIC  & 6.53E-02 $|$ 3.54E-05 & 6.14E-02 $|$ 2.93E-05 & 6.54E-02 $|$ 3.69E-05 & 5.76E-02 $|$ 1.91E-05 & 6.66E-02 $|$ 2.87E-05 \\
			\bottomrule
		\end{tabular}%
	}
	\label{tab:kktsmall}%
\end{table}%

\begin{table}[htbp]
	\centering
	\caption{Average relative KKT residuals ($\eta$) of different algorithms with 50 regularization parameters and 10 replications for generating a solution path.}
	\tiny
	\resizebox{\textwidth}{!}{
	\begin{tabular}{clccccc}
		\hline
		&       & Model 1 & Model 2 & Model 3 & Model 4 & Model 5 \bigstrut[t]\\
		&       & mean $|$ sd & mean $|$ sd & mean $|$ sd & mean $|$ sd & mean $|$ sd \bigstrut[b]\\
		\hline
		\multicolumn{7}{l}{Models 1 to 4: p $=$ 2000;   Model 5: p $=$ 2025} \bigstrut\\
		\hline
		\multirow{10}[2]{*}{\begin{sideways}n $=$ 50\end{sideways}} & MARS(1e-4) & 9.75E-06 $|$ 4.72E-06 & 1.08E-05 $|$ 5.23E-06 & 1.02E-05 $|$ 5.49E-06 & 1.23E-05 $|$ 5.83E-06 & 1.15E-05 $|$ 6.31E-06 \bigstrut[t]\\
		& MARS(1e-8) & 1.59E-09 $|$ 2.47E-09 & 1.70E-09 $|$ 2.46E-09 & 1.13E-09 $|$ 1.82E-09 & 2.39E-09 $|$ 2.96E-09 & 1.23E-09 $|$ 1.72E-09 \\
		& SSNAL & 2.60E-05 $|$ 7.15E-06 & 2.72E-05 $|$ 5.50E-06 & 2.65E-05 $|$ 6.38E-06 & 3.00E-05 $|$ 6.30E-06 & 2.91E-05 $|$ 5.91E-06 \\
		& iADMM & 5.13E-05 $|$ 6.41E-06 & 5.06E-05 $|$ 5.91E-06 & 5.16E-05 $|$ 5.69E-06 & 5.56E-05 $|$ 5.54E-06 & 8.47E-05 $|$ 5.84E-06 \\
		& eADMM & 7.62E-05 $|$ 5.67E-06 & 7.33E-05 $|$ 5.06E-06 & 7.66E-05 $|$ 5.38E-06 & 7.88E-05 $|$ 4.26E-06 & 6.46E-05 $|$ 2.25E-05 \\
		& scio  & 7.05E-04 $|$ 4.70E-05 & 9.81E-03 $|$ 6.74E-03 & 6.87E-04 $|$ 6.92E-05 & 4.18E-02 $|$ 9.85E-03 & 9.57E-04 $|$ 1.67E-04 \\
		& EQUAL & 4.17E-03 $|$ 1.42E-04 & 4.10E-03 $|$ 7.55E-05 & 4.16E-03 $|$ 1.22E-04 & 4.05E-03 $|$ 8.07E-05 & 5.01E-03 $|$ 9.77E-05 \\
		& glasso(1e-3) & 3.36E-06 $|$ 1.75E-07 & 6.57E-06 $|$ 1.89E-06 & 3.21E-06 $|$ 2.70E-07 & 5.34E-06 $|$ 9.74E-07 & 4.47E-06 $|$ 2.70E-07 \\
		& glasso(1e-4) & 2.74E-09 $|$ 1.19E-09 & 1.12E-08 $|$ 1.79E-09 & 2.27E-09 $|$ 7.05E-10 & 1.88E-08 $|$ 2.98E-09 & 4.56E-09 $|$ 1.63E-09 \\
		& QUIC  & 1.01E-01 $|$ 6.16E-05 & 9.93E-02 $|$ 4.22E-05 & 1.01E-01 $|$ 4.85E-05 & 9.73E-02 $|$ 3.06E-05 & 1.00E-01 $|$ 4.05E-05 \bigstrut[b]\\
		\hline
		\multirow{10}[2]{*}{\begin{sideways}n $=$ 100\end{sideways}} & MARS(1e-4) & 7.76E-06 $|$ 3.14E-06 & 8.10E-06 $|$ 4.74E-06 & 7.12E-06 $|$ 3.57E-06 & 7.73E-06 $|$ 3.03E-06 & 8.78E-06 $|$ 3.69E-06 \bigstrut[t]\\
		& MARS(1e-8) & 1.13E-09 $|$ 1.75E-09 & 2.55E-09 $|$ 3.19E-09 & 1.38E-09 $|$ 2.11E-09 & 2.69E-09 $|$ 3.22E-09 & 1.55E-09 $|$ 2.58E-09 \\
		& SSNAL & 1.93E-05 $|$ 4.41E-06 & 1.80E-05 $|$ 6.76E-06 & 1.80E-05 $|$ 4.04E-06 & 1.82E-05 $|$ 3.95E-06 & 2.31E-05 $|$ 4.56E-06 \\
		& iADMM & 5.20E-05 $|$ 6.74E-06 & 4.86E-05 $|$ 5.97E-06 & 6.15E-05 $|$ 6.13E-06 & 6.16E-05 $|$ 5.25E-06 & 6.87E-05 $|$ 3.65E-06 \\
		& eADMM & 7.51E-05 $|$ 4.61E-06 & 7.89E-05 $|$ 6.68E-06 & 8.52E-05 $|$ 7.39E-06 & 8.13E-05 $|$ 5.90E-06 & 8.28E-05 $|$ 9.00E-06 \\
		& scio  & 2.14E-04 $|$ 1.60E-05 & 5.49E-04 $|$ 1.72E-05 & 1.91E-04 $|$ 7.85E-06 & 1.50E-03 $|$ 5.89E-05 & 3.05E-04 $|$ 1.28E-05 \\
		& EQUAL & 2.16E-03 $|$ 1.07E-04 & 2.13E-03 $|$ 8.15E-05 & 1.71E-03 $|$ 2.97E-04 & 1.91E-03 $|$ 3.55E-04 & 3.84E-03 $|$ 7.89E-05 \\
		& glasso(1e-3) & 1.69E-06 $|$ 8.35E-07 & 1.37E-06 $|$ 4.56E-07 & 1.34E-06 $|$ 8.73E-07 & 1.41E-06 $|$ 7.12E-07 & 3.21E-06 $|$ 1.51E-06 \\
		& glasso(1e-4) & 1.41E-09 $|$ 3.90E-10 & 9.10E-09 $|$ 1.04E-09 & 1.19E-09 $|$ 5.89E-10 & 7.09E-09 $|$ 6.08E-09 & 6.12E-09 $|$ 2.99E-09 \\
		& QUIC  & 1.25E-01 $|$ 8.06E-05 & 1.23E-01 $|$ 6.32E-05 & 1.26E-01 $|$ 4.94E-05 & 1.22E-01 $|$ 4.32E-05 & 1.25E-01 $|$ 6.87E-05 \bigstrut[b]\\
		\hline
		\multicolumn{7}{l}{Models 1 to 4: p $=$ 3000;   Model 5: p $=$ 3024} \bigstrut\\
		\hline
		\multirow{6}[2]{*}{\begin{sideways}n = 50\end{sideways}} & MARS(1e-4) & 1.17E-05 $|$ 6.69E-06 & 1.09E-05 $|$ 5.15E-06 & 1.10E-05 $|$ 5.14E-06 & 1.06E-05 $|$ 5.17E-06 & 1.16E-05 $|$ 6.28E-06 \bigstrut[t]\\
		& MARS(1e-8) & 8.22E-10 $|$ 1.15E-09 & 1.68E-09 $|$ 2.17E-09 & 1.04E-09 $|$ 1.35E-09 & 2.15E-09 $|$ 2.69E-09 & 1.12E-09 $|$ 1.70E-09 \\
		& SSNAL & 2.87E-05 $|$ 6.39E-06 & 2.74E-05 $|$ 5.60E-06 & 2.69E-05 $|$ 5.58E-06 & 2.72E-05 $|$ 4.88E-06 & 2.99E-05 $|$ 6.50E-06 \\
		& EQUAL & 1.73E-02 $|$ 1.71E-04 & 1.72E-02 $|$ 1.41E-04 & 1.74E-02 $|$ 1.77E-04 & 1.72E-02 $|$ 1.26E-04 & 7.76E-03 $|$ 7.05E-05 \\
		& glasso(1e-3) & 1.94E-06 $|$ 2.13E-07 & 7.27E-06 $|$ 3.20E-07 & 1.92E-06 $|$ 1.70E-07 & 6.66E-06 $|$ 1.99E-06 & 2.60E-06 $|$ 1.88E-07 \\
		& glasso(1e-4) & 1.10E-09 $|$ 4.73E-10 & 1.04E-08 $|$ 1.44E-09 & 1.09E-09 $|$ 4.02E-10 & 1.98E-08 $|$ 1.70E-09 & 1.75E-09 $|$ 5.53E-10 \bigstrut[b]\\
		\hline
		\multirow{6}[2]{*}{\begin{sideways}n = 100\end{sideways}} & MARS(1e-4) & 7.34E-06 $|$ 4.30E-06 & 8.18E-06 $|$ 3.32E-06 & 6.24E-06 $|$ 2.54E-06 & 8.00E-06 $|$ 3.79E-06 & 8.62E-06 $|$ 4.22E-06 \bigstrut[t]\\
		& MARS(1e-8) & 1.30E-09 $|$ 1.80E-09 & 2.09E-09 $|$ 2.20E-09 & 1.67E-09 $|$ 1.97E-09 & 2.51E-09 $|$ 2.65E-09 & 1.44E-09 $|$ 1.79E-09 \\
		& SSNAL & 1.86E-05 $|$ 5.24E-06 & 2.00E-05 $|$ 4.64E-06 & 1.64E-05 $|$ 3.74E-06 & 1.91E-05 $|$ 4.03E-06 & 2.30E-05 $|$ 5.13E-06 \\
		& EQUAL & 1.00E-02 $|$ 7.08E-05 & 9.99E-03 $|$ 1.45E-04 & 1.00E-02 $|$ 5.45E-05 & 1.00E-02 $|$ 5.95E-05 & 5.02E-03 $|$ 5.31E-05 \\
		& glasso(1e-3) & 1.83E-06 $|$ 1.44E-07 & 2.66E-06 $|$ 1.55E-07 & 1.62E-06 $|$ 1.04E-07 & 1.85E-06 $|$ 1.64E-07 & 3.44E-06 $|$ 3.85E-07 \\
		& glasso(1e-4) & 9.11E-10 $|$ 2.32E-10 & 6.45E-09 $|$ 9.30E-10 & 1.00E-09 $|$ 5.32E-10 & 4.19E-09 $|$ 8.91E-10 & 3.38E-09 $|$ 1.41E-09 \bigstrut[b]\\
		\hline
		\multicolumn{7}{l}{Models 1 to 4: p $=$ 5000;   Model 5: p $=$ 5041} \bigstrut\\
		\hline
		\multirow{3}[2]{*}{\begin{sideways}n = 50\end{sideways}} & MARS(1e-4) & 2.56E-05 $|$ 2.66E-05 & 2.62E-05 $|$ 2.73E-05 & 2.53E-05 $|$ 2.74E-05 & 2.41E-05 $|$ 2.45E-05 & 2.48E-05 $|$ 2.64E-05 \bigstrut[t]\\
		& MARS(1e-8) & 2.58E-09 $|$ 5.28E-09 & 3.14E-09 $|$ 5.73E-09 & 3.07E-09 $|$ 5.53E-09 & 4.51E-09 $|$ 6.81E-09 & 2.69E-09 $|$ 5.05E-09 \\
		& SSNAL & 6.35E-05 $|$ 2.34E-05 & 6.20E-05 $|$ 2.49E-05 & 6.18E-05 $|$ 2.35E-05 & 6.40E-05 $|$ 2.39E-05 & 6.54E-05 $|$ 2.44E-05 \bigstrut[b]\\
		\cline{2-7}    \multirow{3}[2]{*}{\begin{sideways}n = 100\end{sideways}} & MARS(1e-4) & 2.58E-05 $|$ 2.54E-05 & 2.77E-05 $|$ 2.84E-05 & 2.25E-05 $|$ 2.15E-05 & 2.35E-05 $|$ 2.25E-05 & 2.38E-05 $|$ 2.64E-05 \bigstrut[t]\\
		& MARS(1e-8) & 3.86E-09 $|$ 7.62E-09 & 5.54E-09 $|$ 9.14E-09 & 5.34E-09 $|$ 8.51E-09 & 6.39E-09 $|$ 1.04E-08 & 2.38E-05 $|$ 2.64E-05 \\
		& SSNAL & 5.90E-05 $|$ 2.38E-05 & 6.34E-05 $|$ 2.39E-05 & 6.04E-05 $|$ 2.26E-05 & 6.16E-05 $|$ 2.52E-05 & 2.38E-05 $|$ 2.64E-05 \bigstrut[b]\\
		\hline
	\end{tabular}%
	}
	\label{tab:pathkkt50}%
\end{table}%

\begin{table}[htbp]
	\centering
	\caption{Average performance of MARS and glasso for precision matrix estimation with 10 replications, p = 2000 in Models 1 to 4, p = 2025 in Model 5 and n = 100 for all the Models.}
	\scriptsize
	\resizebox{\textwidth}{!}{
		\begin{threeparttable}
			\centering
			\begin{tabular}{cccccccccccc}
				\toprule
				&      & \multicolumn{2}{c}{Frobenius} & \multicolumn{2}{c}{Spectral} & \multicolumn{2}{c}{Infinity} & \multicolumn{2}{c}{TP} & \multicolumn{2}{c}{TN} \\
				\cmidrule{3-12}          & $\lambda$       & MARS  & glasso & MARS  & glasso & MARS  & glasso & MARS  & glasso & MARS  & glasso \\
				\midrule
				\multirow{10}[2]{*}{\begin{sideways}Model 1\end{sideways}} & 0.39  & 17.8652  & 17.9603  & 0.8009  & 0.8273  & 1.0085  & 1.3293  & 0.2078  & 0.2735  & 0.9999  & 0.9976  \\
				& 0.38  & 17.8623  & 17.8964  & 0.8010  & 0.8217  & 1.0369  & 1.2258  & 0.2108  & 0.2601  & 0.9999  & 0.9983  \\
				& 0.37  & 17.8601  & 17.8616  & 0.8010  & 0.8175  & 1.0672  & 1.1529  & 0.2142  & 0.2485  & 0.9998  & 0.9988  \\
				& 0.36  & 17.8608  & 17.8455  & 0.8009  & 0.8138  & 1.1053  & 1.0927  & 0.2189  & 0.2388  & 0.9998  & 0.9992  \\
				& 0.35  & 17.8668  & 17.8404  & 0.8005  & 0.8107  & 1.1545  & 1.0463  & 0.2241  & 0.2306  & 0.9996  & 0.9994  \\
				& 0.34  & 17.8843  & 17.8418  & 0.7995  & 0.8083  & 1.2146  & 1.0111  & 0.2305  & 0.2244  & 0.9995  & 0.9996  \\
				& 0.33  & 17.9217  & 17.8468  & 0.7978  & 0.8064  & 1.2929  & 0.9811  & 0.2386  & 0.2191  & 0.9992  & 0.9998  \\
				& 0.32  & 17.9934  & 17.8527  & 0.7948  & 0.8049  & 1.3887  & 0.9562  & 0.2484  & 0.2143  & 0.9988  & 0.9998  \\
				& 0.31  & 18.1249  & 17.8585  & 0.7901  & 0.8039  & 1.5140  & 0.9355  & 0.2600  & 0.2109  & 0.9983  & 0.9999  \\
				& 0.30  & 18.3653  & 17.8636  & 0.7863  & 0.8030  & 1.6736  & 0.9189  & 0.2729  & 0.2079  & 0.9975  & 0.9999  \\
				\midrule
				\multirow{10}[2]{*}{\begin{sideways}Model 2\end{sideways}} & 0.37  & 25.2683  & 25.5038  & 1.5978  & 1.6069  & 1.8957  & 2.4224  & 0.1200  & 0.1787  & 0.9998  & 0.9952  \\
				& 0.36  & 25.2679  & 25.3890  & 1.5965  & 1.6059  & 1.9380  & 2.2680  & 0.1225  & 0.1674  & 0.9998  & 0.9966  \\
				& 0.35  & 25.2712  & 25.3194  & 1.5944  & 1.6049  & 1.9831  & 2.1440  & 0.1259  & 0.1573  & 0.9996  & 0.9976  \\
				& 0.34  & 25.2816  & 25.2801  & 1.5914  & 1.6041  & 2.0465  & 2.0522  & 0.1299  & 0.1483  & 0.9995  & 0.9983  \\
				& 0.33  & 25.3057  & 25.2597  & 1.5870  & 1.6035  & 2.1240  & 1.9759  & 0.1345  & 0.1407  & 0.9992  & 0.9988  \\
				& 0.32  & 25.3540  & 25.2510  & 1.5807  & 1.6028  & 2.2233  & 1.9155  & 0.1405  & 0.1346  & 0.9988  & 0.9992  \\
				& 0.31  & 25.4439  & 25.2497  & 1.5718  & 1.6022  & 2.3426  & 1.8659  & 0.1482  & 0.1300  & 0.9983  & 0.9995  \\
				& 0.30  & 25.6103  & 25.2527  & 1.5588  & 1.6018  & 2.4956  & 1.8288  & 0.1569  & 0.1260  & 0.9976  & 0.9996  \\
				& 0.29  & 25.9087  & 25.2573  & 1.5410  & 1.6014  & 2.7017  & 1.8008  & 0.1675  & 0.1225  & 0.9965  & 0.9998  \\
				& 0.28  & 26.4441  & 25.2624  & 1.5159  & 1.6011  & 2.9820  & 1.7759  & 0.1801  & 0.1201  & 0.9950  & 0.9998  \\
				\midrule
				\multirow{10}[2]{*}{\begin{sideways}Model 3\end{sideways}} & 0.39  & 17.8914  & 18.1370  & 0.8136  & 0.8689  & 1.0042  & 1.3512  & 0.2020  & 0.2286  & 0.9999  & 0.9976  \\
				& 0.38  & 17.8951  & 18.0416  & 0.8152  & 0.8585  & 1.0412  & 1.2528  & 0.2028  & 0.2227  & 0.9999  & 0.9983  \\
				& 0.37  & 17.9014  & 17.9802  & 0.8168  & 0.8501  & 1.0859  & 1.1779  & 0.2039  & 0.2174  & 0.9999  & 0.9988  \\
				& 0.36  & 17.9123  & 17.9416  & 0.8185  & 0.8430  & 1.1359  & 1.1219  & 0.2055  & 0.2138  & 0.9998  & 0.9992  \\
				& 0.35  & 17.9309  & 17.9182  & 0.8206  & 0.8372  & 1.1937  & 1.0725  & 0.2074  & 0.2103  & 0.9996  & 0.9995  \\
				& 0.34  & 17.9622  & 17.9041  & 0.8226  & 0.8322  & 1.2611  & 1.0296  & 0.2103  & 0.2074  & 0.9995  & 0.9996  \\
				& 0.33  & 18.0146  & 17.8959  & 0.8246  & 0.8282  & 1.3430  & 0.9951  & 0.2139  & 0.2055  & 0.9992  & 0.9998  \\
				& 0.32  & 18.1027  & 17.8914  & 0.8263  & 0.8249  & 1.4367  & 0.9653  & 0.2174  & 0.2039  & 0.9988  & 0.9999  \\
				& 0.31  & 18.2506  & 17.8890  & 0.8272  & 0.8222  & 1.5589  & 0.9383  & 0.2226  & 0.2028  & 0.9983  & 0.9999  \\
				& 0.30  & 18.5006  & 17.8879  & 0.8305  & 0.8197  & 1.7281  & 0.9165  & 0.2286  & 0.2020  & 0.9976  & 0.9999  \\
				\midrule
				\multirow{10}[2]{*}{\begin{sideways}Model 4\end{sideways}} & 0.35  & 12.8997  & 14.4898  & 0.5041  & 0.6159  & 0.8595  & 0.9917  & 0.0017  & 0.0106  & 0.9997  & 0.9912  \\
				& 0.34  & 12.9286  & 13.9230  & 0.5066  & 0.5903  & 0.9222  & 1.0178  & 0.0019  & 0.0083  & 0.9995  & 0.9935  \\
				& 0.33  & 12.9863  & 13.5258  & 0.5201  & 0.5719  & 1.0010  & 1.0512  & 0.0022  & 0.0064  & 0.9992  & 0.9952  \\
				& 0.32  & 13.0933  & 13.2578  & 0.5495  & 0.5570  & 1.0949  & 1.0881  & 0.0026  & 0.0050  & 0.9988  & 0.9966  \\
				& 0.31  & 13.2840  & 13.0846  & 0.5895  & 0.5449  & 1.2165  & 1.1292  & 0.0032  & 0.0040  & 0.9983  & 0.9976  \\
				& 0.30  & 13.6198  & 12.9775  & 0.6428  & 0.5353  & 1.3899  & 1.1789  & 0.0040  & 0.0032  & 0.9976  & 0.9983  \\
				& 0.29  & 14.1996  & 12.9152  & 0.7129  & 0.5277  & 1.6303  & 1.2332  & 0.0051  & 0.0026  & 0.9964  & 0.9988  \\
				& 0.28  & 15.1780  & 12.8824  & 0.8060  & 0.5218  & 1.9698  & 1.3015  & 0.0067  & 0.0022  & 0.9950  & 0.9992  \\
				& 0.27  & 16.8252  & 12.8679  & 0.9323  & 0.5171  & 2.4021  & 1.3992  & 0.0088  & 0.0019  & 0.9930  & 0.9995  \\
				& 0.26  & 19.5567  & 12.8643  & 1.1073  & 0.5135  & 2.9720  & 1.5332  & 0.0117  & 0.0017  & 0.9902  & 0.9996  \\
				\midrule
				\multirow{10}[2]{*}{\begin{sideways}Model 5\end{sideways}} & 0.40  & 17.6002  & 16.7998  & 0.7941  & 0.7680  & 1.8083  & 1.2191  & 0.2586  & 0.4523  & 1.0000  & 0.9982  \\
				& 0.39  & 17.5471  & 16.8997  & 0.7932  & 0.7689  & 1.5643  & 1.1496  & 0.2712  & 0.4226  & 0.9999  & 0.9988  \\
				& 0.38  & 17.4861  & 17.0110  & 0.7924  & 0.7711  & 1.3577  & 1.0987  & 0.2870  & 0.3954  & 0.9999  & 0.9992  \\
				& 0.37  & 17.4176  & 17.1229  & 0.7919  & 0.7739  & 1.1903  & 1.0605  & 0.3044  & 0.3686  & 0.9998  & 0.9994  \\
				& 0.36  & 17.3458  & 17.2298  & 0.7922  & 0.7771  & 1.0519  & 1.0304  & 0.3217  & 0.3460  & 0.9998  & 0.9996  \\
				& 0.35  & 17.2755  & 17.3286  & 0.7936  & 0.7804  & 0.9448  & 1.0010  & 0.3424  & 0.3242  & 0.9996  & 0.9998  \\
				& 0.34  & 17.2150  & 17.4170  & 0.7969  & 0.7835  & 0.8668  & 0.9723  & 0.3652  & 0.3059  & 0.9994  & 0.9998  \\
				& 0.33  & 17.1790  & 17.4945  & 0.8038  & 0.7864  & 0.8042  & 0.9476  & 0.3906  & 0.2883  & 0.9992  & 0.9999  \\
				& 0.32  & 17.1918  & 17.5599  & 0.8165  & 0.7889  & 0.7541  & 0.9244  & 0.4162  & 0.2722  & 0.9988  & 0.9999  \\
				& 0.31  & 17.2962  & 17.6145  & 0.8374  & 0.7910  & 0.7137  & 0.9056  & 0.4435  & 0.2594  & 0.9982  & 1.0000  \\
				\bottomrule
			\end{tabular}%
			\begin{tablenotes}
				{
					\item[] The regularization parameters are the same as in the experiments in Section \ref{sec:numexpcom}.
				}
			\end{tablenotes}
		\end{threeparttable}
	}
	\label{tab:errors}%
\end{table}%

\vskip 0.2in
\bibliography{MARS_revised}

\begin{thebibliography}{40}
\providecommand{\natexlab}[1]{#1}
\providecommand{\url}[1]{\texttt{#1}}
\expandafter\ifx\csname urlstyle\endcsname\relax
  \providecommand{\doi}[1]{doi: #1}\else
  \providecommand{\doi}{doi: \begingroup \urlstyle{rm}\Url}\fi

\bibitem[Banerjee et~al.(2008)Banerjee, El~Ghaoui, and
  d'Aspremont]{Banerjee2008}
Onureena Banerjee, Laurent El~Ghaoui, and Alexandre d'Aspremont.
\newblock Model selection through sparse maximum likelihood estimation for
  multivariate gaussian or binary data.
\newblock \emph{Journal of Machine Learning Research}, 9:\penalty0 485--516,
  2008.

\bibitem[Cai et~al.(2011)Cai, Liu, and Luo]{Cai2011}
Tony Cai, Weidong Liu, and Xi~Luo.
\newblock A constrained $\ell_1$ minimization approach to sparse precision
  matrix estimation.
\newblock \emph{Journal of the American Statistical Association}, 106\penalty0
  (494):\penalty0 594--607, 2011.

\bibitem[Clarke(1983)]{Clarke1990}
Frank~H. Clarke.
\newblock \emph{Optimization and Nonsmooth Analysis}.
\newblock John Wiley and Sons, 1983.

\bibitem[Du et~al.(2020)Du, Du, Huang, Wang, and He]{Du2020}
Changde Du, Changying Du, Lijie Huang, Haibao Wang, and Huiguang He.
\newblock Structured neural decoding with multitask transfer learning of deep
  neural network representations.
\newblock \emph{IEEE Transactions on Neural Networks and Learning Systems},
  2020.

\bibitem[Du(2015)]{Du2015}
Mengyu Du.
\newblock \emph{An inexact alternating direction method of multipliers for
  convex composite conic programming with nonlinear constraints}.
\newblock PhD thesis, Department of Mathematics, National University of
  Singapore, Singapore, 2015.

\bibitem[Fischer(1997)]{fischer1997solution}
Andreas Fischer.
\newblock Solution of monotone complementarity problems with locally
  lipschitzian functions.
\newblock \emph{Mathematical Programming}, 76\penalty0 (3):\penalty0 513--532,
  1997.

\bibitem[Friedman et~al.(2008)Friedman, Hastie, and Tibshirani]{Friedman2008}
Jerome Friedman, Trevor Hastie, and Robert Tibshirani.
\newblock Sparse inverse covariance estimation with the graphical lasso.
\newblock \emph{Biostatistics}, 9\penalty0 (3):\penalty0 432--441, 2008.

\bibitem[Gaynanova et~al.(2016)Gaynanova, Booth, and Wells]{GAYNANOVA2016}
Irina Gaynanova, James~G. Booth, and Martin~T. Wells.
\newblock Simultaneous sparse estimation of canonical vectors in the $p \gg {
  N}$ setting.
\newblock \emph{Journal of the American Statistical Association}, 111\penalty0
  (514):\penalty0 696--706, 2016.

\bibitem[G{\"u}ler(1991)]{guler1991convergence}
Osman G{\"u}ler.
\newblock On the convergence of the proximal point algorithm for convex
  minimization.
\newblock \emph{SIAM Journal on Control and Optimization}, 29\penalty0
  (2):\penalty0 403--419, 1991.

\bibitem[Hess et~al.(2006)Hess, Anderson, Symmans, Valero, Ibrahim, Mejia,
  Booser, Theriault, Buzdar, Dempsey, et~al.]{Hess2006}
Kenneth~R. Hess, Keith Anderson, W.~Fraser Symmans, Vicente Valero, Nuhad
  Ibrahim, Jaime~A. Mejia, Daniel Booser, Richard~L. Theriault, Aman~U. Buzdar,
  Peter~J. Dempsey, et~al.
\newblock Pharmacogenomic predictor of sensitivity to preoperative chemotherapy
  with paclitaxel and fluorouracil, doxorubicin, and cyclophosphamide in breast
  cancer.
\newblock \emph{Journal of Clinical Oncology}, 24\penalty0 (26):\penalty0
  4236--4244, 2006.

\bibitem[Hsieh et~al.(2014)Hsieh, Sustik, Dhillon, and Ravikumar]{Hsieh2014}
Cho-Jui Hsieh, M{\'a}ty{\'a}s~A Sustik, Inderjit~S. Dhillon, and Pradeep
  Ravikumar.
\newblock Quic: quadratic approximation for sparse inverse covariance
  estimation.
\newblock \emph{Journal of Machine Learning Research}, 15\penalty0
  (1):\penalty0 2911--2947, 2014.

\bibitem[Lauritzen(1996)]{Lauritzen1996}
Steffen~L. Lauritzen.
\newblock \emph{Graphical Models}, volume~17.
\newblock Clarendon Press, 1996.

\bibitem[Lemar{\'e}chal and Sagastiz{\'a}bal(1997)]{Lemarechal1997}
Claude Lemar{\'e}chal and Claudia Sagastiz{\'a}bal.
\newblock Practical aspects of the moreau--yosida regularization: Theoretical
  preliminaries.
\newblock \emph{SIAM Journal on Optimization}, 7\penalty0 (2):\penalty0
  367--385, 1997.

\bibitem[Li and Gui(2006)]{LiH2006}
Hongzhe Li and Jiang Gui.
\newblock Gradient directed regularization for sparse gaussian concentration
  graphs, with applications to inference of genetic networks.
\newblock \emph{Biostatistics}, 7\penalty0 (2):\penalty0 302--317, 2006.

\bibitem[Li(2009)]{LiSZ2009}
Stan~Z. Li.
\newblock \emph{Markov Random Field Modeling in Image Analysis}.
\newblock Springer Science \& Business Media, 2009.

\bibitem[Li et~al.(2018)Li, Sun, and Toh]{Li2018}
Xudong Li, Defeng Sun, and Kim-Chuan Toh.
\newblock A highly efficient semismooth newton augmented lagrangian method for
  solving lasso problems.
\newblock \emph{SIAM Journal on Optimization}, 28\penalty0 (1):\penalty0
  433--458, 2018.

\bibitem[Li et~al.(2019)Li, Hu, Zhang, Yu, and Zhang]{Li2019respre}
Yang Li, Jun Hu, Chengxin Zhang, Dong-Jun Yu, and Yang Zhang.
\newblock Respre: high-accuracy protein contact prediction by coupling
  precision matrix with deep residual neural networks.
\newblock \emph{Bioinformatics}, 35\penalty0 (22):\penalty0 4647--4655, 2019.

\bibitem[Lin et~al.(2020)Lin, Yuan, Sun, and Toh]{Lin2020}
Meixia Lin, Yancheng Yuan, Defeng Sun, and Kim-Chuan Toh.
\newblock Adaptive sieving with ppdna: Generating solution paths of exclusive
  lasso models.
\newblock \emph{arXiv preprint arXiv:2009.08719}, 2020.

\bibitem[Liu and Luo(2015)]{Liu2015}
Weidong Liu and Xi~Luo.
\newblock Fast and adaptive sparse precision matrix estimation in high
  dimensions.
\newblock \emph{Journal of Multivariate Analysis}, 135:\penalty0 153--162,
  2015.

\bibitem[Meinshausen(2007)]{Meinshausen2007}
Nicolai Meinshausen.
\newblock Relaxed lasso.
\newblock \emph{Computational Statistics $\&$ Data Analysis}, 52\penalty0
  (1):\penalty0 374--393, 2007.

\bibitem[Meinshausen and B{\"u}hlmann(2006)]{Meinshausen2006}
Nicolai Meinshausen and Peter B{\"u}hlmann.
\newblock High-dimensional graphs and variable selection with the lasso.
\newblock \emph{Annals of Statistics}, 34\penalty0 (3):\penalty0 1436--1462,
  2006.

\bibitem[Moreau(1962)]{moreau1962fonctions}
Jean~J. Moreau.
\newblock Fonctions convexes duales et points proximaux dans un espace
  hilbertien.
\newblock \emph{Comptes rendus hebdomadaires des s{\'e}ances de l'Acad{\'e}mie
  des sciences}, 255:\penalty0 2897--2899, 1962.

\bibitem[Oztoprak et~al.(2012)Oztoprak, Nocedal, Rennie, and
  Olsen]{Oztoprak2012}
Figen Oztoprak, Jorge Nocedal, Steven Rennie, and Peder~A. Olsen.
\newblock Newton-like methods for sparse inverse covariance estimation.
\newblock \emph{Advances in Neural Information Processing Systems},
  25:\penalty0 755--763, 2012.

\bibitem[Ravikumar et~al.(2011)Ravikumar, Wainwright, Raskutti, and
  Yu]{Ravikumar2011}
Pradeep Ravikumar, Martin~J. Wainwright, Garvesh Raskutti, and Bin Yu.
\newblock High-dimensional covariance estimation by minimizing
  $\ell_1$-penalized log-determinant divergence.
\newblock \emph{Electronic Journal of Statistics}, 5:\penalty0 935--980, 2011.

\bibitem[Robinson(1981)]{Robinson1981}
Stephen~M. Robinson.
\newblock Some continuity properties of polyhedral multifunctions.
\newblock In \emph{Mathematical Programming at Oberwolfach}, pages 206--214.
  Springer, 1981.

\bibitem[Rockafellar(1970)]{Rockafellar1970}
R.~Tyrrell Rockafellar.
\newblock \emph{Convex Analysis}.
\newblock Princeton University Press, 1970.

\bibitem[Rockafellar(1976)]{Rockafellar1976}
R.~Tyrrell Rockafellar.
\newblock Augmented lagrangians and applications of the proximal point
  algorithm in convex programming.
\newblock \emph{Mathematics of Operations Research}, 1\penalty0 (2):\penalty0
  97--116, 1976.

\bibitem[Rockafellar and Wets(2009)]{rockafellar2009variational}
R.~Tyrrell Rockafellar and Roger J-B Wets.
\newblock \emph{Variational Analysis}, volume 317.
\newblock Springer Science \& Business Media, 2009.

\bibitem[Scheinberg et~al.(2010)Scheinberg, Ma, and Goldfarb]{Scheinberg2010}
Katya Scheinberg, Shiqian Ma, and Donald Goldfarb.
\newblock Sparse inverse covariance selection via alternating linearization
  methods.
\newblock \emph{Advances in Neural Information Processing Systems}, 23, 2010.

\bibitem[Shewchuk(1994)]{Shewchuk1994introduction}
Jonathan~Richard Shewchuk.
\newblock An introduction to the conjugate gradient method without the
  agonizing pain, 1994.

\bibitem[Sun(1986)]{Sun1986}
Jie Sun.
\newblock \emph{On monotropic piecewise quadratic programming}.
\newblock PhD thesis, Department of Mathematics, University of Washington,
  1986.

\bibitem[Tan et~al.(2018)Tan, Wang, Liu, and Zhang]{Tan2018}
Kean~Ming Tan, Zhaoran Wang, Han Liu, and Tong Zhang.
\newblock Sparse generalized eigenvalue problem: Optimal statistical rates via
  truncated rayleigh flow.
\newblock \emph{Journal of the Royal Statistical Society: Series B (Statistical
  Methodology)}, 80\penalty0 (5):\penalty0 1057--1086, 2018.

\bibitem[Wang and Jiang(2020)]{Wang2020}
Cheng Wang and Binyan Jiang.
\newblock An efficient admm algorithm for high dimensional precision matrix
  estimation via penalized quadratic loss.
\newblock \emph{Computational Statistics $\&$ Data Analysis}, 142:\penalty0
  106812, 2020.

\bibitem[Wille et~al.(2004)Wille, Zimmermann, Vranov{\'a}, F{\"u}rholz, Laule,
  Bleuler, Hennig, Preli{\'c}, von Rohr, Thiele, et~al.]{Wille2004}
Anja Wille, Philip Zimmermann, Eva Vranov{\'a}, Andreas F{\"u}rholz, Oliver
  Laule, Stefan Bleuler, Lars Hennig, Amela Preli{\'c}, Peter von Rohr, Lothar
  Thiele, et~al.
\newblock Sparse graphical gaussian modeling of the isoprenoid gene network in
  arabidopsis thaliana.
\newblock \emph{Genome Biology}, 5\penalty0 (11):\penalty0 1--13, 2004.

\bibitem[Witten et~al.(2011)Witten, Friedman, and Simon]{Witten2011}
Daniela~M. Witten, Jerome~H. Friedman, and Noah Simon.
\newblock New insights and faster computations for the graphical lasso.
\newblock \emph{Journal of Computational and Graphical Statistics}, 20\penalty0
  (4):\penalty0 892--900, 2011.

\bibitem[Yuan(2010)]{Yuan2010}
Ming Yuan.
\newblock High dimensional inverse covariance matrix estimation via linear
  programming.
\newblock \emph{Journal of Machine Learning Research}, 11:\penalty0 2261--2286,
  2010.

\bibitem[Yuan and Lin(2007)]{Yuan2007}
Ming Yuan and Yi~Lin.
\newblock Model selection and estimation in the gaussian graphical model.
\newblock \emph{Biometrika}, 94\penalty0 (1):\penalty0 19--35, 2007.

\bibitem[Zhang and Zou(2014)]{Zhang2014}
Teng Zhang and Hui Zou.
\newblock Sparse precision matrix estimation via lasso penalized d-trace loss.
\newblock \emph{Biometrika}, 101\penalty0 (1):\penalty0 103--120, 2014.

\bibitem[Zhang et~al.(2020)Zhang, Zhang, Sun, and Toh]{zhang2020efficient}
Yangjing Zhang, Ning Zhang, Defeng Sun, and Kim-Chuan Toh.
\newblock An efficient hessian based algorithm for solving large-scale sparse
  group lasso problems.
\newblock \emph{Mathematical Programming}, 179\penalty0 (1):\penalty0 223--263,
  2020.

\bibitem[Zhao et~al.(2010)Zhao, Sun, and Toh]{Zhao2010}
Xinyuan Zhao, Defeng Sun, and Kim-Chuan Toh.
\newblock A newton-cg augmented lagrangian method for semidefinite programming.
\newblock \emph{SIAM Journal on Optimization}, 20\penalty0 (4):\penalty0
  1737--1765, 2010.

\end{thebibliography}

\end{document}